\newif\iflong
\newif\ifshort

\newif\ifincidence
\incidencetrue

\newif\ifcdnf

\iflong
\newcommand{\shortversion}[1]{}
\newcommand{\longversion}[1]{#1}
\else
\shorttrue
\newcommand{\shortversion}[1]{#1}
\newcommand{\longversion}[1]{}
\fi

\documentclass[conference]{IEEEtran}
\IEEEoverridecommandlockouts
\usepackage[margin=1in]{geometry}

\usepackage{amsthm,amsopn,amsmath}
\usepackage{amssymb}

\usepackage{lineno}

\usepackage{booktabs}   %
\usepackage{subcaption} %

\usepackage{enumitem,todonotes}
\usepackage{mathtools}
\DeclarePairedDelimiter\ceil{\lceil}{\rceil}
\DeclarePairedDelimiter\floor{\lfloor}{\rfloor}

\usepackage{tikz_includes}

\newcommand{\problemFont}[1]{\textsc{#1}}

  \newcommand{\bigoh}{\mathcal{O}}

  \newcommand{\ctypes}{\mathcal{T}}
\newcommand{\CbyT}{\mathcal{C}}

\newcommand{\trash}[1]{} 

\newlength\problemlength
\newcommand\dproblem[3]{%
\begin{center}
\fbox{%
\begin{minipage}{.93	\linewidth}%
\begin{list}{}{\labelwidth\problemlength \labelsep.7em \rightmargin1.5em
\leftmargin\problemlength \advance\leftmargin by3em%
\parsep0ex \itemsep.2ex plus.1ex}
\item[{\sl Problem:\hfill}] {\problemFont{#1}}
\item[{\sl Input:  \hfill}] #2
\item[{\sl Output: \hfill}] #3
\end{list}
\end{minipage}
}
\end{center}
}
\newcommand\eproblem[3]{%
\begin{center}
\fbox{%
\begin{minipage}{.93	\linewidth}%
\begin{list}{}{\labelwidth\problemlength \labelsep.7em \rightmargin1.5em
\leftmargin\problemlength \advance\leftmargin by3em%
\parsep0ex \itemsep.2ex plus.1ex}
\item[{\sl Problem:\hfill}] {\problemFont{#1}}
\item[{\sl Input:  \hfill}] #2
\item[{\sl Question: \hfill}] #3
\end{list}
\end{minipage}
}
\end{center}
}

\usepackage{microtype}
\usepackage{multirow}
\usepackage{xcolor}

\usepackage[breaklinks=true]{hyperref}
\newcommand{\FIXCAM}[1]{#1} %

\usepackage[normalem]{ulem}

\newcommand{\Tau}{\mathcal{T}}

\usepackage{relsize}
\usepackage{xspace}

\usepackage[ruled,vlined,linesnumbered]{algorithm2e}

\SetKwInput{KwData}{In}
\SetKwInput{KwResult}{Out}
\SetAlFnt{\small}
\SetAlCapFnt{\small}
\SetAlCapNameFnt{\small}
\SetAlCapHSkip{0pt}
\SetEndCharOfAlgoLine{}
\IncMargin{-0.4em}
\makeatletter
\newcommand{\algorithmfootnote}[2][\footnotesize]{
  \let\old@algocf@finish\@algocf@finish
  \def\@algocf@finish{\old@algocf@finish
    \leavevmode\rlap{\begin{minipage}{\linewidth}
    #1#2
    \end{minipage}}
  }
}
\usepackage{algorithmic}

\newtheorem{LEM}{Lemma}[section]

\newtheorem{PROP}[LEM]{Proposition}

\newtheorem{THM}[LEM]{Theorem}
\newtheorem{COR}[LEM]{Corollary}

\newtheorem{observation}[LEM]{Observation}
\newtheorem{OBS}[LEM]{Observation}

\newtheorem{DEF}[LEM]{Definition}

\newtheorem{EX}[LEM]{Example}
\newtheorem{EXa}[LEM]{Example}

\renewenvironment{EX}{\begin{EXa}}{\hfill\ensuremath{\blacksquare}\end{EXa}}
\newenvironment{restatetheorem}[1][\unskip]{%
  \begingroup
  \def\theLEM{\ref{#1}} 
}%
{%
  \addtocounter{LEM}{-1}
  \endgroup
}%

\newcommand{\complexityClassFont}[1]{\ensuremath{\mathrm{#1}}}
\newcommand{\Ptime}{\complexityClassFont{P}}
\newcommand{\PSPACE}{\ensuremath{\textsc{PSpace}}\xspace}
\newcommand{\SAT}{\textsc{Sat}\xspace}

\newcommand{\QBFSAT}{\textsc{QSat}\xspace}
\newcommand{\QBFSATCNF}{\ensuremath{\textsc{CQSat}}\xspace}
\newcommand{\QBFCNF}{\textsc{CQBF}\xspace}

\newcommand{\SIGMA}[2]{\ensuremath{\Sigma_{{#1}}^{{#2}}}}
\newcommand{\PI}[2]{\ensuremath{\Pi_{{#1}}^{{#2}}}}
\newcommand{\Q}{\ensuremath{Q}}
\newcommand{\SB}{\{}%
\newcommand{\SM}{\mid}%
\newcommand{\SE}{\}}%

\newcommand{\smallalign}[1]{#1
  \setlength{\abovedisplayskip}{3pt}
  \setlength{\belowdisplayskip}{\abovedisplayskip}
  \setlength{\abovedisplayshortskip}{0pt}
  \setlength{\belowdisplayshortskip}{3pt}}

\DeclareMathOperator{\dom}{\mathsf{dom}}%
\DeclareMathOperator{\sgn}{\mathsf{sgn}}%

\newcommand{\NAT}{\ensuremath{\mathbb{N}}}
\newcommand{\Nat}{\mathbb{N}} %

\newcommand{\TTT}{\ensuremath{\mathcal{T}}}%

\newcommand{\tw}[1]{\mathsf{tw}(#1)}
\newcommand{\td}[1]{\mathsf{td}(#1)}
\newcommand{\tdx}{\mathsf{td}}
\newcommand{\pd}[1]{\mathsf{td}(#1)}
\newcommand{\pw}[1]{\mathsf{pw}(#1)}

\newcommand{\lbv}[1]{\mathit{VarIdx}}
\newcommand{\lbvsl}[1]{\mathit{VarPIdxs}}
\newcommand{\lbvs}[1]{\mathit{VarIdxs}}
\newcommand{\lbvv}[1]{\mathit{VarVal}}
\newcommand{\lbvvsl}[1]{\mathit{VarPVals}}
\newcommand{\lbvvs}[1]{\mathit{VarVals}}

\newcommand{\lsat}{\mathit{VarSat}}
\newcommand{\lsel}{\mathit{VarPSels}}

\DeclareMathOperator{\tower}{\ensuremath{\mathsf{tow}}}
\DeclareMathOperator{\poly}{\ensuremath{{poly}}}

\newcommand{\Card}[1]{\left|#1\right|}
\newcommand{\CCard}[1]{\|#1\|}

\newcommand{\lit}[2]{\ensuremath{\mathsf{lit}(#1, #2)}}

\newcommand{\bval}[2]{\ensuremath{[\![#1]\!]_{#2}}}

\newcommand{\loc}{\mathbb{P}}
\newcommand{\bvali}[2]{\ensuremath{[\![#1]\!]_{#2}^\loc}}

\DeclareMathOperator{\var}{var}
\DeclareMathOperator{\vare}{var^{\exists}}
\DeclareMathOperator{\varu}{var^{\forall}}

\DeclareMathOperator{\matr}{\mathsf{matr}}
\newcommand{\bigO}{\ensuremath{{\mathcal O}}}

\newcommand{\eqdef}{\ensuremath{\,\mathrel{\mathop:}=}}

\DeclareMathOperator{\width}{\mathsf{width}}
\DeclareMathOperator{\children}{\mathsf{cld}}

\newcommand{\qbfind}[2]{#1(#2)}
\newcommand{\qbfformula}{\Q}
\newcommand{\booleanformulaA}{\ensuremath{F}}
\newcommand{\booleanformulaB}{\ensuremath{F'}}

\makeatletter
\newcommand{\pushright}[1]{\ifmeasuring@#1\else\omit\hfill$\displaystyle#1$\fi\ignorespaces}
\newcommand{\pushleft}[1]{\ifmeasuring@#1\else\omit$\displaystyle#1$\hfill\fi\ignorespaces}
\makeatother

\title{
Structure-Aware Lower Bounds and Broadening the Horizon of Tractability for QBF}
\author{%
  \IEEEauthorblockN{%
    Johannes K. Fichte\IEEEauthorrefmark{1}, %
    Robert Ganian\IEEEauthorrefmark{2}, %
    Markus Hecher\IEEEauthorrefmark{3}, %
    Friedrich Slivovsky\IEEEauthorrefmark{2}, and %
    Sebastian Ordyniak\IEEEauthorrefmark{4}
  }
  \IEEEauthorblockA{\IEEEauthorrefmark{1}%
    Department of Computer and Information Science (IDA), 
    Link\"oping University, SE-581 83, Linköping, Sweden\\
  }
  \IEEEauthorblockA{\IEEEauthorrefmark{2}%
    Institute of Logic and Computation, TU Wien, Favoritenstra\ss{}e
    9--11, 1040 Wien,
    Austria\\
  }
  \IEEEauthorblockA{\IEEEauthorrefmark{3}%
    Massachusetts Institute of Technology, Cambridge MA 02139, USA
  }
  \IEEEauthorblockA{\IEEEauthorrefmark{4}%
    School of Computing, University of Leeds, Leeds, LS2 9JT, United
    Kingdom
  }
  \thanks{Authors are ordered alphabetically. 
    The work has been carried out while Fichte \& Hecher
    visited the Simons Institute.
    It is supported by Austrian Science Fund (FWF) grants J4656, P32830 and Y1329, Society for Research Funding Lower Austria (GFF) grant ExzF-0004, Vienna Science and Technology Fund (WWTF) grants ICT19-060 and ICT19-065, and 
   the ELLIIT funded by the Swedish government.
  }
}

\begin{document}
%
%
%
%
%
%

\makeatletter

\def\emptycleardoublepage{\clearpage\if@twoside \ifodd\c@page\else
\thispagestyle{empty}%
\hbox{}\newpage\if@twocolumn\hbox{}\newpage\fi\fi\fi}

\makeatother

\newpage

%
%
%
%
%
%
%
%
%
%
%
%
%
%
%
%
%
%
%
%

%
%
%
%
%
%
%
%
%
%
%
\newcommand{\DCNF}{CDNF\xspace}

\maketitle

\begin{abstract}

%
  The \QBFSAT problem, which asks to evaluate a quantified Boolean
  formula (QBF), is of fundamental interest in approximation, counting,
  decision, and probabilistic complexity and is also considered the prototypical \PSPACE-complete problem.
%
     As such, it has previously been studied under
  various structural restrictions (parameters), most notably
  parameterizations of the primal graph representation of
  instances. Indeed, it is known that \QBFSAT\ remains
  \PSPACE-complete even when restricted to instances with constant
  treewidth of the primal graph, but the problem admits a double-exponential fixed-parameter algorithm parameterized by the vertex cover number (primal graph).

%

However, prior works have left a gap in our understanding of the
complexity of \QBFSAT\ when viewed from the perspective of other
natural representations of instances, most notably via incidence
graphs. In this paper, we develop structure-aware reductions which
allow us to obtain essentially tight lower bounds for highly restricted instances
of \QBFSAT, including instances whose incidence graphs have bounded
treedepth or feedback vertex number. We complement these lower bounds with novel algorithms for \QBFSAT which establish a nearly-complete picture of the problem's complexity under standard graph-theoretic parameterizations. We also show implications for other natural graph representations, and obtain novel upper as well as lower bounds for \QBFSAT\ under more fine-grained parameterizations of the primal graph.
\end{abstract}

\section{Introduction}
\ifincidence
The evaluation problem for quantified Boolean formulas (\QBFSAT) is a
natural generalization of the Boolean satisfiability problem (\SAT)
and among the most important problems in theoretical computer science,
with applications in symbolic
reasoning~\cite{BenedettiM08,BloemEKKL14,EglyEiterTW00,OtwellRemshagenTruemper04,Rintanen99,ShuklaBPS19},
constraint satisfaction problems
(CSP)~\cite{GentNightingaleRowley08,Dechter06,Freuder85}, databases,
and logic~\cite{Grohe01}.
Input formulas in \QBFSAT consist of a \emph{(quantifier) prefix} and a \emph{matrix}, which can be an arbitrary Boolean formula but is often assumed to be in \emph{conjunctive normal form (CNF)},~e.g., converted by the classical Tseytin transformation~\cite{Tseytin70}.
\QBFSAT is considered the archetypical representative of \PSPACE-complete problems and has been extensively studied from the perspective of classical approximation~\cite{Vazirani03}, counting~\cite{Toda91}, decision~\cite{Papadimitriou94}, and probabilistic complexity~\cite{Lautemann83}, but also through the lens of 
parameterized complexity~\cite{PanVardi06,AtseriasOliva14a}.

The vast majority of parameterizations studied for \QBFSAT\ rely on a suitable graph representation of the matrix; this is, in fact, similar to the situation for \textsc{Boolean Satisfiability} (\SAT)~\cite{SamerSzeider10b,SamerS21,WallonM20},
\textsc{Integer Linear Programming} (\textsc{ILP})~\cite{GanianO18,ChanCKKP22}, 
\textsc{Constraint Satisfaction} (\textsc{CSP})~\cite{Freuder85,DechterP89,SamerSzeider10a}, and other fundamental problems. For \QBFSAT, the most classical parameterization 
considered in the literature is the treewidth~$k$ of the \emph{primal graph} representation of the formula's matrix in \emph{conjunctive normal form (CNF)}. 
There, the complexity 
is well understood by now: The problem remains \PSPACE-complete when parameterized by this parameter
 alone~\cite{PanVardi06,AtseriasOliva14a,FichteHecherPfandler20} even when restricted to decompositions which are paths,
but can be solved in time $\tower(\ell, k)\cdot\poly(n)$\footnote{%
  The runtime is exponential in the treewidth~$k$, where~$k$ is on
  top of a tower of iterated exponentials of height 
  quantifier depth~$\ell$.}  where~$\ell$ is the quantifier depth of the formula's prefix and~$n$ is the number of variables of the formula.
On a more positive note, parameterizing by the vertex cover number of the primal graph alone is known to yield a fixed-parameter algorithm for \QBFSAT~\cite{LampisMitsou17} that 
is double-exponential. 

The $\ell$-fold exponential gap in terms of parameter dependence between treewidth and vertex cover number raises the following question:
what is the boundary of fixed-parameter tractability when dropping the quantifier depth~$\ell$  as a parameter?
%
In parameterized complexity, there is a whole hierarchy of structural parameters that are more restrictive than treewidth and less restrictive than vertex cover number, most prominently \emph{treedepth}~\cite{sparsity} and the \emph{feedback vertex} and \emph{edge numbers}\footnote{The vertex or edge deletion distances to acyclicity, respectively.}.
However, 
there is an even larger gap: we know very little about the complexity-theoretic landscape of \QBFSAT\ in the context of matrix representations other than the primal graph.
The most prominent example of such a graph representation is the \emph{incidence graph}, which has been extensively studied for \SAT~\cite{SamerSzeider10b,SlivovskySzeider20,SamerS21}, CSP~\cite{SamerSzeider10a,HaanKS15}, and \textsc{ILP}~\cite{EibenGKOPW19}, among others. The aforementioned hardness for \QBFSAT\ carries over from primal treewidth to the treewidth of the incidence graph~\cite{AtseriasOliva14a} and the problem is fixed-parameter tractable using quantifier depth plus treewidth of the incidence graph~\cite{CapelliMengel19}, but no other results for structural parameters of this graph were previously~known.

\vspace{-.5em}
\subsection{Overview of Contributions}
Inspired by the high-level approach used to obtain \QBFSAT\ lower bounds for treewidth~\cite{FichteHecherPfandler20}, 
in Section~\ref{sec:saw} we formalize a notion of \emph{structure-aware} (SAW) reductions for \QBFSAT.
These reductions serve as a tool to precisely demonstrate functional
dependencies between parameters of the input instance and the reduced
instance. 
%
Utilizing this notion of SAW reductions, we establish in Section~\ref{sec:main} tight lower bounds for highly restrictive classes of \QBFSAT\ instances that have profound complexity-theoretic implications for three distinct representations of the matrix. These results essentially rule out efficient algorithms for treedepth and faster algorithms than the one for treewidth when using feedback vertex number. 
We highlight them below, followed by a separate discussion for each of the~representations.
Unless the \emph{Exponential Time Hypothesis (ETH)}~\cite{ImpagliazzoPaturiZane01} fails: 

\begin{enumerate}
\item \QBFSAT\ cannot be solved faster than in $\tower(\ell', k)\cdot\poly(n)$ for~$\ell'$ linear in the quantifier depth~$\ell$, where $k$ is either the feedback vertex number 
or the treedepth of the incidence~graph.
\item \QBFSAT\ cannot be solved faster than in $\tower(\ell', k)\cdot\poly(n)$ for~$\ell'$ linear in the quantifier depth~$\ell$, where $k$ is either the feedback vertex number 
or the treedepth of the primal graph of formulas in combined conjunctive normal form (CNF) and \emph{disjunctive normal form (DNF)}.
\item \QBFSAT\ cannot be solved faster than in $\tower(\ell', k)\cdot\poly(n)$ for~$\ell'$ linear in the quantifier depth~$\ell$, 
where $k$ is either the feedback vertex number 
or the treedepth of the primal graph after deleting a single clause from the~matrix.
\end{enumerate}

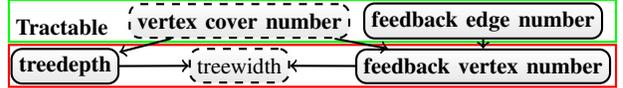
\begin{figure}\centering\vspace{-.5em}
 \begin{tikzpicture}[scale=.6,node distance=0.35mm, every node/.style={scale=0.85}]
\tikzset{every path/.style=thick}
%
\node (vc) [tdnode,dashed,label={[yshift=-0.25em,xshift=0.5em]above left:$ $}] {\textbf{vertex cover number}};
\node (td) [tdnode,label={[yshift=-0.25em,xshift=0.5em]above left:$ $},below left=.5em of vc] {\textbf{treedepth}};
\node (fes) [tdnode,label={[xshift=-1.0em, yshift=-0.15em]above right:$ $}, right = .5em of vc]  {\textbf{feedback edge number}};
\node (fvs) [tdnode,label={[xshift=-1.0em, yshift=-0.15em]above right:$ $}, below = .35em of fes]  {\textbf{feedback vertex number}};
\node (tw) [tdnode,dashed,label={[yshift=-0.25em,xshift=0.5em]above left:$ $},below=.35em of vc] {treewidth};
%
\node [xshift=3.25em,yshift=0em,thick,green,text width=26.35em, text height=1.2em,
draw, align=right] {\textbf{\textcolor{black}{Tractable}\hfill}\vspace{-2em}};
\node [xshift=3.25em,yshift=-2.0em,thick,red,text width=26.35em, text height=1.2em,
draw, align=right] {\qquad\qquad\textbf{\textcolor{black}{}\hfill}\vspace{-2em}};
\draw [->] (vc) to (fvs);
\draw [->] (fes) to (fvs);
\draw [->] (fvs) to (tw);
\draw [->] (vc) to (td);
\draw [->] (td) to (tw);
\end{tikzpicture}\vspace{-.5em}

\caption{Nearly-complete picture for \QBFSAT and parameters on the incidence graph; directed arks indicate that the source parameter upper-bounds the destination, see, e.g.,~\cite{SorgeWeller21,graphclasses}. 
The red frame represents intractability with lower bounds essentially matching
known upper bounds (ETH-tight for treewidth and feedback vertex number); 
the green frame indicates tractability (fpt) results. Bold-face text marks new results.}
\label{fig:params-in}
\end{figure}

\noindent \textbf{1.\ Results for Incidence Graphs.}\quad
Our two complexity-theoretic lower bounds identify that,  with respect to the fundamental representation as incidence graph, the boundaries of intractability for \QBFSAT\ lie significantly below treewidth. They also raise the question of whether we can obtain efficient algorithms for the problem using parameters which place stronger restrictions on the incidence graph. The two by far most natural structural graph parameters satisfying these properties are the mentioned vertex cover number and feedback edge~number. 

We complement our lower bounds with fixed-parameter algorithms for \QBFSAT\ with respect to both parameters, which are provided in Section~\ref{sec:vcfen-inc}. 
Thereby we establish a nearly-complete picture of the problem's complexity based on 
structure of the incidence graph, see~Figure~\ref{fig:params-in}. 

\smallskip
\noindent \textbf{2.\ Implications for Primal Graphs of Combined Matrices.}\quad
Previous complexity-theoretic studies of \QBFSAT\ have typically assumed that the matrix is represented in CNF form, which admits standard graphical representations and may be obtained from an arbitrary formula by using the classical Tseytin transformation~\cite{Tseytin70}. However, empirical evidence has shown that normal form transformations adversely affect the performance of \QBFSAT solvers~\cite{AnsoteguiGS05}, and solvers now typically support more general input formats than CNF~\cite{JordanKlieberSeidl16}.

Given these developments, it is natural to consider the complexity of \QBFSAT\ from the viewpoint of more general normal forms of the matrix which still admit suitable graph representations. An obvious step in this direction would be to combine CNF and DNF,~i.e., consisting of a conjunction of a CNF and a DNF formula. This combined ``\DCNF'' is 
used 
by backtracking search algorithms for QBF, since it is able to emulate 
forms of circuit-level reasoning while enjoying 
optimized data structures~\cite{GoultiaevaSB13}.

Since the \DCNF\ is a strict generalization of the CNF, the lower bounds we established for the incidence graph of the CNF in Section~\ref{sec:main} immediately carry over. However, unlike in the CNF case, our reductions also directly rule out fixed-parameter tractability of \QBFSAT\ with respect to both the treedepth and the feedback vertex number of primal graphs for matrices in \DCNF. 

\smallskip
\noindent \textbf{3.\ Tightening the Gap on Primal Graphs.}\quad
For classical CNF matrices, our SAW reductions of Section~\ref{sec:main} almost---but not quite---settle the aforementioned complexity-theoretic gap between the treewidth and vertex cover number of the primal graph. In particular, we prove that allowing the addition of a single clause to instances with bounded treedepth or feedback vertex number in the primal graph leads to intractability. Given this development, we view settling the parameterized complexity of \QBFSAT\ with respect to these two parameters as the main open questions left in our understanding of the problem's complexity landscape.

As our last contribution, we obtain new fixed-parameter algorithms for \QBFSAT\ with the aim of tightening this gap. First, we obtain a linear kernel (and hence also a fixed-parameter algorithm) for \QBFSAT\ parameterized by the feedback edge number of the primal graph. Second, we establish the fixed-parameter tractability for the problem with respect to several relaxations of the vertex cover number that may be seen as ``stepping stones'' towards treedepth on primal graphs of CNFs.
%
A more elaborate overview of our results is provided in Figure~\ref{fig:saw} (left).

\vspace{-.4em}
\subsection{Approach and Techniques}
For establishing fine-grained lower bounds for parameters between treewidth and vertex cover
on the graph representations above, 
we utilize
the notion of structure-aware (SAW) reductions as visualized in
Figure~\ref{fig:saw} (right). 
%
%
%
We develop concrete SAW reductions that are conceptually self-reductions from \QBFSAT to \QBFSAT,
where we trade an exponential decrease of the parameters feedback vertex number or treedepth for an exponential increase
of runtime dependency on the corresponding parameter.
In order to obtain tight lower bounds that ideally match existing upper bounds (and rule out 
algorithms significantly better than the one for treewidth), one has to carefully
carry out this trade-off so that the  order of magnitude of the runtime dependency increase does not 
exceed the parameter decrease's magnitude.
%
%
More precisely, our transformations reduce from~$\QBFSAT$ using the respective parameter~$k$ and quantifier depth~$\ell$,
to~$\QBFSAT$ when parameterized by~$\log(k)$ with quantifier depth~$\ell{+}1$.
%
%
By iterating this construction (see also Figure~\ref{fig:saw} (right)), we trade an $i$-fold exponential parameter decrease (from~$k$ to~$\log^i(k)$)
for a quantifier depth increase of~$i$, which then, assuming ETH, results in a QBF that is $\ell{+}i$-fold exponential in~$\log^i(k)$ to solve.
%
%

As a consequence of our reductions, we also obtain an interesting result for classical complexity: 
It turns out that \emph{a single additional clause} is already responsible for intractability 
of $\QBFSAT$ on the well-known tractable fragment of 2-CNF formulas.
More specifically, $\QBFSAT$ on 2-CNFs plus one clause with quantifier depth~$\ell>1$ is indeed~$\SIGMA{\ell-1}{P}$-complete, see Corollary~\ref{cor:hardnesslongclause}.

Notably, the construction of our reductions also allows us to strengthen our established lower bounds to graph representations
that are \emph{purely restricted to variables of the innermost quantifier (block)}.
This is a consequence of the fact that our concrete SAW reductions are carried out such that the majority of structural dependencies
reside in the innermost quantifier block of the constructed instance.
Further, the lower bounds even hold for 
parameters covering the vertex deletion distance to (almost)
simple paths, as well as for restricted variants of treedepth.  
Both results are construction-specific consequences, but these findings are in fact significantly stronger 
than the lower bounds for feedback vertex set
and treedepth, thereby providing deeper insights into the hardness of~\QBFSAT.

%

%

Our lower bound results using SAW reductions allow us to draw a rather
comprehensive picture for the (fine-grained) parameterized complexity
of the incidence graph by strengthening previous hardness results to
much more restrictive parameters such as feedback vertex set and
treedepth. We complement these negative results for the incidence
graph by giving fpt-algorithms for \QBFSATCNF, i.e., \QBFSAT
restricted to formulas in CNF, both for vertex cover number and feedback
edge number. Our main technical contribution here is a kernelization
algorithm for feedback edge set for both the primal and incidence graph.

We then turn our attention towards solving \QBFSATCNF using structural
restrictions of the primal graph. While we have to leave open whether
\QBFSATCNF is fixed-parameter tractable parameterized by either the
feedback vertex number or the treedepth of the primal graph, we are able
to make some progress towards establishing fixed-parameter
tractability for the latter. In particular, using novel insights into
winning strategies of Hintikka games, we are able to obtain
fixed-parameter algorithms for three variants of the so-called
\emph{$c$-deletion set} parameter, which is a 
parameter between vertex cover number and treedepth. 

\fi

\newcommand{\ubound}[1]{\ensuremath{\blacktriangle^{#1}}}
\newcommand{\lbound}[1]{\ensuremath{\blacktriangledown^{#1}}}

\newcommand{\ilbound}[1]{\ensuremath{{\blacktriangledown}^{#1}}}

%
\newcommand{\kubound}[1]{\ensuremath{{\triangle}^{#1}}}
\newcommand{\klbound}[1]{\ensuremath{\triangledown^{#1}}}

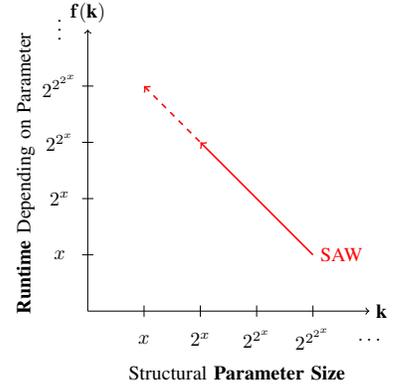
\begin{figure*}
\centering\vspace{-.25em}
\begin{minipage}{0.7\textwidth}\centering
\resizebox{1\textwidth}{!}{%
\begin{tabular}{cccccl}
  \toprule
  Prefix$^A$ & Matrix & Matrix-NF & Graph & Complexity & Ref \\ 
  \midrule
         &  & CNF &   & PSPACE &   \cite{StockmeyerMeyer73} \\
  qd $\ell$& & CNF &  & $\SIGMA{{\ell}}{P}$ / $\PI{{\ell}}{P}$ &  \cite{StockmeyerMeyer73}\\

  qd $\ell$ &  & 2,1-CDNF & & $\SIGMA{{\ell-1}}{P}$ / $\PI{{\ell-1}}{P}$ & Cor~\ref{cor:hardnesslongclause}\\
  \midrule
         & tw & CNF &  P & PSPACE &  \cite{PanVardi06,AtseriasOliva14a} \\
  qd & tw  & CNF &  P, I & \kubound{\dagger} / \klbound{\ddagger} & Prop~\ref{ref:ubo}~\cite{Chen04a},\cite{CapelliMengel19} / \cite{PanVardi06}, Prop~\ref{ref:pwlb}~\cite{FichteHecherPfandler20} \\
  qd & tw & m,1-CDNF$^B$ & P & \ubound{\dagger} / \klbound{\ddagger} 
                          & Thm~\ref{ref:ub} / Prop~\ref{ref:pwlb}~\cite{FichteHecherPfandler20}\\
  \midrule
  qd & fvs  & m,1-CDNF$^B$ & P & \ubound{\dagger} / \lbound{\ddagger} & Cor~\ref{ref:ubs} /~Thm~\ref{lab:lb}\\
  qd & fvs  & CNF & I & \ubound{\dagger}, \lbound{\ddagger} & Cor~\ref{ref:ubs} / Cor~\ref{lab:inc-lb}\\
  qd & td  & m,1-CDNF$^B$ & P &  \ubound{\dagger} / \ilbound{\sim\ddagger} & Cor~\ref{ref:ubs} / Thm~\ref{lab:lb_pdp}, Cor~\ref{lab:lb_tdp}\\
  qd & td  & CNF & I & \ubound{\dagger} / \ilbound{\sim\ddagger} & Cor~\ref{ref:ubs} / Cor~\ref{lab:inc-tdp}\\
  \midrule
             & fes & CNF & P & $\blacksquare$ & Thm~\ref{thm:kernel}\\
             & fes & CNF & I & $\blacksquare$ & Thm~\ref{thm:kernel-inc}\\
         & dels & CNF & P & $\blacksquare^{C}$ & Thm~\ref{the:fpt:cdels}, Thm~\ref{thm:1cdeletionset}, Thm~\ref{thm:qbfcnfunifptex}\\
  \midrule

         & vc & m,1-CDNF$^B$ & P & $\ubound{\dagger_2}$ / $\klbound{\ddagger_2}$ 
                                                     & Thm~\ref{thm:cnfvcn} / Prop~\ref{prop:vco}~\cite{LampisMitsou17}\\
         & vc & CNF & I & $\ubound{\dagger_2}$ 
                                                 & Thm~\ref{thm:incvco} \\ 
  \bottomrule
\end{tabular}
}
\end{minipage}%
\begin{minipage}{.32\textwidth} 
\resizebox{1\textwidth}{!}{
\begin{tikzpicture}
  \draw[->] (0, 0) -- (5, 0) node[right] {\textbf{k}};
  \draw[->] (0, 0) -- (0, 5) node[above] {$\mathbf{f(\textbf{k})}$};
  \node at (1,-0.5) {$x$};
  \node at (2,-0.5) {$2^x$};
  \node at (3,-0.5) {$2^{2^x}$};
  \node at (4,-0.5) {$2^{2^{2^x}}$};
  \node at (5,-0.5) {$\ldots$};
  \node at (2.65,-1.1) {Structural \textbf{Parameter Size}};

  \draw[] (1, 0.1) -- (1, -0.1);
  \draw[] (2, 0.1) -- (2, -0.1);
  \draw[] (3, 0.1) -- (3, -0.1);
  \draw[] (4, 0.1) -- (4, -0.1);
  %
  \node at (-0.5,1) {$x$};
  \node at (-0.5,2) {$2^x$};
  \node at (-0.5,3) {$2^{2^x}$};
  \node   at (-0.5,4) {$2^{2^{2^x}}$};
  \node [rotate=90]  at (-0.5,5) {$\ldots$};
  \node [rotate=90] at (-1.1,2.45) {\textbf{Runtime} Depending on Parameter};
  \draw[] (0.1, 1) -- (-0.1,1);
  \draw[] (0.1, 2) -- (-0.1,2);
  \draw[] (0.1, 3) -- (-0.1,3);
  \draw[] (0.1, 4) -- (-0.1,4);
 %

  %
  \draw[<-,thick,red] (2, 3) -- (4, 1) node[right] {SAW};
  \draw[<-,thick,red,dashed] (1, 4) -- (2, 3) node[right] {};
\end{tikzpicture}
}
\end{minipage}
\vspace{-.5em}
\caption{(Left):
  Runtime bounds for~\QBFSAT on a QBF~$Q$ when parameterized by
  parameters listed in prefix and matrix.  
  The triangles \ubound{} refer to established precise upper bounds
  and \lbound{} to precise lower bounds. By \klbound{} and \kubound{}
  we refer to previously known precise upper and lower bounds.
  The box $\blacksquare$ illustrates new fixed-parameter tractability results. 
  Graph parameters are applied to either the primal (P) or incidence (I) graph; the parameters are: 
  ``qd'' refers to the \emph{quantifier depth}; %
  ``fvs'' indicates the \emph{feedback vertex number}; %
  ``td'' indicates the \emph{treedepth}; %
  ``fes''refers to the \emph{feedback edge number}; and %
  ``dels'' refers to the \emph{size+$c$} for a \emph{$c$-deletion
    set}. %
  %
  %
  The runtime bounds are abbreviated by the marks where $\ell$ refers to the prefix 
  and $k$ to the parameterization of the matrix. Detailed results:
  $^\dagger$:~$\tower(\ell, O(k))\cdot\poly(\Card{\var(Q)})$; 
  $^\ddagger$: $\tower(\ell, o(k))\cdot\poly(\Card{\var(Q)})$; 
  $^{\sim\ddagger}$: $\tower(\ell, o(k-\ell))\cdot\poly(\Card{\var(Q)})$ 
  %
  $^{\dagger_2}$: $2^{2^{\mathcal{O}(k)}}\cdot\poly(\Card{\var(Q)})$, for constant~$m$: $2^{k^{\mathcal{O}(m)}}\cdot\poly(\Card{\var(Q)})$;
  $^{\ddagger_2}$: $2^{2^{o(k)}}\cdot\poly(\Card{\var(Q)})$; 
  %
  %
  and 
  $^A$: $\exists$ odd/$\forall$ even.
  $^B$: the lower bound already holds for 3,1-CDNF; and
  $^C$:
  fpt under restrictions, parameterized by~$|D|{+}c$ for a $c$-deletion set $D$. 
  (Right): Structure-Aware reductions for compensating 
  exponential parameter decrease via runtime dependency.\vspace{-.6em}}\label{fig:saw}
\end{figure*}

%

\section{Preliminaries}

We assume basics from graph theory, cf.~\cite{Diestel12,BondyMurty08}. A graph~$G=(V,E)$ is
a \emph{subgraph} of~$G'{=}(V',E')$ if $V\subseteq V'$, $E\subseteq E'$. A \emph{(connected) component} of a graph is a largest connected subgraph. A graph is 
\emph{acyclic} if no subgraph forms a cycle. 
%
%
For a graph~$G=(V,E)$ and a set~$S\subseteq V$ ($D\subseteq E$) of vertices (edges),
we define the \emph{subtraction graph} obtained from~$G$ 
by~$G-S\eqdef (V\setminus S, \{e\mid e\in E, e \cap S=\emptyset\})$ (by $G-D\eqdef (V, E\setminus D)$).
%
%
Further, the \emph{union} of given graphs~$G_1=(V_1, E_1)$ and~$G_2=(V_2, E_2)$ is given by~$G_1\sqcup G_2\eqdef (V_1\cup V_2, E_1 \cup E_2)$.
%
  Expression $\tower(\ell,k)$ is a tower of exponentials of
  $2$ of height~$\ell$ with $k$ on top. 

\paragraph*{Computational Complexity}  %
We give a brief background on parameterized complexity~
\cite{FlumGrohe06,Niedermeier06}.
%
%
Let $\Sigma$ and $\Sigma'$ be two finite non-empty alphabets.
A \emph{parameterized problem} $L$ is a subset of
$\Sigma^* \times \Nat$ for some finite alphabet $\Sigma$. 
%
  $L$ is \emph{fixed-parameter tractable (fpt)} if there exists a computable
  function $f$ and 
an algorithm
  deciding whether $(\mathcal{I},k)\in L$ in
  \emph{fpt-time} $\bigO(f(k)\poly(\CCard{\mathcal{I}}))$, where $\CCard{\mathcal{I}}$ is the 
  size of~$\mathcal{I}$.
Let $L \subseteq \Sigma^* \times \Nat$ and
$L' \subseteq \Sigma'^*\times \Nat$ be two parameterized problems.
A \emph{polynomial-time parameterized-reduction}~$r$,
\emph{pp-reduction} for short, from $L$ to $L'$ is a many-to-one
reduction from $\Sigma^*\times \Nat$ to $\Sigma'^*\times \Nat$ such
that 
$(\mathcal{I},k) \in L$ if and only if
$r(\mathcal{I},k)=(\mathcal{I}',k')\in L'$ with $k' \leq p(k)$
for a fixed computable function $p: \Nat \rightarrow \Nat$ 
and $r$ is computable in time
$\bigO(
\poly(\CCard{\mathcal{I}}))$. 
%
Parameter values are usually 
computed
based on a structural property~$K$ of the instance, 
called \emph{parameter},~e.g., size of a smallest feedback vertex set,
or treewidth. 
%
%
Usually for algorithms we need 
structural representations instead of parameter values, i.e., a
feedback vertex set or tree decomposition.  Therefore,~we~let $\Gamma$ be a 
finite non-empty alphabet and call
$S \in \Gamma^*$ a \emph{structural representation}~of~$\mathcal{I}$.
Then, a \emph{parameterization~$\kappa$} for parameter~$K$ is a
mapping $\kappa: \Gamma^* \rightarrow \NAT$ computing 
$k=\kappa(S)$ in polynomial~time. 
%
%
%
%
%

\paragraph*{Quantified Boolean Formulas (QBFs)}
\emph{Boolean formulas} 
are defined in the usual
way~\cite{BiereEtAl21,KleineBuningLettman99}; \emph{literals} are variables or their negations. We let the \emph{sign of a literal~$l$} be defined by~$\sgn(l)\eqdef 1$ if~$l$ is a variable and~$\sgn(l)\eqdef 0$ otherwise.  For a
Boolean formula~$F$, we denote by~$\var(F)$ the set of variables
of~$F$. Logical operators~$\rightarrow, \wedge, \vee, \neg$ refer to implication, conjunction, disjunction, and negation, respectively, 
as in the usual meaning.
%
  A \emph{term} or \emph{clause} is a set~$S$ of literals; interpreted as a conjunction or disjunction of literals, respectively. 

  We denote by~$\var(S)$ the set of variables appearing in~$S$; without loss of generality we assume 
$\Card{S}=\Card{\var(S)}$.
  A Boolean formula $F$ is in \emph{conjunctive normal form
    (CNF)} if $F$ is a conjunction of clauses and $F$ is in
  \emph{disjunctive normal form (DNF)} if $F$ is a disjunction of
  terms.
  In both cases, we identify $F$ by its set of clauses or terms, respectively.
  %
  %
  A Boolean formula is in \emph{$d$-CNF} or \emph{$d$-DNF} if each set in~$F$
  consists of at most $d$ many literals.
  %
%
%
  Let $\ell\geq 0$ be integer. A \emph{quantified Boolean
    formula}~$\Q$ 
is of the form
  $\mathcal{Q}. F$ for \emph{prefix}~$\mathcal{Q}=Q_{1} V_1.  Q_2 V_2.\cdots Q_\ell V_\ell$, where
  \emph{quantifier $Q_i \in \{\forall, \exists\}$} for $1 \leq i \leq \ell$ and
  $Q_j \neq Q_{j+1}$ for $1 \leq j \leq \ell-1$; and where $V_i$ are
  disjoint, non-empty sets of Boolean variables with
  $\var(Q)\eqdef \var(F)=\bigcup^\ell_{i=1}V_i$; and $F$ is a Boolean
  formula. If~$F$ is in ($c$-)CNF, $Q$ is a 
\emph{($c$-)$\QBFCNF$}. 
%
%
  %
  We call $\ell$ the \emph{quantifier depth} of~$Q$ and let
  $\matr(\Q)\eqdef F$.
Further, we denote the variables of~$Q$ by $\var(Q)\eqdef \var(F)$
and the existential (universal) variables by $\vare(Q)$ ($\varu(Q)$), 
defined by~$\vare(Q)\eqdef\bigcup_{1\leq i\leq \ell, Q_i=\exists} V_i$ ($\varu(Q)\eqdef\bigcup_{1\leq i\leq \ell, Q_i=\forall} V_i$), respectively. 
%
%
%
%
%
%
  %
  %
  
  %
%
  An \emph{assignment} is a mapping~$\alpha: X \rightarrow \{0,1\}$
  from a set~$X$ of variables. 
%
%
Given a Boolean formula~$F$ and an assignment~$\alpha$ for~$\var(F)$.
Then, for~$F$ in CNF, $F[\alpha]$ is a Boolean formula obtained by removing every~$c\in F$
with~$x\in c$ and $\neg x\in c$ if 
$\alpha(x)=1$ and $\alpha(x)=0$, respectively,
and by removing from every remaining clause~$c\in F$ literals~$x$ and $\neg x$
with~$\alpha(x)=0$ and $\alpha(x)=1$, respectively.
Analogously, for~$F$ in DNF values $0$ and~$1$ are swapped. 
More generally, for a QBF $\Q$ whose matrix consists of a conjunction of a CNF formula $F_1$ and a DNF formula $F_2$, and an assignment $\alpha$ for $X\subseteq \var(\Q)$, we define $\Q[\alpha]$ as the QBF obtained from $\Q$ by removing variables assigned by $\alpha$ from the prefix, replacing $F_1$ with $F_1[\alpha]$ and $F_2$ with $F_2[\alpha]$.

%
%
%
%
  For a given QBF~$\Q$ and an assignment~$\alpha$, $\Q[\alpha]$ is a
  QBF obtained from~$\Q$, where variables~$x$ mapped by~$\alpha$ are removed from preceding
  quantifiers accordingly, and~$\matr(\Q[\alpha])= (\matr(\Q))[\alpha]$.
  %
  %

A Boolean formula~$F$ \emph{evaluates to true} (or \emph{is satisfied})
if there exists an assignment~$\alpha$ for~$\var(F)$
such that~$F[\alpha]=\emptyset$ if~$F$ is in CNF or~$F[\alpha]=\{\emptyset\}$
if~$F$ is in DNF. We say that then~$\alpha$ \emph{satisfies}~$F$ or~$\alpha$ is a \emph{satisfying assignment} of~$F$. 
  A QBF~$\Q$ \emph{evaluates to true (or is valid)} if~$\ell=0$
  and $\matr(Q)$ evaluates to true under the empty assignment.
Otherwise,
  i.e., if~$\ell \neq 0$, we distinguish according to~$Q_1$.
  If~$Q_1=\exists$, then~$\Q$ evaluates to true if and only if there
  exists an assignment~$\alpha: V_1\rightarrow \{0,1\}$ such
  that~$\Q[\alpha]$ evaluates to true.  If~$Q_1=\forall$,
  then~$\Q$ evaluates to true if for every
  assignment~$\alpha: V_1 \rightarrow\{0,1\}$, we have that $\Q[\alpha]$ evaluates to
  true.  
%
We say that two QBFs are \emph{equivalent} if one evaluates to true 
whenever the other does.
Given a (C)QBF~$\Q$, the \emph{evaluation problem~\QBFSAT (\QBFSATCNF)} of QBFs asks
  whether $\Q$ evaluates to true. 
Then, $\SAT$ 
is 
$\QBFSAT$, but restricted
to one $\exists$ quantifier.
  In general,~\QBFSAT is \PSPACE-complete
~\cite{KleineBuningLettman99,Papadimitriou94,StockmeyerMeyer73}.
  
%
%

\begin{EX}\label{ex:running1}
  Consider CQBF~$\Q=\forall a,b. \exists c,d. C$,
  where~$C\eqdef \{c_1, c_2, c_3, c_4\}$ is a conjunction of clauses,
  with~$c_1\eqdef \neg a \vee \neg b \vee c$,
  $c_2\eqdef a \vee b \vee c$,
  $c_3\eqdef \neg a \vee \neg c \vee d$, and
  $c_4\eqdef a \vee \neg c \vee \neg d$.
%
  %
  Note that~$\Q[\alpha]$ is valid under any~$\alpha: \{a,b\} \rightarrow \{0,1\}$,
  %
which can be shown by giving 
an assignment~$\beta: \{c,d\}\rightarrow \{0,1\}$
for an arbitrary~$\alpha$. Concretely, let~$\beta(c) \eqdef 1$ whenever~$\alpha(a)=\alpha(b)$
and~$\beta(c)\eqdef 0$ otherwise. Further, $\beta(d)\eqdef 1$ whenever~$\alpha(a)=\alpha(b)=1$ and~$\beta(d)\eqdef 0$ otherwise.
Indeed, for any such~$\alpha$, we have that~$C[\alpha][\beta]=\emptyset$ and~$D[\alpha][\beta]=\{\emptyset\}$.
Consider, e.g.,~$\alpha=\{a\mapsto 0, b\mapsto 1\}$, satisfying~$c_1, c_2$ and~$c_3$; then $c_4$ is satisfied by~$\beta$. 
%
\end{EX}

\paragraph*{Extended Normalizations of the Matrix---Formulas in CDNF}
Our investigations also consider
%
%
a natural and
more general \emph{conjunctive/disjunctive normal form (CDNF)}
for QBFs. A QBF~$Q$, whose \emph{innermost} quantifier is~$Q_\ell$, is \emph{in CDNF}, 
whenever for a CNF~$C$ and DNF~$D$, we have $\matr(Q)=C{\wedge} D$ if~$Q_\ell{=}\exists$, 
and~$\matr(Q)=D {\vee} C$ if~$Q_\ell{=}\forall$.
Naturally, we say that $Q$ is in~\emph{$d$-CDNF} if~$C$ is in $d$-CNF and~$D$ is in $d$-DNF.
Further, $Q$ is in \emph{$d$,1-CDNF} if
$C$ is in $d$-CNF and~$D$ is in 1-DNF (i.e., $D$ is a long clause interpreted as a disjunction of singleton terms). 
Then, the problem $\QBFSAT_\ell$ refers to~\QBFSAT when restricted to QBFs in CDNF and quantifier depth~$\ell$. 

\paragraph*{Graph Representations}
In order to apply graph parameters to (Quantified) Boolean formulas, we need a graph representation. For a Boolean formula~$F$ in CNF or DNF we define the \emph{primal graph}~$G_F{\eqdef}(\var(F), E)$~\cite{SamerSzeider10b} over the variables of~$F$, where two variables are adjoined by an edge, whenever they appear together in at least one clause or term of~$F$, i.e.,~$E\eqdef \{\{x,y\}\mid f\in F, \{x,y\}\subseteq\var(f), x{\neq} y\}$.
The \emph{incidence graph}~$I_F{\eqdef}(\var(F) \cup F, E')$ of~$F$ is over the variables and clauses (or terms) of~$F$ and~$E'\eqdef \{\{f,x\} \mid f\in F, x\in \var(f)\}$.
For a QBF~$Q$ in CDNF
with~$\matr(Q)=C{\wedge}D$ or~$\matr(Q)=D{\vee}C$, 
respectively,
let the \emph{primal graph} of~$Q$ be $G_Q\eqdef$ $G_C \sqcup G_D$
and the \emph{incidence graph} of~$Q$ be~$I_Q\eqdef I_C \sqcup I_D$. 

\begin{EX}\label{ex:running3}
Recall $Q$ and~$C=\matr(Q)$ from Example~\ref{ex:running1}; observe primal and incidence graphs $G_Q$, $I_Q$ in Figure~\ref{fig:primal} (left,middle).
%
Assume a QBF~$Q'$ in CDNF obtained from~$Q$,
where~$\matr(Q')\eqdef C\wedge D$ 
with $D$ being a disjunction of (singleton) terms,~i.e.,~$D\eqdef \{\{b\}, \{\neg d\}\}$. 
%
%
Note that by definition the (1-)DNF formula~$D$ \emph{does not cause an additional edge} in the primal graph~$G_{Q'}$, i.e., the graph is equivalent to the primal graph~$G_C$ \emph{without~$D$}. 
So, $G_Q$ coincides with primal graph~$G_{Q'}$.
\end{EX}

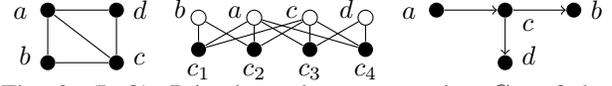
\begin{figure}[t]
\vspace{-.75em}
\centering
	\begin{tikzpicture}[node distance=7mm,every node/.style={circle,inner sep=2pt}]
\node (b) [fill,label={[text height=1.5ex]left:$a$}] {};
\node (a) [fill,right=of b,label={[text height=1.5ex,yshift=0.0cm,xshift=-0.05cm]right:$d$}] {};
\node (c) [fill,below of=b,label={[text height=1.5ex,yshift=0.09cm,xshift=0.05cm]left:$b$}] {};
\node (d) [fill,below of=a,label={[text height=1.5ex,yshift=0.09cm,xshift=-0.05cm]right:$c$}] {};
\draw (a) to (b);
\draw (a) to (d);
\draw (b) to (c);
\draw (b) to (d);
\draw (c) to (d);
\node (b2) [fill,right=11.5em of a,label={[text height=1.5ex]left:$a$}] {};
\node (a2) [fill,right=of b2,label={[text height=1.5ex,yshift=-0.15cm,xshift=-0.05cm]right:$c$}] {};
\node (c2) [fill,right=of a2,label={[text height=1.5ex,yshift=0.00cm,xshift=-0.05cm]right:$b$}] {};
\node (d2) [fill,below of=a2,label={[text height=1.5ex,yshift=0.09cm,xshift=-0.05cm]right:$d$}] {};
\draw[->] (b2) to (a2);
\draw [->] (a2) to  (c2);
\draw [->] (a2) to (d2);
%
%
\node (b) [right=2.5em of d,yshift=.5em,fill,label={[text height=1.5ex,yshift=0.5em]below:$c_1$}] {};
\node (a) [fill,right=1.5em of b,label={[text height=1.5ex,yshift=0.5em,xshift=-0.00cm]below:$c_2$}] {};
\node (d2) [fill,right=1.5em of a,label={[text height=1.5ex,yshift=0.5em,xshift=-0.00cm]below:$c_3$}] {};
\node (c) [fill,right=1.5em of d2,label={[text height=1.5ex,yshift=0.5em,xshift=-0.00cm]below:$c_4$}] {};
\node (vb) [draw=black,fill=white,above=.6em of b,xshift=-0em,yshift=.0em,fill,label={[text height=1.5ex,yshift=0.35em,xshift=.3em]left:$b$}] {};
\node (va) [draw=black,fill=white,right=1.5em of vb,yshift=.0em,fill,label={[text height=1.5ex,yshift=0.35em,xshift=.3em]left:$a$}] {};
\node (vc) [draw=black,fill=white,right=1.5em of va,yshift=.0em,fill,label={[text height=1.5ex,yshift=0.35em,xshift=.3em]left:$c$}] {};
\node (vd) [draw=black,fill=white,right=1.5em of vc,yshift=.0em,fill,label={[text height=1.5ex,yshift=0.35em,xshift=.3em]left:$d$}] {};
\draw (b) to (vb);
\draw (b) to (va);
\draw (b) to (vc);
\draw (a) to (vb);
\draw (a) to (va);
\draw (a) to (vc);
\draw (d2) to (vd);
\draw (d2) to (va);
\draw (d2) to (vc);
\draw (c) to (vd);
\draw (c) to (va);
\draw (c) to (vc);
%
\end{tikzpicture}\vspace{-1em}
\caption{(Left): Primal graph representation~$G_Q$ of the QBF~$Q$ of Example~\ref{ex:running1}. (Middle):  Incidence graph~$I_Q$ of QBF~$Q$. (Right): Treedepth decomposition of~$G_Q$.
}\label{fig:graphs}\label{fig:primal}
\end{figure}



\paragraph*{Treewidth and Pathwidth}
%

Let $G=(V,E)$ be a graph.
A \emph{tree decomposition (TD)}~\cite{RobertsonSeymour83,RobertsonSeymour91} of graph~$G$ is a pair
$\TTT=(T,\chi)$ where $T$ is a tree,
and $\chi$ is a mapping that assigns to each node $t$ of~$T$ a set
$\chi(t)\subseteq V$, called a \emph{bag}, such that the following
conditions hold:
(i) $V=\bigcup_{t\text{ of }T}\chi(t)$ and
$E \subseteq\bigcup_{t\text{ of }T}\SB \{u,v\} \SM u,v\in \chi(t)\SE$; and (ii)
for each \FIXCAM{$q, s, t$,} such that $s$ lies on the path from $q$ to
$t$, we have $\chi(q) \cap \chi(t) \subseteq \chi(s)$.
Then, $\width(\TTT) \eqdef \max_{t\text{ of }T}\Card{\chi(t)}-1$.  The
\emph{treewidth} $\tw{G}$ of $G$ is the minimum $\width({\TTT})$ over
all TDs $\TTT$ of $G$.  %
%
%
%
For \QBFSAT, the following tractability result is known. 

\begin{PROP}[Treewidth UB~\cite{Chen04a}]\label{ref:ubo}
Given any CQBF $Q$ of quantifier depth~$\ell$ with~$k{=}\tw{G_Q}$. 
$\QBFSAT$ on~$Q$ can be decided in time~$\tower(\ell, \mathcal{O}(k))\cdot\poly(\Card{\var(Q)})$.
\end{PROP}

However, it is not expected that one can significantly improve this, since already for the weaker pathwidth there are limits. 
The \emph{pathwidth} $\pw{G}$ of graph $G$ is the minimum width over
all TDs of $G$ whose trees are~paths.

\begin{PROP}[LB for Pathwidth~\cite{FichteHecherPfandler20}]\label{ref:pwlb}
Given any CQBF~$Q$ of quantifier depth~$\ell$ with~$k=\pw{G_Q}$. Then, under ETH, $\QBFSAT$ on~$Q$ cannot be decided in time~$\tower(\ell, o(k))\cdot\poly(\Card{\var(Q)})$.
\end{PROP}

%
%
%

\paragraph*{Treedepth}

Given a graph~$G=(V,E)$. Then, a \emph{treedepth decomposition}~$T=(V,F)$ of~$G$ is a forest of rooted trees,
where for every edge~$\{u,v\}\in E$ we require that $u$ is an ancestor or descendant of~$v$ in~$T$.
The \emph{treedepth}~$\td{G}$ of~$G$ is the smallest height among every  treedepth decomposition of~$G$, cf.\ Figure~\ref{fig:graphs} (right).
%
%
%

\paragraph*{Vertex Cover Number}
Given a graph~$G=(V,E)$. Then, a set~$S\subseteq V$ of vertices is a \emph{vertex cover (of~$G$)} if for every edge~$e\in E$ we have that~$e\cap S\neq\emptyset$.
Further, we define the \emph{vertex cover number} of a graph~$G$ to be the smallest size among every vertex cover of~$G$.
Interestingly, $\QBFSAT$ is tractable 
when parameterized by this number. 

\begin{PROP}[UB for Vertex Cover Number~\cite{LampisMitsou17}]\label{prop:vco}
Given any CQBF~$Q$ of $\QBFSAT$ with~$k$ being the vertex cover number of~${G_Q}$. Then, the validity of~$Q$ can be decided in time~$2^{2^{\mathcal{O}(k)}}\cdot\poly(\Card{\var(Q)})$ ($2^{{\mathcal{O}(k^3)}}\cdot\poly(\Card{\var(Q)})$ for~$\matr(Q)$ in 3-CNF).
\end{PROP}
%

\paragraph*{Feedback Sets and Distance Measures} 

Lower bound results for $\QBFSAT$ parameterized by treewidth (pathwidth) or vertex cover number 
motivates other~parameters. 
%
%
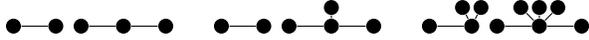
\begin{figure}[t]
\vspace{-.75em}
\hspace{-1em}\begin{tikzpicture}[node distance=3.5mm,every node/.style={circle,inner sep=2pt}]
\node (a) [fill,label={[text height=1.5ex]left:$ $}] {};
\node (b) [fill,right=of a,label={[text height=1.5ex,yshift=0.0cm,xshift=-0.05cm]right:$ $}] {};
\node (c) [fill,right=0.35em of b,label={[text height=1.5ex,yshift=0.09cm,xshift=0.05cm]left:$ $}] {};
\node (d) [fill,right=of c,label={[text height=1.5ex,yshift=0.09cm,xshift=-0.05cm]right:$ $}] {};
\node (e) [fill,right=of d,label={[text height=1.5ex,yshift=0.09cm,xshift=-0.05cm]right:$ $}] {};
\draw (a) to (b);
\draw (c) to (d);
\draw (d) to (e);
\node (a2) [fill,right=1.5em of e,label={[text height=1.5ex]left:$ $}] {};
\node (b2) [fill,right=of a2,label={[text height=1.5ex,yshift=0.0cm,xshift=-0.05cm]right:$ $}] {};
\node (c2) [fill,right=0.35em of b2,label={[text height=1.5ex,yshift=0.09cm,xshift=0.05cm]left:$ $}] {};
\node (d2) [fill,right=of c2,label={[text height=1.5ex,yshift=0.09cm,xshift=-0.05cm]right:$ $}] {};
\node (d2p) [fill,above=0.1em of d2,label={[text height=1.5ex,yshift=0.09cm,xshift=-0.05cm]right:$ $}] {};
\node (e2) [fill,right=of d2,label={[text height=1.5ex,yshift=0.09cm,xshift=-0.05cm]right:$ $}] {};
\draw (a2) to (b2);
\draw (c2) to (d2);
\draw (d2) to (e2);
\draw (d2) to (d2p);
\node (a3) [fill,right=1.5em of e2,label={[text height=1.5ex]left:$ $}] {};
\node (b3) [fill,right=of a3,label={[text height=1.5ex,yshift=0.0cm,xshift=-0.05cm]right:$ $}] {};
\node (b3p) [fill,above=0.1em of b3,xshift=-.35em,label={[text height=1.5ex,yshift=0.0cm,xshift=-0.05cm]right:$ $}] {};
\node (b3p2) [fill,above=0.1em of b3,xshift=.35em,label={[text height=1.5ex,yshift=0.0cm,xshift=-0.05cm]right:$ $}] {};
\node (c3) [fill,right=0.35em of b3,label={[text height=1.5ex,yshift=0.09cm,xshift=0.05cm]left:$ $}] {};
\node (d3) [fill,right=of c3,label={[text height=1.5ex,yshift=0.09cm,xshift=-0.05cm]right:$ $}] {};
\node (d3p) [fill,above=0.1em of d3,xshift=-0em,label={[text height=1.5ex,yshift=0.09cm,xshift=-0.05cm]right:$ $}] {};
\node (d3p2) [fill,above=0.1em of d3,xshift=.7em,label={[text height=1.5ex,yshift=0.09cm,xshift=-0.05cm]right:$ $}] {};
\node (d3p3) [fill,above=0.1em of d3,xshift=-.7em,label={[text height=1.5ex,yshift=0.09cm,xshift=-0.05cm]right:$ $}] {};
\node (e3) [fill,right=of d3,label={[text height=1.5ex,yshift=0.09cm,xshift=-0.05cm]right:$ $}] {};
\draw (a3) to (b3);
\draw (b3) to (b3p);
\draw (b3) to (b3p2);
\draw (c3) to (d3);
\draw (d3) to (e3);
\draw (d3) to (d3p);
\draw (d3) to (d3p2);
\draw (d3) to (d3p3);

\end{tikzpicture}\vspace{-0.5em}
\caption{Disjoint paths (left); half-ladder graph (middle); and caterpillar graph (right).} 
\label{fig:ladder}
\end{figure}
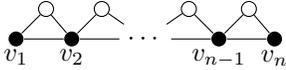
\begin{figure}[t]
\vspace{-1em}
\centering
	\begin{tikzpicture}[node distance=5mm,every node/.style={circle,inner sep=2pt}]
%
\node (b) [yshift=.5em,fill,label={[text height=1.5ex,yshift=.65em]below:$v_1$}] {};
\node (a) [fill,right=1.5em of b,label={[text height=1.5ex,yshift=0.65em,xshift=-0.00cm]below:$v_2$}] {};
\node (ba) [draw,above right=1em of b,label={[text height=1.5ex]above:$ $}] {};
\node (cd) [draw,above right=1em of a,label={[text height=1.5ex]above:$ $}] {};
\node (c) [draw=black,color=white,right=1.5em of a,label={[text height=1.5ex,yshift=0.00cm,xshift=0.00cm]left:$ $}] {\textcolor{black}{$\cdots$}};
\node (de) [fill,right=1.5em of c,label={[text height=1.5ex,yshift=1.1em,xshift=-0.00cm]below:$v_{n-1}$}] {};
\node (cd2) [draw,above left=1em of de,label={[text height=1.5ex,yshift=-1em]above:$ $}] {};
\node (cd3) [draw,above right=1em of de,label={[text height=1.5ex,yshift=-1em,xshift=2em]above:$ $}] {};
\node (d) [fill,right=1.5em of de,label={[text height=1.5ex,yshift=0.65em,xshift=-0.00cm]below:$v_n$}] {};
\draw (a) to (b);
\draw (a) to (ba);
\draw (b) to (ba);
\draw (c) to (cd);
\draw (a) to (cd);
\draw (c) to (cd2);
\draw (de) to (cd2);
\draw (de) to (cd3);
\draw (d) to (cd3);
\draw (a) to (c);
\draw (c) to (de);
\draw (de) to (d);
\end{tikzpicture}\vspace{-1.5em}
\caption{
Graph of pathwidth~$2$ with FV number in~$\mathcal{O}(n)$.}
\label{fig:constpaths}
\end{figure}

\begin{figure*}[t]\centering
 \begin{tikzpicture}[scale=.6,node distance=0.05mm, every node/.style={scale=0.85}]
\tikzset{every path/.style=thick}
\node (leaf1) [tdnode,label={[yshift=-0.25em,xshift=0.5em]above left:$ $}] {\emph{distance to sparse half-ladder}};
\node (vc) [tdnode,dashed,label={[yshift=-0.25em,xshift=0.5em]above left:$ $},above=.35em of leaf1] {vertex cover number};
\node (disj) [tdnode,label={[yshift=-0.25em,xshift=0.5em]above left:$ $},right=2em of leaf1] {\emph{distance to half-ladder}};
\node (dl) [tdnode,label={[yshift=-0.25em,xshift=0.5em]above left:$ $},right=2em of disj] {\emph{distance to caterpillar}};
\node (pw) [tdnode,label={[yshift=-0.25em,xshift=0.5em]above left:$ $},right=.5em of dl] {\emph{pathwidth}};
\node (td) [tdnode,label={[yshift=-0.25em,xshift=0.5em]above left:$ $},above=.35em of pw,xshift=0em] {\textbf{\emph{treedepth}}};
\node (tw) [tdnode,dashed,label={[yshift=-0.25em,xshift=0.5em]above left:$ $},below=.35em of pw] {treewidth};
\node (leaf2) [tdnode,label={[xshift=-1.0em, yshift=-0.15em]above right:$ $}, below = .35em of leaf1]  {\emph{sparse feedback vertex number}};
\node (leaffvs) [tdnode,text width=5.1em,yshift=1.2em,xshift=3.8em,thick,label={[xshift=-1.0em, yshift=-0.15em]above right:$ $}, left = 3.7em of leaf2]  {\textbf{FV number} \textbf{(incidence graph)}};
\node (fvs) [tdnode,label={[xshift=-1.0em, yshift=-0.15em]above right:$ $}, xshift=-.5em, below = .35em of disj]  {\textbf{\emph{feedback vertex number}}};
\node (dop) [tdnode,label={[xshift=-1.0em, yshift=-0.15em]above right:$ $}, right = .75em of fvs]  {\emph{distance to outerplanar}};
\node (tdi) [tdnode,text width=4.3em,yshift=-1.2em,xshift=.4em,label={[yshift=-0.25em,xshift=0.5em]above left:$ $},right=.35em of td] {\textbf{treedepth (incidence graph)}};
\draw [->] (leaf1) to (leaf2);
\draw [->] (leaf1) to (disj);
\draw [->] (leaf2) to (fvs);
\draw [->] (leaf2.west) to (leaf2.west-|leaffvs.east);
\draw [->] (dl) to (fvs);
\draw [->] (disj) to (dl);
\draw [->] (dl) to (pw);
\draw [->] (td) to (pw);
\draw [->] (vc) to (td);
\draw [->] (pw) to (tw);
\draw [->] (fvs) to (dop);
\draw [->] (dop) to (tw);
\draw [->] (vc) to (leaf1);
\draw [->] (td.east) to  (td.east-|tdi.west); 
\end{tikzpicture}\vspace{-.5em}
\caption{Parameters for QBFs of the primal graph (incidence graph), where italic text refers to \emph{parameters between vertex cover number and treewidth}. A directed arc indicates the source being weaker than the destination, i.e., the destination is linearly bounded by the source. Bold-face parameters mark selected new lower bounds for $\QBFSAT_\ell$.}\label{fig:params}
\end{figure*}
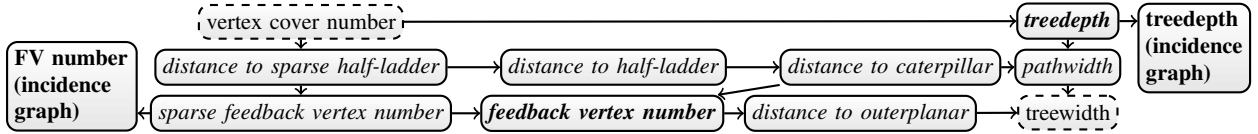
Let~$G=(V,E)$ be a graph. Then, a set~$S\subseteq V$ of vertices is called a \emph{feedback vertex set (FVS)} of~$G$ if~$G-S$ is an acyclic graph, and the \emph{feedback vertex number} (of~$G$) refers to the smallest size among all feedback vertex sets of~$G$.
Further, $S$ is referred to by \emph{distance set to half-ladder} 
if~$G-S$ is a \emph{half-ladder (graph)}, consisting of (vertex) disjoint paths such that additionally each vertex might be adjacent to \emph{one} fresh vertex. 
If we allow more than one such fresh vertex, we call the graph a \emph{caterpillar}, cf.\ Figure~\ref{fig:ladder}.
The smallest~$k=\Card{S}$ among these distance sets~$S$ 
is the \emph{distance} (to the corresponding graph class).
We say $S$ is a \emph{$c$-deletion set}, for some integer $c$, 
if every component of $G{-}S$ has at most $c$ vertices.
A set $D\subseteq E$ is a \emph{feedback edge set (FES)} for $G$ if
$G-D$ is acyclic. 
%
%

We utilize these sets~$S,D$ \emph{for 
%
a QBF~$Q$}, where $G=G_Q$. 
%
Then, $S$ is \emph{sparse} if for every two distinct vertices~$u,v$ of~$G_Q-S$ 
there is at most one clause or term $f$ of $\matr(Q)$ with $u,v\in\var(f)$. 
%
%

\begin{EX}
Recall QBF~$Q$ from Example~\ref{ex:running1} and observe that the feedback vertex number is~$1$, e.g., $\{a\}$ is a FVS of~$G_Q$ as well as a distance set to half-ladder 
of~$G_Q$. However, the sparse feedback vertex number of~$Q$ is~$2$, since no single vertex can be removed from~$G_Q$ such that each edge corresponds to at most one clause of~$\matr(Q)$. Set $\{a,c\}$ is a sparse FVS of~$Q$ since $G_Q{-}\{a,c\}$ is edgeless.
While for~$G_Q$ pathwidth is identical to the sparse FV number, 
the graph of Figure~\ref{fig:constpaths} has pathwidth $2$, but admits only a large FVS, e.g., all white nodes.
\end{EX}

Inspired by related and more general works on parameter hierarchies \cite{SorgeWeller21,graphclasses}, we obtain a hierarchy of parameters for QBFs: 
Figure~\ref{fig:params} depicts parameters, where a directed arc from parameter~$p_1$ to parameter~$p_2$ indicates that~$p_1$ is \emph{weaker} than~$p_2$, i.e., $p_2$ is upper-bounded by~$\mathcal{O}(p_1)$.
Consequently, lower bounds for the weaker (linearly smaller) parameter $p_1$ form \emph{stronger results} and automatically carry over to the stronger parameter~$p_2$. 

\begin{EXa}
Observe that already for a CQBF~$Q$, the feedback vertex number~$k_G$ of~$G_Q$ and the feedback vertex number~$k_I$ of~$I_Q$ are incomparable,
cf., Figure~\ref{fig:params}. 
It is easy to see that~$k_I \ll k_G$ by constructing an instance with one large clause.
However, there are also cases where~$k_I \gg k_G$: One can construct pairs of variables where each pair appears in (at least) two clauses of size~$2$, i.e.,
each pair is involved in a cycle in~$G_I$. Then,
$k_G$ is zero, but~$k_I$ amounts to the number of pairs.
\end{EXa}


\section{Structure-Aware (SAW) Reductions}\label{sec:saw}

Recall the gap between runtimes for $\QBFSAT$ using treewidth (pathwidth, cf., Proposition~\ref{ref:pwlb}) and $\QBFSAT$ 
when parameterized by vertex cover number (see Proposition~\ref{prop:vco}). 
Interestingly, runtime bounds for $\QBFSAT$ and vertex cover number on CNFs also hold on CDNFs.
\begin{THM}[UB for $\QBFSAT_\ell$ and Vertex Cover Number, $\star$\footnote{Statements marked with a star (``$\star$'') are  proven in the appendix.}]\label{thm:cnfvcn}
There is an algorithm that, given a QBF $\qbfformula$ in CDNF with vertex cover
number $k$ of~$G_\qbfformula$, decides whether $\qbfformula$ is true in time $2^{2^{\bigO(k)}}\cdot
\poly(\Card{\var(\qbfformula)})$.
If $\qbfformula$ is in $d$-CDNF, the algorithm runs in time $2^{k^{\bigO(d)}} \cdot \poly(\Card{\var(\qbfformula}))$.
\end{THM}\vspace{-.75em}
\begin{proof}[Proof (Idea)]
The result can be established by enhancing a DPLL-style backtracking algorithm with
formula caching~\cite{BeameIPS10}.
The number of subformulas of the matrix that can be obtained by
assigning variables can be bounded by a function that is linear in the
number of variables and only exponential in the size of the
vertex cover.
This upper bounds the size of the search tree. 
%
\end{proof}


%
%
%

%

Similarly, the known runtime result for treewidth (Proposition~\ref{ref:ubo}) carries over to CDNFs.

\begin{THM}[UB for $\QBFSAT_\ell$ and Treewidth]\label{ref:ub}
Given any QBF~$Q$ in CDNF of quantifier depth~$\ell$ with~$k=\tw{G_Q}$. Then, $\QBFSAT_\ell$ on~$Q$ can be decided in time~$\tower(\ell, \mathcal{O}(k))\cdot\poly(\Card{\var(Q)})$.
\end{THM}

We show this result by relying on 
the following concept of structure-aware reductions~\cite{FichteHecherPfandler20,FichteHecherMahmood21}. 
These reductions will be a key component of the constructions for the new lower bound results of this paper. 
%
%
 They provide a constructive way of utilizing an actual structural representation of the instance,
thereby precisely bounding
the parameter increase (decrease) in
terms of the representation. 

\begin{DEF}[SAW Reduction]
  Let $\Sigma$, $\Sigma'$, $\Gamma$, $\Gamma'$ be alphabets,
  $\mathcal{P} \subseteq \Sigma^* {\times}\NAT$,
  $\mathcal{P'} \subseteq \Sigma'^* {\times} \NAT$~be~parameterized problems
  with parameterizations $\kappa, \kappa'$, 
  and~$f$~be a computable function.
  An $f$-structure-aware reduction $\mathcal{R}$ from
  $\mathcal{P}$ to $\mathcal{P'}$ 
  maps $\Sigma^* {\times} \Gamma^*$ to
  $\Sigma'^* {\times} \Gamma'^*$ where for
  $(\mathcal{I}, S) {\;\in\;} \Sigma^* {\times} \Gamma^*$ with
  $(\mathcal{I'}, S') {=} \mathcal{R}(\mathcal{I}, S)$, we~have 
  $(({\cal I},S),\kappa(S))\allowbreak{\mapsto}  (({\cal I}',S'), \kappa'(S'))$ is a pp-reduction s.t.~(i)\allowbreak 
   $S'{=}g(S)$ for polynomial-time function $g$ 
  \emph{(functional dependency)};
  (ii)~$\kappa'(S') \leq \mathcal{O}(f(\kappa(S)))$ \emph{($f$-boundedness)}.






\end{DEF}

The definition of SAW reductions serves the following purposes.
First, such a reduction always provides a structural representation of the reduced instance,
whereby (i) the functional dependency immediately gives insights into how such a representation can be obtained. 
Further, the property (ii) $f$-boundedness ensures that the resulting parameter of the reduced instance fulfills precise guarantees, which will be essential for 
the next subsection.
%

To demonstrate these reductions, we briefly explain arcs of Figure~\ref{fig:params}.
%
Interestingly, almost every arc of Figure~\ref{fig:params} 
can be shown by the trivial linear-SAW reduction that takes a QBF~$Q$ of~$\QBFSAT_\ell$ and a structural representation~$S$ of the respective parameter and returns $(Q,S)$.
Indeed, e.g., any vertex cover~$S$ is a distance set to half-ladder 
and for any path decomposition~$S$ it holds that it is a tree decomposition.
Further, any distance set~$S$ to outerplanar can be turned into a TD~$S'$, since each outerplanar graph~\cite{Syslo79} has a TD~$\mathcal{T}$ of width at most~$2$ and we obtain TD~$S'$ by adding~$S$ to every bag of~$\mathcal{T}$.
%
%
%
%
%
%

With these SAW reductions at hand,
one can easily establish Theorem~\ref{ref:ub}. 
%
%
Since the treewidth parameter is linearly bounded by both feedback vertex size and treedepth, 
we instantly obtain the following results. 

\begin{COR}[UB for $\QBFSAT_\ell$ and Feedback Vertex Number/Treedepth]\label{ref:ubs}
Given any QBF~$Q$ in CDNF of quantifier depth~$\ell$ with~$k$ being the feedback vertex number or treedepth of ${G_Q}$. Then, $\QBFSAT_\ell$ on~$Q$ can be decided in time~$\tower(\ell, \mathcal{O}(k))\cdot\poly(\Card{\var(Q)})$.
\end{COR}

\section{Lower Bounds via SAW Reductions}\label{sec:main}

In order to establish conditional lower bounds for parameterized problems, one might reduce from \SAT on 3-CNFs and then directly apply the widely believed exponential time hypothesis (ETH)~\cite{ImpagliazzoPaturiZane01}. Indeed, many results have been shown, where the parameter of the reduced instance depends on the number of variables of the Boolean formula, e.g.,~\cite{MarxMitsou16,LampisMitsou17,PilipczukSorge20}, immediately followed by applying ETH.
Oftentimes consequences of the ETH are sufficient, claiming that \SAT on 3-CNFs cannot be solved in time~$2^{o(k)}\cdot\poly(n)$, where~$k$ is a 
parameter of the instance and~$n$ is the variable number. 

Having established the concept of structure-aware (SAW) reductions,
we apply this type of reductions as a precise tool for generalizing the lower bound result of Proposition~\ref{ref:pwlb}. 
%
More specifically, the next subsection focuses on defining a self-reduction
from \QBFSAT to \QBFSAT, where we trade an increase of quantifier depth
for an exponential decrease in the parameter of interest.
%
This is done in such a way that we are able to show lower bounds matching their upper bounds when assuming ETH.



In order to find suitable candidate parameters, recall that for the vertex cover number~$k$,
the problem \QBFSAT can be solved in double-exponential runtime in~$k$, regardless of quantifier depth.
However for the treewidth (or pathwidth) this is not the case, since for quantifier depth~$\ell$,
one requires a runtime that is~$\ell$-fold exponential in the pathwidth, cf.\ Proposition~\ref{ref:pwlb}.
This motivates our quest to investigate suitable parameters that are ``between'' vertex cover number and pathwidth.
In Section~\ref{sec:result}, we show that the (sparse) feedback vertex number is insufficient as well,
i.e., we obtain ETH-tight lower bounds that match the upper bounds of Corollary~\ref{ref:ubs} for this parameter. 
%
%
%
Section~\ref{sec:tdpth} adapts the reduction, thereby providing deeper insights into hardness for treedepth. 
Further, this section outlines an ETH-tight lower bound for treedepth, similar to Corollary~\ref{ref:ubs}, 
for instances of high treewidth. 

\vspace{-.25em}
\subsection{Tight QBF Lower Bound for Feedback Vertex Number}\label{sec:result}

The overall approach proceeds via SAW reductions as follows.
We assume an instance~$Q$ of~$\QBFSAT_\ell$ and a sparse feedback vertex set~$S$ such that~$\Card{S}=k$. 
Then, we devise a $\log(k)$-SAW reduction~$\mathcal{R}$, constructing an equivalent QBF~$Q'$
such that~${Q'}$ has a sparse feedback vertex set~$S'$ of size $\mathcal{O}(\log(k))$.

Without loss of generality, we restrict ourselves to the case where the innermost quantifier of the QBF~$Q$ is~$\exists$, as one can easily adapt for the other case or solve the inverse problem and invert the result in constant time.
Further, we assume that the first quantifier of~$Q$ is~$\exists$ as well.
Let~$Q=\exists V_1. \forall V_2. \cdots \exists V_\ell. C \wedge D$ be such a QBF that admits a sparse feedback vertex set~$S$ with~$k=\Card{S}$.
For the purpose of our lower bound, we actually prove a \emph{stronger result in the more restricted 3,1-CDNF form}, 
where we assume~$C$ in 3-CNF and~$D$ in 1-DNF, i.e., $G_Q=G_C$ as discussed in Example~\ref{ex:running3}. 
Finally, we assume that each~$c_i\in C$ consists of exactly three literals; however, the reduction works with individual smaller clause sizes.
The reduced instance~$Q'$ and sparse feedback vertex set~$S'$ of~${Q'}$ that is obtained by the SAW reduction, uses the additional quantifier block
in order to ``unfold'' $S'$ (i.e., reconstruct an assignment of~$S$). 

\vspace{-.05em}
\paragraph*{Auxiliary Variables}
%
%
%
%
In order to construct~$Q'$, we require the following \emph{additional (auxiliary) variables}.
First, we use pointer or \emph{index variables} that are used to address precisely \emph{one element} of~$S$.
In order to address~$3$ elements of~$S$ for the evaluation of a 3-CNF (3-DNF) formula, we require three of those indices.
These index variables are of the form~$\lbvs{S}\eqdef \{idx_j^1, \ldots, idx_j^{\ceil{\log(\Card{S})}} \mid 1\leq j\leq 3\}$
and, intuitively, for each of the three indices these allow us to ``address'' each of the~$k$ many elements of~$S$ via a specific assignment of~$\ceil{\log(k)}$ many Boolean variables.
These~$2^{\ceil{\log(k)}}$ many combinations of variables per index~$j$ are sufficient to address any of the~$k$ elements of~$S$.
To this end, we assign each element~$x\in S$ and each~$1\leq j\leq 3$ a set consisting of an arbitrary, but fixed and \emph{unique combination of literals} over the index variables~$idx_j^1, \ldots, idx_j^{\ceil{\log(\Card{S})}}$, denoted by~$\bval{x}{j}$.

Further, 
for each clause~$c_i\in C$ with~$c_i=\{l_1, l_2, l_3\}$ we assume an arbitrary ordering among the literals of~$c_i$ and write~$\lit{c_i}{j}\eqdef l_j$ for the~\emph{$j$-th literal of~$c_i$} ($1\leq j\leq 3$).
%
%
%
We also require three Boolean variables~$val_1, val_2, val_3$, where~$val_j$ captures a \emph{truth (index) value} for the element of~$S$ that is addressed via the variables for the~$j$-th index.
These variables are referred to by~$\lbvvs{S}\eqdef \{val_1, val_2, val_{3}\}$.
Finally, we use one variable to store whether $D$ is satisfied as well as~$\Card{C}$ many auxiliary variables that indicate whether a clause~$c\in C$ is satisfied.
These variables are addressed by the set~$\lsat\eqdef\{sat, sat_{1}, \ldots,sat_{ {\Card{C}}}\}$ of \emph{satisfiability variables},
where we assume clauses~$C{=}\{c_1,\ldots,$ $c_{\Card{C}}\}$ are ordered  according to some fixed total ordering.

\paragraph*{The Reduction}

The reduction~$\mathcal{R}$ takes~$Q$ and~$S$ and constructs 
an instance~$Q'$ as well as a sparse feedback vertex set~$S'$ of~${Q'}$.
The QBF $Q'$ is of the form $Q'\eqdef$
\vspace{-.6em}
\[\exists V_1.\ \forall V_2.\ \cdots \exists V_\ell.\ 
\forall 
\lbvs{S}, \lbvvs{S}, \lsat
.\ C' \vee D',  \] 

\vspace{-.55em}
\noindent where~$C'$
is in DNF, defined as a disjunction of terms:

\begingroup
\allowdisplaybreaks
\vspace{-.95em}
{\smallalign{\normalfont\small}
\begin{flalign}
	\label{red:guessatomj}&x \wedge \bigwedge_{b \in \bval{x}{j}} b \wedge \neg val_j&\pushleft{\text{for each } x\in S, 1\leq j\leq 3}\qquad\raisetag{2.5em}\\
	\label{red:guessnegatomj}&\neg x \wedge \bigwedge_{b \in \bval{x}{j}} b \wedge val_j&\pushleft{\text{for each } x\in S, 1\leq j\leq 3}\qquad\raisetag{2.5em}\\
%
	&sat_{i} \wedge l & \pushleft{\text{for each } c_i\in C, 1\leq j\leq 3\text{ with}}\notag\\[-.4em]\label{red:usatv} &&\pushleft{\lit{c_i}{j}=l,
\var(l)\in \var(C)\setminus S}\raisetag{1.15em}\\
	&sat_{i} \wedge \neg b & \pushleft{\text{for each } c_i\in C, 1\leq j\leq 3, x\in S,}\notag\\[-.4em] \label{red:usatxp} &&\pushleft{b\in \bval{x}{j}
\text{ with}\var(\lit{c_i}{j})=x}\raisetag{1.15em}\\
%
	&sat_{i} \wedge val_j & \pushleft{\text{for each } c_i\in C, 1\leq j\leq 3, x\in S 
\text{ with}}\notag\\[-.4em] 
\label{red:usatxv}&&\pushleft{\lit{c_i}{j}=x}\raisetag{1.15em}\\
%
%
%
%
%
%
%
%
%
	 &sat_{i} \wedge \neg val_j & \pushleft{\text{for each } c_i\in C, 1\leq j\leq 3, x\in S
\text{ with}}\notag\\[-.4em]
\label{red:usatxnegv}&&\pushleft{\lit{c_i}{j}=\neg x}\raisetag{1.15em}\\[.5em]
%
%
%
%
	\label{red:usatvd} &sat \wedge l & \pushleft{\text{for each } \{l\}\in D}\qquad\raisetag{1.3em} 
%
%
%
%
%
%
%
\end{flalign}}\endgroup
%

\noindent Additionally, we define~$D'$ in 1-CNF, which is a conjunction of the following singletons.
\vspace{-.25em}
{\smallalign{\normalfont\small}
\begin{flalign}
	\label{red:usatnot} &\pushleft{\neg sat_i} & {\makebox[5em]{}\text{for each }1\leq i \leq \Card{C}}\\
	\label{red:usatnotd} &\pushleft{\neg sat} 
\end{flalign}}
\vspace{-1em}

Observe that the fresh auxiliary variables appear under the innermost universal quantifier of~$Q'$.
So, intuitively, Formulas~(\ref{red:guessatomj}) ensure that whenever
some~$x\in S$ is set to $1$ and the~$j$-th index targets~$x$, that we then ``skip'' the corresponding assignment if~$val_j$ is set to $0$.
This is similar to Formulas~(\ref{red:guessnegatomj}) for the case~$x\in S$ is set to false, ensuring that for the remaining formulas of~$Q'$ 
 whenever the~$j$-th index targets some~$x\in S$,
the corresponding value~$val_j$ agrees with the implicit assignment of $x$.

Formulas~(\ref{red:usatv})--(\ref{red:usatxnegv}) are used to check that the clauses~$c_i\in C$ are satisfied.
Intuitively, the variables~$sat_i$ serve as switches that require clause~$c_i$ to be satisfied if set to true.
Formulas~(\ref{red:usatv}) ensure that $c_i$ is satisfied whenever a literal~$l\in c_i$, whose variable~$\var(l)$ is not in~$S$, is assigned true.
For literals~$l\in c_i$ whose variables~$\var(l)$ are in~$S$, Formulas~(\ref{red:usatxp}) evaluate to true if the corresponding $j$-th index ($j$ such that~$l=\lit{c_i}{j}$) does not target~$\var(l)$, or one of Formulas~(\ref{red:usatxv}) and~(\ref{red:usatxnegv}) is true if the targeted literal is true.
Finally, one of Formulas~(\ref{red:usatvd}) evaluates to true if~$D$ is true, similarly to Formulas~(\ref{red:usatv}).
Observe that since the~$\lsat$ variables are universally quantified, multiple~$sat_i$ variables might be set to true.
Intuitively, this makes it ``easier'' to satisfy some term among Formulas~(\ref{red:usatv})--(\ref{red:usatvd}), since it is sufficient for one of the clauses $c_i$ to be satisfied.
The only problematic assignment of~$\lsat$ variables is the one where both~$sat$ as well as all the $sat_i$ variables are set to false.
This is prevented by 1-CNF formula $D'$, i.e., Formulas~(\ref{red:usatnot}),(\ref{red:usatnotd}).

Example~\ref{ex:redu} ($\star$) illustrates $\mathcal{R}$ on a specific formula.
%
%
%

\begin{figure*}[t]\centering
\vspace{-1.25em}
\begin{tikzpicture}[node distance=7mm,every node/.style={circle,inner sep=2pt}]
\node (b) [fill,label={[text height=1.5ex]left:$v_1$}] {};
\node (y) [fill,below right=2.35em of b,label={[text height=1.5ex]left:$v_3$}] {};
\node (z) [fill,below left=2.35em of y,label={[text height=1.5ex]left:$v_4$}] {};
\node (f) [fill,below left=2.35em of z,label={[text height=1.5ex]left:$v_5$}] {};
\node (d) [fill,below=1.5em of z,label={[text height=1.5ex,yshift=.5em]below:$v_6$}] {};
\node (q) [fill,below right=2.35em of y,label={[text height=1.5ex]left:$v_8$}] {};
\node (e) [fill,below right=2.35em of z,label={[text height=1.5ex]right:$v_7$}] {};
%
%
\draw (y) to (b);
\draw (y) to (q);
\draw (z) to (e);
\draw (z) to (f);
\draw (d) to (z);
%
\node (c) [tdnode,ellipse,draw=black,text height=5em,yshift=-3em,text width=2em,xshift=-.75em,color=white,right=5em of b,label={[text height=1.5ex,yshift=0.00cm,xshift=0.00cm]below:$S$}] {}; 
\node (s1) [draw=black,right=5em of b,yshift=-.8em,label={[text height=1.5ex,yshift=0.00cm,xshift=0.35em]right:$s_1$}] {};
\node (s2) [draw=black,below=1em of s1,label={[text height=1.5ex,yshift=0.00cm,xshift=0.35em]right:$s_2$}] {};
\node (s3) [draw=black,below=1em of s2,label={[text height=1.5ex,yshift=0.00cm,xshift=0.35em]right:$s_3$}] {};
\node (s4) [draw=black,below=1em of s3,label={[text height=1.5ex,yshift=0.00cm,xshift=0.35em]right:$s_4$}] {};
\draw (s1) to (s2);
\draw [in=5] (s1) to (s3);
\draw [in=5] (s2) to (s4);
\draw (s2) to (s3);
\draw (s3) to (s4);
\draw [dashed,in=20,out=140] (s1) to (b);
\draw [dashed,in=20,out=140] (s2) to (b);
\draw [dashed,in=20,out=140] (s2) to (y);
\draw [dashed,in=-25,out=140] (s4) to (z);
\draw [dashed,in=20,out=140] (s4) to (e);
\draw [dashed,in=20,out=140] (s2) to (b);
\draw [dashed,in=20,out=140] (s3) to (b);
\draw [dashed,in=20,out=130] (s4) to (b);
\node (arr) [draw=white,right=3em of c,label={}] {$\xRightarrow{\;\mathcal{R}(Q,S)\;}$};
%
%
%
%
\node (b2) [fill,right=19em of b,label={[text height=1.5ex]left:$v_1$}] {};
\node (sb2) [fill,right=20.85em of b, yshift=.3em,label={[text height=1.5ex,yshift=1em,xshift=-.45em]below:{\tiny $sat_2$}}] {};
\node (x2s) [fill,below left=.9em of b2,label={[text height=1.5ex,xshift=.45em]left:{\tiny $sat_1$}}] {};
\node (y2) [fill,below right=2.35em of b2,label={[text height=1.5ex]left:$v_3$}] {};
\node (y2s) [fill,below right=.9em of b2,label={[text height=1.5ex,xshift=.45em]left:{\tiny $sat_3$}}] {};
\node (z2) [fill,below left=2.35em of y2,label={[text height=1.5ex]left:$v_4$}] {};
\node (z2s2) [fill,below right=.9em of y2,label={[text height=1.5ex,xshift=.45em]left:{\tiny $sat_7$}}] {};
\node (f2) [fill,below left=2.35em of z2,label={[text height=1.5ex]left:$v_5$}] {};
\node (f2s) [fill,below left=.75em of z2,label={[text height=1.5ex,xshift=.45em]left:{\tiny$sat_4$}}] {};
\node (d2) [fill,below=1.5em of z2,label={[text height=1.5ex,yshift=.5em]below:$v_6$}] {};
\node (d2s) [fill,below=.9em of z2,label={[text height=1.5ex,xshift=.45em]left:{\tiny$sat_5$}}] {};
\node (q2) [fill,below right=2.35em of y2,label={[text height=1.5ex,xshift=-.1em,yshift=.2em]left:$v_8$}] {};
\node (e2) [fill,below right=2.35em of z2,label={[text height=1.5ex,xshift=.45em,yshift=-.25em]left:$v_7$}] {};
\node (sat9) [fill,right=3.5em of e2,yshift=.5em,label={[text height=1.5ex,xshift=.45em,yshift=-.25em]left:{\tiny $sat_8$}}] {};
\node (e2s) [fill,below right=.9em of z2,label={[text height=1.5ex,xshift=1.58em,yshift=.65em]left:{\tiny$sat_6$}}] {};
\draw (b2) to (sb2);
\draw (b2) to (x2s);
\draw (y2) to (b2);
\draw (y2) to (q2);
\draw (z2) to (e2);
\draw (z2) to (f2);
\draw (d2) to (z2);

\node (sx) [draw=black,right=9em of b2,yshift=1em,label={[text height=1.5ex,yshift=0.00cm,xshift=0.35em]right:$ $}] {};
\node (c2) [tdnode,ellipse,draw=black,text height=5em,yshift=-2.3em,text width=7em,xshift=-.95em,color=white,right=5em of b2,label={[text height=1.5ex,yshift=0.00cm,xshift=0.00cm]below:$S'$}] {}; 
\node (s1z) [fill=black,right=30em of s1,label={[text height=1.5ex,yshift=0.00cm,xshift=0.35em]right:$s_1$}] {};
\node (s2z) [fill=black,below=1em of s1z,label={[text height=1.5ex,yshift=0.00cm,xshift=0.35em]right:$s_2$}] {};
\node (s3z) [fill=black,below=1em of s2z,label={[text height=1.5ex,yshift=0.00cm,xshift=0.35em]right:$s_3$}] {};
\node (s4z) [fill=black,below=1em of s3z,label={[text height=1.5ex,yshift=0.00cm,xshift=0.35em]right:$s_4$}] {};

\node (s12) [draw=black,right=5em of b2,yshift=-.8em,label={[text height=1.5ex,yshift=0.00cm,xshift=0.35em]right:$idx_1^1$}] {};
\node (s13) [draw=black,right=10.5em of b2,yshift=-.8em,label={[text height=1.5ex,yshift=0.00cm,xshift=0.35em]right:$idx_3^1$}] {};
\node (s22) [draw=black,below=1em of s12,label={[text height=1.5ex,yshift=0.00cm,xshift=0.35em]right:$idx_1^2$\; $\ldots$}] {};
\node (s23) [draw=black,below=1em of s13,label={[text height=1.5ex,yshift=0.00cm,xshift=0.35em]right:$idx_3^2$}] {};
\node (vs22) [draw=black,below=1em of s22,label={[text height=1.5ex,yshift=0.00cm,xshift=0.35em]right:$val_1$}] {};
\node (vsdots) [draw=black,below=1em of s22,label={[text height=1.5ex,yshift=0.00cm,xshift=0.35em]right:$val_1$}] {};
\node (vsdots2) [draw=black,below=1em of s23,label={[text height=1.5ex,yshift=0.00cm,xshift=0.35em]right:$val_3$}] {};
\node (sat) [draw=black,below=2em of s22,xshift=3em,label={[text height=1.5ex,yshift=0.00cm,xshift=0.35em]right:$sat$}] {};
\draw [dashed,in=20,out=140] (s13) to (sx);
\draw [dashed,in=20,out=140] (s23) to (sx);
\draw [dashed,in=20,out=140] (vsdots2) to (sx);
%
%
%
\draw [dashed,in=-35,out=250] (sat) to (f2);
\draw [dashed,in=-35,out=250] (sat) to (d2);
\draw [dashed,in=-25,out=250] (sat) to (e2);
\draw [dashed,in=15,out=140] (s12) to (y2s);
\draw [dashed,in=15,out=140] (s22) to (y2s);
\draw [dashed,in=15,out=140] (vs22) to (y2s);
\draw [dashed,in=15,out=140] (s12) to (y2);
\draw [dashed,in=15,out=140] (s22) to (y2);
\draw [dashed,in=15,out=140] (vs22) to (y2);
\draw [dashed,in=15,out=140] (s12) to (b2);
\draw [dashed,in=25,out=140] (s22) to (b2);
\draw [dashed,in=35,out=140] (vs22) to (b2);
\draw [dashed,in=-35,out=250] (s1z) to (c2);
\draw [dashed,in=-35,out=250] (s2z) to (c2);
\draw [dashed,in=-35,out=250] (s3z) to (c2);
\draw [dashed,in=-35,out=250] (s4z) to (c2);
\draw [dashed,in=-35,out=270] (s4z) to (sat);
%
%
\draw [dashed,in=-52,out=260] (s12) to (z2);
\draw [dashed,in=-60,out=260] (s22) to (z2);
\draw [dashed,in=-70,out=260] (vs22) to (z2);
\draw [dashed,in=-50,out=260] (s12) to (e2);
\draw [dashed,in=-65,out=260] (s22) to (e2);
\draw [dashed,in=-80,out=260] (vs22) to (e2);
\draw [dashed,in=80,out=-70] (s12) to (sat9);
\draw [dashed,in=80,out=-70] (s22) to (sat9);
\draw [dashed,in=80,out=-70] (vs22) to (sat9);
%
%
%
%
%
%
\end{tikzpicture}
\vspace{-1.25em}
\caption{Structure-aware reduction~$\mathcal{R}$ for some QBF~$Q$; dashed lines show potentially dense graph parts. (Left): A primal graph~$G_Q$ together with a sparse feedback vertex set~$S$ that connects~$G_Q$. (Right): Corresponding primal graph~$G_{Q'}$ (simplified) and the resulting sparse feedback vertex set~$S'$, where~$Q'$ and~$S'$ are obtained by~$\mathcal{R}(Q,S)$. The illustration depicts three different kind of clauses of~$\matr(Q)$:  type (i) using only variables in~$S$, like $sat_8$; type (ii): using two variables in~$S$, like $sat_1, sat_2$, and type (iii) using only one variable in~$S$, like $sat_3, sat_4, sat_5, sat_6, sat_7$.}\label{fig:sketch}
\end{figure*}
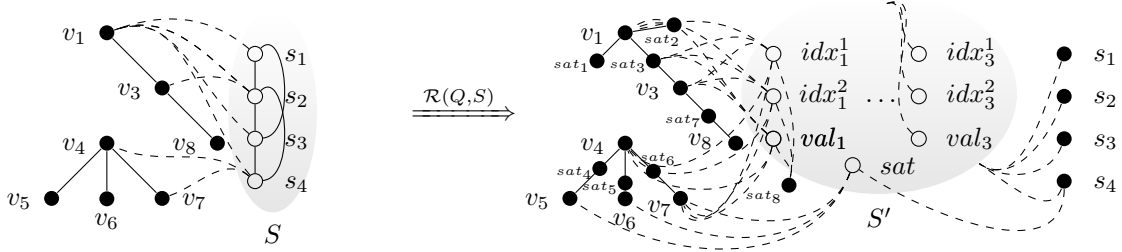

\paragraph*{Structure-Awareness}

Besides~$Q'$, reduction~$\mathcal{R}$ above further gives rise to the sparse feedback vertex set of~$Q'$ defined by~$S'\eqdef \lbvs{S} \cup \lbvvs{S} \cup \{sat\}$. 
Indeed, the size of~$S'$ compared to~$\Card{S}$ is exponentially smaller and therefore~$\mathcal{R}$ is indeed a structure-aware reduction.
The reduction and the relations between~$Q$ and~$Q'$, as well as~$S$ and~$S'$ are visualized in Figure~\ref{fig:sketch}.
Formally, we obtain the following result stating that~$\mathcal{R}$ is indeed a SAW reduction for sparse feedback vertex sets.

\begin{LEM}[Decrease Feedback Vertex Number,~$\star$]\label{lem:compr}
Given QBF~$Q$ in 3,1-CDNF and a sparse feedback vertex set~$S$ of $Q$, $\mathcal{R}$ constructs QBF~$Q'$ with sparse feedback vertex set~$S'$ of~$Q'$ such that~$\Card{S'}$ is in~$\mathcal{O}(\log(\Card{S}))$.
\end{LEM}

\paragraph*{Towards 3-DNF of~$C'$}
Observe that the formula~$C'$ generated by the reduction~$\mathcal{R}$ is almost in 3-DNF. The only formulas that are not already in the required format are Formulas~(\ref{red:guessatomj}) and~(\ref{red:guessnegatomj}).
It is easy to observe that, however, even those formulas can be transformed such that only at most 3 literals per term are used. To this end, one needs to introduce additional auxiliary variables (that are added to the innermost $\forall$ quantifier).
Indeed, a straight-forward transformation recursively splits Formulas~(\ref{red:guessatomj}) and~(\ref{red:guessnegatomj}) into two terms, where the first term consists of two literals and a new auxiliary variable~$v$ that is added (positively), and the second term consists of the remaining literals of the term and~$\neg v$.
In turn, each term has at most two new auxiliary variables and one only needs to take care that the resulting term that contains~$x\in S$ or~$\neg x\in S$ does not use two of these auxiliary variables (preventing cycles in the resulting primal graph).

\paragraph*{Runtime and Correctness}
Next, we show runtime and correctness of~$\mathcal{R}$, followed by main~results.

\begin{THM}[Runtime, $\star$]\label{lab:runtime}
For a QBF~$Q$ in 3,1-CDNF with $\matr(Q)=C\wedge D$ and set~$S\subseteq\var(Q)$ of variables of~$Q$, $\mathcal{R}$ runs in time~$\mathcal{O}(\ceil{\log(\Card{S}+1)} \cdot (\Card{S} + \Card{C}) + \Card{D})$.
\end{THM}

\begin{THM}[Correctness, $\star$]\label{lab:corr}
  Given a QBF~$Q$ in 3,1-CDNF and a set~$S\subseteq\var(Q)$ of variables of~$Q$, 
  reduction $\mathcal{R}$ computes an instance~$Q'$ that is equivalent to $Q$. In fact, any assignment~$\alpha$ to variables of~$\matr(Q)$ satisfies~$\matr(Q)$ iff every extension~$\alpha'$ of~$\alpha$ to variables~$\lbvs{S} \cup \lbvvs{S} \cup \lsat$ satisfies~$\matr(Q')$.
\end{THM}

QSAT is well known to be polynomial-time tractable when restricted to 2-CNF formulas~\cite{AspvallPT79}.
As an application of our reduction~$\mathcal{R}$, we now observe that allowing a single clause of arbitrary length already leads to intractability.
\begin{COR}[$\star$]\label{cor:hardnesslongclause}
Problem~$\QBFSAT$ over a QBF $Q{=}Q_1 V_1$ $\cdots Q_\ell V_\ell. C \wedge D$ of quantifier depth~$\ell\geq 2$ with~$Q_\ell=\exists$, $C$ being in 2-CNF, and~$D$ being in 1-DNF, is~$\SIGMA{{\ell-1}}{\Ptime}$-complete (if $Q_1{\,=\,}\exists$, $\ell$ odd) and~$\PI{{\ell-1}}{\Ptime}$-complete (if $Q_1{\,=\,}\forall$, $\ell$ even).
\end{COR}

\noindent A similar result can be obtained with long terms when the innermost quantifier is universal.


\paragraph*{Lower Bound Result}
Having established 
structure-awareness, runtime, as well as correctness of the reduction~$\mathcal{R}$ above, we proceed with the lower bound results.

\begin{THM}[LB for Sparse Feedback Vertex Set,~$\star$]\label{lab:lb}
Given an arbitrary QBF~$Q$ in CDNF of quantifier depth~$\ell$ and a minimum sparse feedback vertex set~$S$ of~$Q$ with~$k=\Card{S}$. 
Then, under ETH, $\QBFSAT_\ell$ on~$Q$ cannot be decided in time~$\tower(\ell, o(k))\cdot\poly(\Card{\var(Q)})$.
\end{THM}

As a consequence, we obtain the following result. 

\begin{COR}[LB for Incidence Feedback Vertex Set,~$\star$]\label{lab:inc-lb}
Given an arbitrary QBF~$Q$ with~$F=\matr(Q)$ in CNF (DNF) such that the innermost quantifier~$Q_\ell$ of~$Q$ is $Q_\ell=\exists$ ($Q_\ell=\forall$) 
with the feedback vertex number of~$I_{F}$ being~$k$.
Then, under ETH, $\QBFSAT_\ell$ on~$Q$ cannot be decided in time~$\tower(\ell, o(k))\cdot\poly(\Card{\var(Q)})$.
\end{COR}

The lower bound for the (sparse) feedback vertex number carries over to even more restrictive parameters, see the appendix. 
Further, the lower bound holds still holds when restricting feedback vertex sets to variables of the innermost quantifier block, see Corollary~\ref{cor:lastquantifier}.

\subsection{Hardness Insights \& New Lower Bounds for Treedepth}\label{sec:tdpth}
So far, we discussed in Section~\ref{sec:result}  why Corollary~\ref{ref:ubs} cannot be significantly improved for feedback vertex number.
In this section, we provide a hardness result in the form
of a conditional lower bound for the parameter treedepth. 
Notably, our approach for treedepth also involves a SAW reduction, 
where it turns out that we can even reuse major parts of reduction~$\mathcal{R}$ as defined by Formulas~(\ref{red:guessatomj})--(\ref{red:usatnotd}). 
Let~$Q=\exists V_1. \forall V_2. \cdots \exists V_\ell. C \wedge D$ be a QBF in CDNF and~$T$ be a treedepth decomposition of~$G_Q$ 
that consists of a path~$S$ of height~$h$, where each element of the path might be connected to a tree of constant height.
The result of applying~$\mathcal{R}$ on~$Q$ and~$S$ 
is visualized in Figure~\ref{fig:sketchtd}.
%
%
%
%
%
%
The final normalization step (from DNF to 3-DNF) results in multiple paths of length~$\mathcal{O}(\log(h))$ due to additional auxiliary variables when normalizing Formulas~(\ref{red:guessatomj2}) and~(\ref{red:guessnegatomj2}).
These paths do not increase the size of a sparse feedback vertex set and could previously be ignored.
But reducing the treedepth requires compressing \emph{each} of these paths, and so we must turn the reduction $\mathcal{R}$ that takes a single set~$S$ as an argument into a SAW reduction $\mathcal{R}_\tdx$ dealing with multiple paths simultaneously.
Formal details of $\mathcal{R}_\tdx$ are given in the appendix. 

%
%

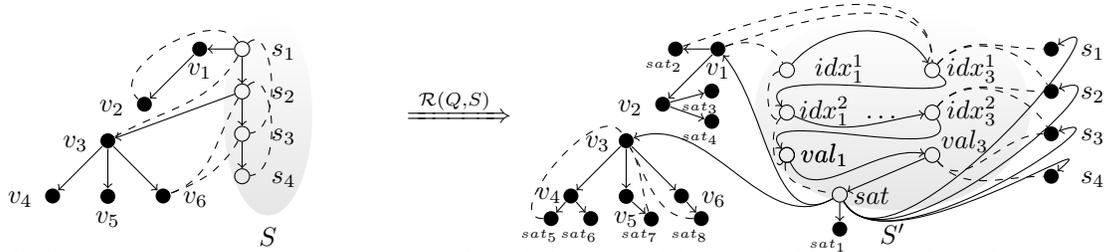
\begin{figure*}[t]\centering
\vspace{-3em}
\begin{tikzpicture}[node distance=7mm,every node/.style={circle,inner sep=2pt}]
\node (c) [tdnode,ellipse,draw=black,text height=5em,yshift=-3em,text width=2em,xshift=-.75em,color=white,label={[text height=1.5ex,yshift=0.00cm,xshift=0.00cm]below:$S$}] {}; 
\node (s1) [draw=black,right=-3em of c,yshift=2.5em,label={[text height=1.5ex,yshift=0.00cm,xshift=0.35em]right:$s_1$}] {};
\node (s2) [draw=black,below=1em of s1,label={[text height=1.5ex,yshift=0.00cm,xshift=0.35em]right:$s_2$}] {};
\node (s3) [draw=black,below=1em of s2,label={[text height=1.5ex,yshift=0.00cm,xshift=0.35em]right:$s_3$}] {};
\node (s4) [draw=black,below=1em of s3,label={[text height=1.5ex,yshift=0.00cm,xshift=0.35em]right:$s_4$}] {};
\draw [->] (s1) to (s2);
\draw [in=5,dashed] (s1) to (s3);
\draw [in=5,dashed] (s2) to (s4);
\draw [->](s2) to (s3);
\draw[->] (s3) to (s4);
\node (b) [fill,left=1em of s1, label={[text height=1.5ex,yshift=0.5em]below:$v_1$}] {};
\node (y) [fill,below left=2.35em of b,label={[text height=1.5ex]left:$v_2$}] {};
\node (z) [fill,below left=1.35em of y,label={[text height=1.5ex]left:$v_3$}] {};
\node (f) [fill,below left=2.35em of z,label={[text height=1.5ex]left:$v_4$}] {};
\node (d) [fill,below=1.5em of z,label={[text height=1.5ex,yshift=.5em]below:$v_5$}] {};
\node (e) [fill,below right=2.35em of z,label={[text height=1.5ex]right:$v_6$}] {};
%
%
\draw [<-] (y) to (b);

\draw [->](z) to (e);
\draw [->](z) to (f);
\draw [<-] (d) to (z);
%

%
\draw [dashed,in=140,out=120] (s1) to (y);
\draw [->] (s1) to (b);
\draw [->] (s2) to (z);
\draw [dashed,in=20,out=-120] (s3) to (e);
\draw [dashed,in=20,out=-120] (s2) to (e);
\draw [dashed,in=30,out=-120] (s1) to (z);
%
\node (arr) [draw=white,right=3em of c,label={}] {$\xRightarrow{\;\mathcal{R}(Q,S)\;}$};
%
%
%
%
%
\node (b2) [fill,right=19em of b, label={[text height=1.5ex,yshift=0.5em]below:$v_1$}] {};
%
%
\node (sb1) [fill,left=1em of b2, label={[text height=1.5ex,yshift=.95em,xshift=-.45em]below:{\tiny $sat_2$}}] {};
\draw [->] (b2) to (sb1);
\node (y2) [fill,below left=2.35em of b2,label={[text height=1.5ex]left:$v_2$}] {};
\node (sb3) [fill,right=1.25em of y2,yshift=.5em, label={[text height=1.5ex,yshift=.95em,xshift=-.45em]below:{\tiny $sat_3$}}] {};
\draw [->] (y2) to (sb3);
\node (sb4) [fill,right=1.25em of y2, yshift=-.65em, label={[text height=1.5ex,yshift=.95em,xshift=-.45em]below:{\tiny $sat_4$}}] {};
\draw [->] (y2) to (sb4);
\node (z2) [fill,below left=1.35em of y2,label={[text height=1.5ex,xshift=.15em]left:$v_3$}] {};
\node (f2) [fill,below left=2.35em of z2,label={[text height=1.5ex,xshift=.4em,yshift=.15em]left:$v_4$}] {};
\node (d2) [fill,below=1.5em of z2,label={[text height=1.5ex,yshift=.5em]below:$v_5$}] {};
\node (e2) [fill,below right=2.35em of z2,label={[text height=1.5ex]right:$v_6$}] {};
\node (sb5) [fill,below=0.25em of f2,xshift=-.75em, label={[text height=1.5ex,yshift=.95em,xshift=-.45em]below:{\tiny $sat_5$}}] {};
\node (sb6) [fill,below=0.25em of f2,xshift=.75em, label={[text height=1.5ex,yshift=.95em,xshift=-.45em]below:{\tiny $sat_6$}}] {};
\draw [->] (f2) to (sb5);
\draw [->] (f2) to (sb6);
\node (sb7) [fill,below=0.25em of d2,xshift=.95em, label={[text height=1.5ex,yshift=.95em,xshift=-.45em]below:{\tiny $sat_7$}}] {};
\node (sb8) [fill,below=0.25em of e2,xshift=.75em, label={[text height=1.5ex,yshift=.95em,xshift=-.45em]below:{\tiny $sat_8$}}] {};
\draw [->] (d2) to (sb7);
\draw [->] (e2) to (sb8);
\draw [dashed, in=120, out=-180] (sb5) to (z2);
\draw [dashed, in=-60, out=110] (sb7) to (z2);
\draw [dashed, in=-60, out=-180] (sb8) to (z2);
%
\draw [<-] (y2) to (b2);

\draw [->](z2) to (e2);
\draw [->](z2) to (f2);
\draw [<-] (d2) to (z2);
\node (c2) [tdnode,ellipse,draw=black,text height=5em,yshift=-2.3em,text width=7em,xshift=-.95em,color=white,right=2em of b2,label={[text height=1.5ex,yshift=0.00cm,xshift=0.00cm]below:$S'$}] {}; 
\node (s1z) [fill=black,right=30em of s1,label={[text height=1.5ex,yshift=0.00cm,xshift=0.35em]right:$s_1$}] {};
\node (s2z) [fill=black,below=1em of s1z,label={[text height=1.5ex,yshift=0.00cm,xshift=0.35em]right:$s_2$}] {};
\node (s3z) [fill=black,below=1em of s2z,label={[text height=1.5ex,yshift=0.00cm,xshift=0.35em]right:$s_3$}] {};
\node (s4z) [fill=black,below=1em of s3z,label={[text height=1.5ex,yshift=0.00cm,xshift=0.35em]right:$s_4$}] {};

\node (s12) [draw=black,right=2em of b2,yshift=-.8em,label={[text height=1.5ex,yshift=0.00cm,xshift=0.35em]right:$idx_1^1$}] {};
\draw [dashed, in=120, out=20] (b2) to (s12);
\node (s13) [draw=black,right=7.5em of b2,yshift=-.8em,label={[text height=1.5ex,yshift=0.00cm,xshift=-0.15em]right:$idx_3^1$}] {};
\draw [dashed, in=90, out=20] (b2) to (s13);
\draw [dashed, in=90, out=20] (sb1) to (s13);
\draw [dashed, in=-180, out=50] (s13) to (s1z);
\draw [dashed, in=120, out=50] (s13) to (s2z);
\node (s22) [draw=black,below=1em of s12,label={[text height=1.5ex,yshift=0.00cm,xshift=-0.15em]right:$idx_1^2$\; $\ldots$}] {};
\draw [dashed, in=-140, out=-170] (s12) to (s22);
\node (s23) [draw=black,below=1em of s13,label={[text height=1.5ex,yshift=0.00cm,xshift=-0.15em]right:$idx_3^2$}] {};
\node (vs22) [draw=black,below=1em of s22,label={[text height=1.5ex,yshift=0.00cm,xshift=-0.15em]right:$val_1$}] {};
\node (vsdots) [draw=black,below=1em of s22,label={[text height=1.5ex,yshift=0.00cm,xshift=-0.15em]right:$val_1$}] {};
\node (vsdots2) [draw=black,below=1em of s23,label={[text height=1.5ex,yshift=0.00cm,xshift=-0.35em,yshift=.5em]right:$val_3$}] {};
\node (sat) [draw=black,below=2.5em of s22,xshift=2em,label={[text height=1.5ex,yshift=0.00cm,xshift=-0.15em]right:$sat$}] {};
\draw [dashed, in=-120, out=-140] (s22) to (vs22);
\draw [dashed, in=170, out=-100] (vs22) to (sat);
\draw [dashed, in=120, out=50] (s23) to (s3z);
\draw [dashed, in=180, out=50] (s23) to (s2z);
\draw [dashed, in=180, out=-50] (vsdots2) to (s3z);
\draw [dashed, in=180, out=-50] (vsdots2) to (s4z);
\node (sb2) [fill,below=1em of sat, yshift=.3em,label={[text height=1.5ex,yshift=.95em,xshift=-.45em]below:{\tiny $sat_1$}}] {};
\draw [->] (sat) to (sb2);
\draw [->,out=50,in=120] (s12) to (s13);
\draw [->,out=-35,in=130] (s13) to (s22);
\draw [->,out=-25,in=-170] (s22) to (s23);
\draw [->,out=-55,in=130] (s23) to (vs22);
\draw [->,out=-65,in=170] (vs22) to (vsdots2);
\draw [->] (vsdots2) to (sat);
\draw [->,out=-150,in=-50] (sat) to (b2);
\draw [->,out=-150,in=20] (sat) to (z2);
%

\draw [->,out=-60,in=33] (sat) to (s1z);
\draw [->,out=-60,in=20] (sat) to (s2z);
\draw [->,out=-60,in=20] (sat) to (s3z);
\draw [->,out=-60,in=40] (sat) to (s4z);
\end{tikzpicture}
\vspace{-2.2em}
\caption{Visualization of structure-aware reduction~$\mathcal{R}$ for some QBF~$Q$. (Left): A primal graph~$G_Q$ aligned in a treedepth decomposition~$T$ of depth~$h$, where~$S$ is a path in~$T$ of height~$\mathcal{O}(h)$; dashed lines show potential edges between ancestors and descendants. (Right): The corresponding primal graph~$G_{Q'}$ (simplified), aligned in a treedepth decomposition~$T'$ and a corresponding path~$S'$ in~$T'$ of height~$\mathcal{O}(\log(h))$; $Q'$ and~$S'$ are obtained by~$\mathcal{R}(Q,S)$. 
}\label{fig:sketchtd}
\end{figure*}

Using~$\mathcal{R}_\tdx$, we obtain the following lower bound for treedepth~$k$, which yields an ETH-tight lower bound if $k\in{\mathcal{O}}(\ell)$ for quantifier depth~$\ell$, see Corollary~\ref{ref:ubs}. 

\begin{THM}[LB for Treedepth Decompositions,~$\star$]\label{lab:lb_pdp}
Given an arbitrary QBF~$Q$ in CDNF of quantifier depth~$\ell$ and a treedepth decomposition~$T$ of~$G_Q$ of height~$k=\pd{G_Q}$. 
Then, under ETH, $\QBFSAT_\ell$ on~$Q$ cannot be decided in time~$\tower(\ell, o(k-\ell))\cdot\poly(\Card{\var(Q)})$.
\end{THM}

\noindent We show that this result still implies a hierarchy of runtimes under ETH, where
the tower height depends linearly on the quantifier depth of the QBF. 

\begin{COR}[LB for Treedepth, $\star$]\label{lab:lb_tdp}
There is a linear function~$f$ and an integer $\ell_0>0$ where for
any QBF $Q$ in 3,1-CDNF of quantifier depth~$\ell\geq \ell_0$: 
Under ETH, 
$\QBFSAT_\ell$ on $Q$ cannot be decided in time~$\tower(f(\ell), o(\td{G_Q}))\cdot\poly(\Card{\var(Q)})$.
\end{COR}

\begin{COR}[LB for Incidence Treedepth, $\star$]\label{lab:inc-tdp}
Given a QBF~$Q$ with~$F=\matr(Q)$ in CNF (DNF) such that the innermost quantifier~$Q_\ell$ of~$Q$ is $Q_\ell=\exists$ ($Q_\ell=\forall$) 
and~$k=\td{I_Q}$. 
Then, under ETH, $\QBFSAT_\ell$ on~$Q$ cannot be decided in time~$\tower(\ell, o(k-\ell))\cdot\poly(\Card{\var(Q)})$.
\end{COR}
%

\noindent Note that this corollary immediately yields a result similar to Corollary~\ref{lab:lb_tdp} for the incidence~graph. 
%

\section{Algorithms 
Using Vertex Cover and FES}
\label{sec:vcfen-inc}

Our results from the previous section already provide a rather
comprehensive picture of the (fine-grained) parameterized complexity
of \QBFSATCNF{}, when considering many of the most prominent
structural parameters on the incidence graph. In particular, they rule
our fixed-parameter tractability of treedepth and feedback vertex
set. In this section, we will complement this picture for the
incidence graph by giving
fpt-algorithms for \QBFSATCNF{} parameterized by the vertex cover
number as well as the feedback edge set number. We start with our
algorithm for the vertex cover number, which essentially follows from
the simple observation that formulas with a small vertex cover number
cannot have too many distinct clauses together with the 
well-known result that \QBFSATCNF{} is fpt
parameterized by the number of clauses~\cite{Williams02}.
\begin{THM}[$\star$]\label{thm:incvco}
  Given any CQBF~$Q$ of $\QBFSAT$ with~$k$ being the vertex cover
  number of~${I_Q}$. Then, the validity of~$Q$ can be decided in
  time~$1.709^{3^k}\cdot\poly(\Card{\var(Q)})$.
\end{THM}
Note that tractability for the vertex cover number of the incidence
graph does not immediately carry over to the primal graph and
therefore neither Proposition~\ref{prop:vco} nor Theorem~\ref{thm:cnfvcn} are a direct consequence of
Theorem~\ref{thm:incvco}; indeed a small vertex cover number of the
primal graph still allows for an arbitrary number of distinct clauses.

We are now ready to provide our algorithm for the feedback edge number
of the incidence graph. Interestingly and in contrast to vertex cover
number, the parameterized complexity of \QBFSATCNF{} for the feedback
edge number has been open even for the primal graph. While the FEN of
the primal graph and the incidence graph are again orthogonal
parameters (consider, e.g., two variables that occur together in more
than one clause), we will show that the algorithm for the incidence
graph can essentially be obtained using the techniques developed for
the primal graph. We will therefore start by giving our result for
the FEN of the primal graph, which also constitutes the main technical
contribution of this section. We establish the result by proving 
existance of a kernelization algorithm.
\vspace{-.35em}
\begin{THM}[$\star$]\label{thm:kernel}
  Let $\qbfformula$ be a $\QBFCNF$. In polynomial time, we can construct an
  equivalent $\QBFCNF$ with at most $12k-8$ variables and at most
  $10k-9+3^{\lfloor (\sqrt{24k+1}+1)/2\rfloor}$ clauses, where $k$ is the feedback edge number of $G_\qbfformula$.
\end{THM}\vspace{-.15em}
\ifshort
The main ideas behind the kernelization are as
follows. Given $\qbfformula$, we first compute a smallest FES $D$ of primal
graph $G{=}G_\qbfformula$ in polynomial time. Then, graph $H{=}G{-}D$ is a (spanning) forest of
$G$. We introduce a series of reduction rules that allow us to
reduce the size of $\qbfformula$ and $H$. We start by observing that
we can remove \emph{unit clauses}, i.e., clauses containing only one
literal, and \emph{pure literals}, i.e., variables that either only
occur positively or only negatively in $\qbfformula$. We then consider
\emph{clean} edges of $H$, i.e., edges that do not appear in any
triangle of $G$. Note that all but at most $2\cdot|D|$ edges of $H$
are clean, because every edge of $H$ that is not incident to any edge in $D$
is necessarily clean. Crucially, endpoints of a clean edge can only
occur together in clauses of size at most 2. This property allows us
to simplify formula $\qbfformula$ significantly (using three reduction rules,
whose correctness follows by using Hintikka strategies~\cite{Gradel05})
s.t.\ we can assume the endpoints of every clean
edge are contained in exactly one clause of $\qbfformula$. This
allows us to introduce a simple reduction rule for removing 
every leaf of $H$ that is not an endpoint of an edge in~$D$. 
Then, in the reduced instance, $H$ has at
most $2\cdot|D|$ leaves and therefore at most $2\cdot|D|-2$ vertices of
\emph{degree} (number of adjacent vertices) larger than 2. Our last reduction rule, which is involved and based on Hintikka strategies, allows us to reduce
degree 2 vertices in $H$ by showing that any maximal (clean) path of
degree 2 vertices in $H$ must contain a variable (the innermost variable), that can be ``removed''. This shows that the size of $H$ and therefore also $G$ is bounded,
which in turn allows us to obtain a bound for~$\qbfformula$.
\fi

From Theorem~\ref{thm:kernel}, we know that we can brute-force on the
kernel (output) after preprocessing. Immediately, we obtain a
single-exponential fpt algorithm for $\QBFSATCNF$.
\vspace{-1.4em}
\begin{COR}\label{cor:fpt:fes}
  $\QBFSATCNF$ is fpt parameterized by the feedback edge number of the primal graph.
\end{COR}\vspace{-.15em}
Interestingly, our kernelization even provides a linear kernel,
i.e., the size of the kernel depends only linearly on the parameter,
if we restrict ourselves to $c$-$\QBFCNF$s.
\vspace{-.3em}
\begin{COR}
  Let $c$ be an integer; $\qbfformula$ be a $c$-$\QBFCNF$.~In polynomial time, we obtain an 
  equivalent $c$-$\QBFCNF$ with at most $12k{\,-\,}8$ variables and 
  $10k{\,-\,}9{\,+\,}3k$ $(3^c/\binom{c}{2})$ clauses, s.t.\ $k$ is the size of a smallest
  FES of~$G_\qbfformula$.
\end{COR}\vspace{-.1em}
Similarly, we can show the existance of a smaller kernel for
the feedback edge number of the incidence~graph.
\vspace{-.2em}
\begin{THM}[$\star$]\label{thm:kernel-inc}
  Let $\qbfformula$ be a $\QBFCNF$ with  feedback edge number~$k$ of $I_\qbfformula$. In polynomial time, we can construct an
  equivalent $\QBFCNF$ with $\leq 24k{-}17$ variables and clauses.
\end{THM}

\vspace{-.7em}
\section{Tractabilty 
for \QBFSATCNF 
on Primal Graphs}\label{sec:tdp}
\vspace{-.1em}

Above, our results draw a comprehensive picture of the
fine-grained complexity of \QBFSATCNF{} with respect
to the incidence graph. 
However, when
considering the primal graph there is a gap between the
tractability for vertex cover number and feedback edge number and the known intractability for
treewidth. 
To address this, one may ask what is the complexity of
\QBFSATCNF{} with respect to parameters feedback vertex number and
treedepth? 
We progress towards %
resolving the question for treedepth, 
which 
%
is not only completely open, but existing techniques do not even allow us to solve the problem for significant restrictions of treedepth.  Such a parameter is
$c$-deletion set, i.e.,~deletion distance to components of size at most $c$, which is well-known to be 
inbetween vertex cover number and treedepth, see related work on \emph{vertex integrity}~\cite{LampisM21}. 
We provide three novel algorithms,
each representing a step towards generalizing the tractability of
\QBFSATCNF\ for vertex cover number. Each 
uses a
different 
approach providing new insights 
that are promising for 
treedepth. 


First, we show that 
different variants of $c$-deletion sets can be efficiently computed, which we achieve by
the following proposition.
\newcommand{\mP}{\mathcal{P}}
Let $\mP(\qbfformula,D,c)$ be any property that can  be true or false for
a \QBFCNF{} $\qbfformula$ and $c$-deletion set $D$ of~$Q$. We say
$\mP$ is \emph{efficiently computable} if there is an algorithm that given
$\qbfformula$, $D$, and $c$ decides whether $\mP(\qbfformula,D,c)$ holds in fpt-time
parameterized by $|D|{\,+\,}c$.
\iflong
\begin{PROP}
\fi
\ifshort\vspace{-.2em}
\begin{PROP}[$\star$]
\fi  
  \label{pro:computedelsetgeneral}
  Let $\mP$ be any efficiently computable property and let $\qbfformula$ be a
  \QBFCNF{}. Then, computing a smallest $c$-deletion set $D$
  of $\qbfformula$ that satisfies $\mP(\qbfformula,D,c)$ is fixed-parameter tractable
  parameterized by $|D|+c$.
\end{PROP}

We can eliminate all universal variables in a
$c$-deletion set $D$ of a $\QBFCNF$ without losing the
structure of the formula, i.e., 
we obtain a formula, which is not too large and still has a $2^cc$-deletion set of
size at most $2^cc$.
\iflong
\begin{PROP}
  \fi
\ifshort\vspace{-.4em}
\begin{PROP}[$\star$]
\fi  
  \label{pro:delset-exists}
  Let $\qbfformula$ be a $\QBFCNF$ and let $D$ be a $c$-deletion set for $\qbfformula$. Then, in time
  $\mathcal{O}(2^u\CCard{\qbfformula})$, where $u=|D\cap \varu(\qbfformula)|$, we can construct an equivalent $\QBFCNF$
  $\qbfformula'$ and a set $D'\subseteq \vare(\qbfformula')$ with $|D'|\leq 2^u|D|$
  s.t.\ $D'$ is a $2^uc$-deletion set for~$\qbfformula'$.
\end{PROP}

\vspace{-.8em}
\subsection{Components of Type $\exists^{\leq 1}\forall$}

$\QBFSATCNF$ is fpt
parameterized by the size of a $c$-deletion set into components of the
form $\exists^{\leq 1}\forall$, i.e., components have at most
one existential variable occurring before all its (arbitrarily many)
universal
variables in the prefix. 
\vspace{-.2em}
\begin{THM}\label{the:fpt:cdels}
  \QBFSATCNF{} is fixed-parameter tractable parameterized by $k+c$,
  where $k$ is the size of a smallest $c$-deletion set into components
  of the form $\exists^{\leq 1}\forall$.
\end{THM}
This generalizes fixed-parameter tractability of
$\QBFSATCNF$ parameterized by vertex cover, since
every component can have arbitrary many variables as well as
one quantifier alternation, as opposed to containing only one variable. 


Checking whether every component is of the form
$\exists^{\leq 1}\forall$ can be achieved in polynomial time.
Since we can compute a smallest $c$-deletion set into components of
the form $\exists^{\leq 1}\forall$ in fpt-time parameterized by its
size plus $c$ due to Proposition~\ref{pro:computedelsetgeneral}, it suffices to show the
following. 
\vspace{-1.4em}
\begin{THM}[$\star$]\label{thm:delsetformfpt}
  Let $\qbfformula$ be a $\QBFCNF$ and $D \subseteq \var(\qbfformula)$ be a
  $c$-deletion set for $\qbfformula$ into components of the
  form $\exists^{\leq 1}\forall$. Then, deciding $\qbfformula$  is fpt 
parameterized by $|D|+c$.
\end{THM}\vspace{-.35em}
\ifshort
The main ingredient for the proof of Theorem~\ref{thm:delsetformfpt}
is Lemma~\ref{lem:delset-ef-red} ($\star$). It allows us to remove all but at most~$2^c$ components of every component type. 
Together with bounding the number of component types,~we  
reduce $\qbfformula$ to a bounded-size formula that we~brute-force.
\fi


\vspace{-.35em}
\subsection{Single-Variable Deletion Sets}
Next, we consider deletion sets consisting of only a single variable~$e$, but where the quantifier prefix restricted to variables occurring in a component can have an arbitrary shape.
Without loss of generality, we may assume~$e$ to be existentially quantified and innermost.

The argument involves an evaluation game where a universal player and an existential player take turns assigning their respective variables, in the order of the quantifier prefix.
The universal player tries to assign so that no assignment of $e$ is left for the existential player to satisfy all clauses.
If universal has a strategy ensuring some assignment of $e$ cannot be played by existential, we say the strategy \emph{forbids} this assignment.
The QBF is false if and only if universal can forbid both assignments.

Since components do not share universal variables, the universal strategy can be decomposed into strategies played in the individual components.
If there are distinct components where universal can forbid assignments $e \mapsto 0$ and $e \mapsto 1$, respectively, corresponding strategies can be composed into a universal winning strategy.
An interesting case arises when there is a single component where universal must choose an assignment to forbid, and existential must similarly choose which assignment to play in the remaining components.
Universal wins if and only if the latest point (that is, the innermost variable) where they can choose which assignment to forbid
comes after the latest point where existential can choose which assignment to play.
A formal development of these intuitions leads to the following. 
\vspace{-.2em}
\begin{THM}[$\star$]\label{thm:1cdeletionset}
  Let $\qbfformula$ be a $\QBFCNF$ with a $c$-deletion set of size $1$. Deciding $\qbfformula$ is fpt parameterized by $c$.
\end{THM}

\vspace{-.65em}
\subsection{Formulas with Many Components of Each Type}

Let $\qbfformula$ be a $\QBFCNF$
and let $D\subseteq \var(\qbfformula)$ be a $c$-deletion set for
$\qbfformula$. Moreover, let $\qbfformula'$ and $D'$ be the $\QBFCNF$ and
$2^cc$-deletion set of $\qbfformula'$ of size at most $2^c|D|$ obtained after
eliminating all universal variables in $D$ using
Proposition~\ref{pro:delset-exists}. We say $D'$ is
\emph{universally complete} if every component type of $\qbfformula'$
consists of at least $2^{|D'|}$ components. We show 
\QBFSATCNF{} to be fpt parameterized by the size of
a smallest universally complete $c$-deletion~set.
\begin{THM}\label{thm:qbfcnfunifpt}
  \QBFSATCNF{} is fpt by $k+c$,
  where $k$ is the size of a smallest universally complete $c$-deletion set.
\end{THM}
This result is surprising and counter-intuitive at first, as it seems to
indicate that deciding a large \QBFCNF{} $\qbfformula$ (with many
components of each type) is simple, while we do not know whether this holds for sub-formulas of $\qbfformula$. However, we 
show that many components of the same type
allows the universal player to play all possible local
counter-strategies for each type. This makes the
relative ordering of various components in the prefix irrelevant.

Note that deciding whether a given $c$-deletion set is universally
complete is an efficiently computable property due to
Proposition~\ref{pro:delset-exists}. This implies that we can compute
a smallest universally complete $c$-deletion set in fpt-time parameterized by its
size plus $c$ by 
Proposition~\ref{pro:computedelsetgeneral}. So, to show
Theorem~\ref{thm:qbfcnfunifpt}, we assume we are given a
smallest universally complete $c$-deletion set $D$ for $\qbfformula$. Moreover,
by Proposition~\ref{pro:delset-exists}, we can assume
we are given the corresponding $\qbfformula'$ and $D'$. 
It suffices to show:
\vspace{-.15em}
\begin{THM}[$\star$]\label{thm:qbfcnfunifptex}
  Let $\qbfformula$ be a $\QBFCNF$ and let $D \subseteq \vare(\qbfformula)$ be a
  universally complete $c$-deletion set for $\qbfformula$. Then deciding  $\qbfformula$ is fpt 
parameterized by $|D|+c$.
\end{THM}
\ifshort

\fi

\section{Conclusion}

We consider 
evaluating quantified Boolean
formulas (\QBFSAT)
%
under structural restrictions. 
%
While the classical complexity and 
the 
parameters 
treewidth and vertex cover number are well understood on primal graphs, 
we address 
the incidence graph and the gap between both parameters. 
%
We provide new upper and lower bounds and establish a comprehensive complexity-theoretic picture for \QBFSAT\ concerning the most fundamental graph-structural parameters of this graph. 
We thereby sharpen the boundaries between
parameters where one can drop the quantifier depth in the
parameterization and those where one cannot, providing a nearly-complete picture of
parameters of the incidence graph, cf., Figure~\ref{fig:params-in}.
We show lower bounds for feedback vertex number and treedepth by 
designing 
structure-aware (SAW) reductions. 
We then complement known upper bounds for vertex cover number
by tractability (fpt) results for
feedback edge~number.

Despite this paper closing many gaps and providing deeper insights into the hardness of \QBFSAT
for structural parameters, it does not fully settle 
\QBFSAT for 
parameters 
of 
the primal graph.
A single clause makes the difference: if we omit edges induced by one clause, our lower bounds for feedback vertex number and treedepth carry over, as indicated in Figure~\ref{fig:saw} (left).
As a first step towards~filling this gap and analyzing treedepth of the primal graph, 
we establish fpt for variants of the deletion set parameter. 
%

\textit{Techniques:}
%
While the ideas behind structure-aware (SAW) reductions have been implicitly used in limited  contexts,~e.g.,~\cite{FichteHecherPfandler20}, we formalize and fully develop the technique. 
%
We establish a template to design specific self-reductions from \QBFSAT to \QBFSAT when we are
interested in precise lower bounds under the exponential time hypothesis (ETH) for various
parameters. 
%
As illustrated in Figure~\ref{fig:saw} (right), SAW reductions allow us to trade an exponential
decrease 
of one parameter (structure of the matrix) 
for an
exponential increase (increasing the tower height) of runtime
dependency on a second parameter. 
%
%
%
%
%


%

\textit{Future Work:}
Our analysis opens up several interesting questions: 
Is \QBFSAT on CNFs fpt parameterized by either
the feedback vertex number or the treedepth of the primal graph?
We have indications for both possible outcomes.
(N) 
Our lower bounds are close since allowing merely one additional clause yields intractability for treedepth and feedback vertex number. 
%
%
This points toward hardness,
and our SAW reductions might provide a good starting point to understand and obtain intractability for
one or both parameters on the primal graph.
(Y) 
The algorithmic techniques we developed for variants of
the $c$-deletion set, point in the other direction and may serve as the underpinning for an fpt result.
%
%
%
%
While our lower bounds for feedback vertex number of the incidence
graph are tight under ETH, this is open 
for 
treedepth. 
Besides, lower bounds under SETH might be interesting. 
A further question is whether our techniques carry
over to 
DQBF~\shortversion{\cite{AzharPR01}}\longversion{\cite{AzharPR01,GanianEtAl20}}
and QCSP~\cite{FichteHecherKieler20}.

Finally, we expect 
SAW~reductions to be a useful tool
for problems 
within the polynomial hierarchy.
Indeed, many problems of practical interest would benefit
from precise bounds
; see, e.g.,~\cite{PanVardi06,FichteHecherPfandler20}. 
Some 
 are highly relevant 
in
other
communities, e.g., explainability~\cite{Darwiche20}.

\bibliography{qbf_param23_short}

\clearpage

\shortversion{\clearpage
\section*{Appendix}

\section{Additional Preliminaries}

For the ease of algorithm presentation, we 
equivalently view QBFs in a different (atomic) form
$\qbfformula=Q_{1} v_1 Q_2 v_2\cdots Q_{n} v_{n}.F$, where~$v_j$
for~$1\leq j \leq n$ are the variables of~$F$ and we
\emph{do not} have~$Q_i \neq Q_{i-1}$ for~$2\leq i\leq n$.

\paragraph{Hintikka Games}
Let $\qbfformula=Q_{1} v_1 Q_2v_2 \cdots Q_n v_n.F$ be a
QBF. Let $V$ be a set of variables.
We denote by $V^\exists$ and $V^\forall$ the set of all existential
respectively universal variables in $V$. For a variable $v \in
\var(\qbfformula)$, we denote by $V_{<v}$ ($V_{>v}$) the set of all
variables in $V$ that appear before (after) $v$ in the prefix of
$\qbfformula$. Similarly, for an assignment $\delta : V' \rightarrow
\{0,1\}$, where $V' \subseteq \var(\qbfformula)$, we denote by $\delta_{<v}$ ($\delta_{>v}$)
the assignment $\delta$ restricted to all variables in $V_{<v}$
($V_{>v}$) and by
$\delta^{Q}$ for $Q \in \{\exists,\forall\}$ the assignment $\delta$
restricted to all variables in $V^Q$. Finally, for two assignments
$\delta$ and $\delta'$ over disjoint sets of variables of
$\var(\qbfformula)$, we denote by $\delta \cup \delta'$ the disjoint
combination of the two assignments, i.e., $\delta \cup \delta'$ is
equal to the assignment $\delta'' :
\var(\delta)\cup\var(\delta')\rightarrow \{0,1\}$ defined by setting
$\delta''(v)=\delta(v)$ if $v \in \var(\delta)$ and
$\delta''(v)=\delta'(v)$ if $v \in \var(\delta')$.

For ease of notation, let $V=\var(\qbfformula)$ in the following.
A \emph{strategy for Eloise} (an \emph{existential strategy}) is a sequence of mappings
$\Tau=(\tau_v: \{0,1\}^{V_{<v}^{\forall}}\rightarrow \{0,1\})_{v \in
  V^\exists}$. For $\delta: V^\forall\rightarrow \{0,1\}$, we let $\alpha(\Tau, \delta) :
\var(\qbfformula) \rightarrow \{0,1\}$ be the assignment defined by setting:
\begin{itemize}
\item $\alpha(\Tau, \delta)(y)=\delta(y)$ if $y \in V^\forall$ and
\item $\alpha(\Tau, \delta)(x)=\tau_x(\delta_{<x})$ if $x \in V^\exists$.
\end{itemize}
An
existential strategy $\Tau$ is \emph{winning} if, for any assignment
$\delta: V^\forall\rightarrow \{0,1\}$, in the following also called
\emph{universal play}, the formula $F$ is satisfied
by $\alpha(\Tau,\delta)$.

A \emph{strategy for Abelard} (a \emph{universal strategy}) is defined
analogously, whereas the mappings $\delta$ and $\tau$ are swapped, and
we call a universal strategy winning if $F$ is not true. Formally,
it is a sequence of mappings
$\Lambda=(\lambda_v: \{0,1\}^{V_{<v}^{\exists}}\rightarrow
\{0,1\})_{v \in V^\forall}$. A universal strategy $\Lambda$ is
\emph{winning} if, for any assignment $\delta:V^\exists\rightarrow
\{0,1\}$, in the following also called \emph{existential play}, the formula $F$ is
false under the assignment $\alpha(\Lambda,\delta)$ defined
analogously to the existential case. That is, $\alpha(\Lambda, \delta) :
V \rightarrow \{0,1\}$ is the assignment defined by setting:
\begin{itemize}
\item $\alpha(\Lambda, \delta)(x)=\delta(x)$ if $x \in V^\exists$ and
\item $\alpha(\Lambda, \delta)(y)=\lambda_y(\delta_{<y})$ if $y \in V^\forall$.
\end{itemize}

Let $\Tau=(\tau_v: \{0,1\}^{V_{<v}^{\forall}}\rightarrow
\{0,1\})_{v \in V^\exists}$ be an existential
strategy and let $\Lambda=(\lambda_v: \{0,1\}^{V_{<v}^{\exists}}\rightarrow
\{0,1\}_{v \in V^\forall}$ be a universal strategy. Then,
playing $\Tau$ against $\Lambda$ gives rise to an assignment
$\alpha(\Tau,\Lambda) : V \rightarrow \{0,1\}$ of $\qbfformula$ that
can be defined recursively as follows. Let $\alpha_0$ be the empty
assignment. Moreover, for every $i$ with $1 \leq i \leq n$, we
distinguish two cases. If $Q_i=\exists$, then
$\alpha_i$ is the extension of $\alpha_{i-1}$ by the assignment
$\tau_{v_i}(\alpha_{i-1}^{\forall})$ for $v_i$. Similarly, if $Q_i=\forall$, then
$\alpha_i$ is the extension of $\alpha_{i-1}$ by the assignment
$\lambda_{v_i}(\alpha_{i-1}^{\exists})$ for $v_i$. Finally, we set $\alpha(\Tau,\Lambda)=\alpha_n$.
Note that $\alpha(\Tau,\Lambda)$ gives rise to the existential play
$\alpha(\Tau,\Lambda)^{\exists}$ and the universal play
$\alpha(\Tau,\Lambda)^{\forall}$ and it holds that
$\alpha(\Tau,\Lambda)=\alpha(\Tau,\alpha(\Tau,\Lambda)^{\forall})$ and 
$\alpha(\Tau,\Lambda)=\alpha(\Lambda,\alpha(\Tau,\Lambda)^{\exists})$.

A mapping $\delta$ from a subset of $V^\forall$ to $\{0,1\}$ is called a
\emph{universal play}, and similarly a mapping $\delta$ from a subset
of $V^\exists$ to $\{0,1\}$ is called an \emph{existential play}. 

\begin{PROP}[Folklore]
  A \textsc{QBF} $\qbfformula$ is true iff there exists a winning existential
  strategy on $\qbfformula$ iff there exists no winning universal strategy on
  $\qbfformula$.
\end{PROP}

\paragraph{Kernelization}

%
%
  A \emph{kernelization} is an algorithm that, given an
  instance~$(\mathcal{I},k) \in \Sigma^* \times \Nat$ outputs 
in time $\bigO{(\poly(\CCard{\mathcal{I}'} + k))}$ a
  pair~$(\mathcal{I}',k') \in \Sigma^* \times \Nat$, such that
  (i)~$(\mathcal{I},k) \in L$ if and only if $(\mathcal{I'},k') \in L$
  and (ii)~$\CCard{\mathcal{I}} + k' \leq g(k)$ where $g$ is an
  arbitrary computable function, called the size of the kernel.
  If $g$ is a polynomial then we say that $L$ admits a
  \emph{polynomial kernel}.
  It is well-known that a parameterized problem is fixed-parameter
  tractable if and only if it is decidable and has a
  kernelization~\cite{DowneyFellows13}. 

\paragraph{Feedback Edge Set}
For this parameter, we use the following result.
\begin{PROP}[Folklore]\label{pro:comp-fes}
  Let $G$ be an undirected graph. Then, a smallest feedback edge set
  for $G$ can be computed in polynomial time.
\end{PROP}
\begin{proof}
  A smallest FES $D$ can be
  computed in polynomial time by computing a spanning forest for
  $G$ and taking all edges outside the spanning forest into $D$.
\end{proof}



\section{Omitted Proof for Section~\ref{sec:saw}}

\begin{restatetheorem}[thm:cnfvcn]
\begin{THM}
There is an algorithm that, given a QBF $\qbfformula$ in CDNF with vertex cover
number $k$ of~$G_\qbfformula$, decides whether $\qbfformula$ is true in time $2^{2^{\bigO(k)}}\cdot
\poly(\CCard{\qbfformula})$.
If $\qbfformula$ is in $d$-CDNF, the algorithm runs in time $2^{k^{\bigO(d)}} \cdot \poly(\CCard{\qbfformula})$.
\end{THM}
\end{restatetheorem}
\begin{proof}
 Let $\qbfformula$ be a QBF with matrix $F = C \land D$, where $C$ is in CNF and $D$ in DNF. Further, let $X$ be a vertex cover of the primal graph such that $\Card{X} = k$, which can be computed in time $2^k \cdot poly(\CCard{\qbfformula})$.
 For simplicity, we assume that each quantifier block contains a single variable, so that the prefix of $\qbfformula$ can be written as $Q_1v_1 \dots Q_nv_n$.
Consider a backtracking algorithm that maintains a variable assignment $\alpha: \{v_1, \dots, v_i\} \rightarrow \{0, 1\}$ for some $1 \leq i \leq n$ and a cache containing pairs $(F', \mathit{val})$ consisting of previously encountered subformulas $F'$ of the matrix and their truth values $\mathit{val} \in \{0, 1\}$.
The algorithm first checks whether there is an entry $(F[\alpha], \mathit{val})$ in the cache and if so, returns the corresponding value.
If there is no entry in the cache, it recursively determines the truth value $\mathit{val}$ of $Q_{i+1}v_{i+1}\dots Q_nv_n.F[\alpha]$ as follows.
If $i = n$, no unassigned variables remain and $\mathit{val} = 1$ if $F$ is satisfied by $\alpha$ and $\mathit{val} = 0$ if $F$ is falsified.
Otherwise, the algorithm branches on the next unassigned variable $v_{i+1}$ and recursively calls itself with the assignments $\alpha \cup \{v_{i+1} \mapsto 0\}$ and $\alpha \cup \{v_{i+1} \mapsto 1\}$.
The algorithm sets $\mathit{val}$ to the maximum of the return values if $v_{i+1}$ is existentially quantified, and to the minimum if it is universally quantified.
It then adds the pair $(F[\alpha], \mathit{val})$ to the cache and returns $\mathit{val}$.

Clearly, the algorithm correct. Discarding recursive calls, it runs in time $poly(K)$, where $K$ is the number of entries in the cache.
Both the number of nodes in the search tree and $K$ are bounded by the number of subformulas $F[\alpha]$ for partial assignments $\alpha: \{v_1, \dots, v_i\} \rightarrow \{0, 1\}$.
We argue that there are at most $2^k \cdot 2^{(3^k)^2}$ such subformulas for each $1 \leq i \leq n$.
A partial assignment $\alpha$ can be written as a disjoint union $\alpha = \alpha_X \cup \beta$, where $\alpha_X$ assigns variables in the vertex cover $X$, and $\beta$ assigns variables outside the vertex cover.
We can also partition the clauses and terms of $F$ as $F = F' \cup F_{> i}$, where $F_{>i}$ does not contain variables in $\{v_1, \dots, v_i\}$, so that $F[\alpha] = F'[\alpha] \cup F_{>i}$. It is sufficient to bound the number of distinct subformulas $F'[\alpha]$.
We further partition $F'$ as $F' = F_X \cup F''$, where $\var(F_X) \cap \{v_1, \dots, v_i\} \subseteq X$, and each clause or term in $F''$ contains a variable in $\{v_1, \dots, v_i\} \setminus X$.
This allows us to write $F'[\alpha] = F'[\alpha_X][\beta] = F_X[\alpha_X] \cup F''[\alpha_X][\beta]$.
There are at most $2^k$ many subformulas $F_X[\alpha_X]$ and $F''[\alpha_X]$. For each subformula $F''[\alpha_X]$, the subformula $F''[\alpha_X][\beta]$ is obtained by assigning variables that do not occur in the vertex cover $X$, which means that all their neighbors must be inside the vertex cover.
Thus $F''[\alpha_X][\beta]$ is a subformula constructed from variables in~$X$.
There are at most $3^k$ clauses (or terms) on $k$ variables, and thus at most $2^{3^k}$ CNF (or DNF) formulas.
Because $F$ consists of a CNF and a DNF, there are at most $2^{(3^k)^2}$ CDNF formulas, and $2^k \cdot 2^{(3^k)^2}$ subformulas $F'[\alpha]$ overall, which is in $2^{3^{\bigO(k)}}$.
In the case of $d$-CDNF, there are $(2k+1)^d$ clauses or terms of size at most $d$, so we get a bound of $2^{k^{\bigO(d)}}$ on the number of subformulas.
Since we get a different domain of $\alpha$ for each $1 \leq i \leq n$, we have a bound of $n \cdot 2^{3^{\bigO(k)}}$ on the number of cache entries and nodes in the search tree, or $n \cdot 2^{k^{\bigO(d)}}$ in for $d$-CDNF.
Thus the overall running time is $2^{2^{\bigO(k)}}\cdot \poly(\CCard{\qbfformula})$ for general CDNF, and $2^{k^{\bigO(d)}} \cdot \poly(\CCard{\qbfformula})$ for $d$-CDNF.
\end{proof}

\begin{restatetheorem}[ref:ub]
\begin{THM}[UB for $\QBFSAT_\ell$ and Treewidth]
Given any QBF~$Q$ in CDNF of quantifier depth~$\ell$ with~$k=\tw{G_Q}$. Then, $\QBFSAT_\ell$ on~$Q$ can be decided in time~$\tower(\ell, \mathcal{O}(k))\cdot\poly(\Card{\var(Q)})$.
\end{THM}
\end{restatetheorem}
\begin{proof}[Proof] 
We illustrate the proof on the case where for~$Q$ the innermost quantifier~$Q_\ell=\exists$. 
Let $Q=Q_1 V_1. \cdots Q_\ell V_\ell.$ $C\wedge D$ and let $\mathcal{T}=(T,\chi)$ be a TD
of~$G_Q$ of width~$\mathcal{O}(k)$, computable in time~$2^{\mathcal{O}(k)}\cdot\poly(\Card{\var(Q)})$~\cite{Korhonen22}. 
For each node~$t$ of~$T$, let the set of child nodes be given by~$\children(t)$;
we assume without loss of generality that~$\Card{\children(t)}\leq 2$ (obtainable by adding auxiliary nodes).
We use auxiliary variables~$S\eqdef\{sat_t \mid t\text{ in }T\}$. 
Then, we define a linear-SAW reduction from~$Q$ and~$\mathcal{T}$, constructing
a QBF~$Q'\eqdef Q_1 V_1. \cdots Q_\ell (V_\ell \cup S). (C\cup C'')$, 
whose matrix is in CNF: 
\vspace{-.25em}
{\smallalign{\normalfont\small}
\begin{align}
	&\label{t:aux}sat_t \rightarrow\hspace{-.5em}\bigvee_{t'\in\children(T)}sat_{t'} \vee\hspace{-2em}\bigvee_{d\in D, \var(d)\subseteq\chi(t)}\hspace{-2em}d& \text{for every }t\text{ of }T\\
	&\label{t:root}sat_{r} & \text{for root }r\text{ of }T
\end{align}}
\noindent Formulas~(\ref{t:aux}) define when a term is satisfied for a node~$t$, 
which together with Formula~(\ref{t:root}) can be easily converted to 
the set~$C'$ of CNFs (using distributive law).
Observe that 
at least one term has to be satisfied
at the root node.
The reduction yields 
a TD~$\mathcal{T}'\eqdef(T,\chi')$ of~$G_{Q'}$ 
%
where, for every~$t$ of~$T$, $\chi'(t)\eqdef \chi(t)\cup \{sat_t\} \cup \{sat_{t'}\mid t'\in\children(t)\}$. 
The width of $\mathcal{T}'$ is $ \width(\mathcal{T}) + 3 \in\mathcal{O}(\width(\mathcal{T}))$,
so on~$Q'$ the algorithm from Proposition~\ref{ref:ubo} runs in 
time~$\tower(\ell, \mathcal{O}(k))\cdot\poly(\Card{\var(Q)})$. 
\end{proof}

\section{Lower Bounds for More Restricted Parameters}\label{app:strongerlowerbounds}
The reduction $\mathcal{R}$ presented in Section~\ref{sec:result} can be shown to reduce parameters that are even more restrictive than the sparse feedback vertex number (Lemma~\ref{lem:compr}).
\begin{COR}[Decreasing Distance to Sparse Half-Ladder
]\label{cor:compr}
Given a QBF~$Q$ in 3,1-CDNF such that~$S$ is 
a distance set to sparse half-ladder of~$Q$. Then, the reduction~$\mathcal{R}$ constructs a QBF~$Q'$ with distance set~$S'$ to sparse half-ladder of~${Q'}$ such that~$\Card{S'}$ is bounded by~$\mathcal{O}(\log(\Card{S}))$.
\end{COR}
\begin{proof}\vspace{-.5em}
The result follows from Lemma~\ref{lem:compr}
and the observation that~$S'$ is a distance set to sparse \emph{caterpillar} of~${Q'}$ whenever~$S$ is a distance set to sparse half-ladder of~$Q$.
This caterpillar is due to those 3-CNF clauses~$c_i$ of~$Q$ using only one variable~$v$ from~$\var(\matr(\Q))\setminus S$, cf.\ type (ii) clauses of Figure~\ref{fig:sketch}. However, the caterpillar can be turned into a half-ladder by introducing one additional innermost existential copy variable~$v_i$ of~$v$ for every such~$c_i$. Then, for an arbitrary total ordering among those copies, we ensure equivalence of~$v$ with the first copy~$v_i$ by ($v_i \wedge \neg v$) and ($v\wedge \neg v_i$), as well as equivalence between the first and the second copy, and so on. We adapt Formulas~(\ref{red:usatv}) such that instead of~$v$ having all neighbors~$sat_i$ in~$G_{Q'}$, each copy~$v_i$ gets one neighbor~$sat_i$, thereby dissolving star subgraphs into half-ladders.
%
%
\end{proof}

\begin{COR}[LB for Distance Sparse Half-Ladder]\label{cor:disjpaths}
Given an arbitrary QBF~$Q$ in CDNF of quantifier depth~$\ell$ and a minimum distance set~$S$
to sparse half-ladder of~$Q$ with~$k=\Card{S}$.
Then, under ETH, $\QBFSAT_\ell$ on~$Q$ cannot be decided in time~$\tower(\ell, o(k))\cdot\poly(\Card{\var(Q)})$.
\end{COR}
\begin{proof}\vspace{-.28em}
The result follows from Theorem~\ref{lab:lb}, where instead of Lemma~\ref{lem:compr}, Corollary~\ref{cor:compr} is used. 
\end{proof}


Interestingly, for finding the respective distance set~$S$, 
variables among the innermost quantifier are enough, which becomes apparent 
when inspecting reduction~$\mathcal{R}$ and Lemma~\ref{lem:compr}.

\begin{COR}[LB for Distance to Sparse Half-Ladder (Innermost Quantifier)]\label{cor:lastquantifier}
Given any QBF~$Q=Q_1 V_1 \cdots Q_\ell V_\ell. F$ in CDNF and a minimum distance set~$S\subseteq V_\ell$ 
to sparse half-ladder of~$Q$ with~$k=\Card{S}$.
Then, under ETH, $\QBFSAT_\ell$ on~$Q$ cannot be decided in time~$\tower(\ell, o(k))\cdot\poly(\Card{\var(Q)})$.
\end{COR}
\begin{proof}\vspace{-.28em}
The result follows from Corollary~\ref{cor:disjpaths}, by observing that~$\mathcal{R}$ constructs resulting distance sets~$S'$ only over the innermost quantifier, cf.\ Lemma~\ref{lem:compr}. 
\end{proof}

Obviously, the results of Theorem~\ref{lab:lb} and Corollaries~\ref{cor:disjpaths}, \ref{cor:lastquantifier} 
immediately carry over to the (non-sparse) feedback vertex number and the distance to half-ladder, respectively, cf.~Figure~\ref{fig:params}.
However, it turns out that this result can be strengthened even further, which we show below.
In particular, the reduction~$\mathcal{R}$ even works for a slightly weaker parameter than the distance to caterpillar: 
If the \emph{height (largest component)} of the caterpillar is bounded by a fixed value~$h$, 
we have that~$\mathcal{R}$ does not significantly increase the height, i.e., the height of the caterpillar of the resulting graph is in~$\mathcal{O}(h)$.
%

\begin{COR}
Given an arbitrary QBF~$Q$ in CDNF of quantifier depth~$\ell$ and a minimum distance set~$S$ to sparse caterpillar of~$Q$ of bounded height~$\mathcal{O}(h)$ with~$k=\Card{S}$. Then, under ETH, $\QBFSAT_\ell$ on~$Q$ cannot be decided in time~$\tower(\ell, o(k))\cdot\poly(\Card{\var(Q)})$.
\end{COR}
\begin{proof}
The result follows from Corollary~\ref{cor:disjpaths} with the observation that the height of the caterpillar~$G_{Q'}$ for~$Q'$, obtained by~$\mathcal{R}(Q, S)$, is bounded by~$2\cdot h+1$, i.e., $\mathcal{O}(h)$, cf.\ Figure~\ref{fig:sketch}.
\end{proof}

\noindent Note that we cannot lift this result to the distance to half-ladder by applying a similar procedure as in the proof of Corollary~\ref{cor:compr}, which would increase the height to~$\mathcal{O}(h\cdot\log(k))$.

Interestingly, when the half-ladder \emph{size (number of components)} is bounded, the problem $\QBFSAT_\ell$ stays hard as well. 

\begin{COR}
Given an arbitrary QBF~$Q$ in CDNF of quantifier depth~$\ell$ and a minimum distance set~$S$ to sparse half-ladder (of~$Q$) of size~$m$ with~$k=\Card{S}$. Then, under ETH, $\QBFSAT_\ell$ on~$Q$ cannot be decided in time~$\tower(\ell, o(k))\cdot\poly(\Card{\var(Q)})$.
\end{COR}
\begin{proof}
The result follows from the observation that~$\mathcal{R}$ as given above can be extended by additional paths for clauses that only use variables over~$S$. 
These clauses can be easily connected to one (large) path. 
Thereby, Formulas~(\ref{red:usatxp}) 
for clauses~$c_i\in C$ with~$\var(c)\subseteq S$ (``type (i) clause'', cf.\ Figure~\ref{fig:sketch}) are modified such that
every occurrence of $sat_i$ is replaced by a conjunction consisting of~$sat_j \wedge sat_i$, where~$j$ is the largest value s.t.\ type~(i) clause~$c_j$ of~$C$ directly precedes~$c_i$ of~$C$ in the total ordering. 
Similarly, elements of~$S$ can be connected to one path, via, e.g., additional auxiliary variables. The resulting, modified reduction of~$\mathcal{R}$ is referred to by~$\mathcal{R}'$.
Then, Corollary~\ref{cor:disjpaths} can be easily lifted in order to show the claim, since the size of the half-ladder~$G_{Q'}$ for~$Q'$ obtained by~$\mathcal{R}'(Q, S)$ is bounded by~$m+2$.
\end{proof}

However, if both height and size are bounded, we obtain a result similar to vertex cover number.

\begin{COR}
Given any QBF~$Q$ in 3,1-CDNF of quantifier depth~$\ell$ and a distance set~$S$ to caterpillar (of~$G_Q$) of size~$m$ and height~$h$ with~$k=\Card{S}$. Then, there is an algorithm for $\QBFSAT_\ell$ on~$Q$ running in time~$2^{\mathcal{O}((k+hm)^3)}\cdot\poly(\Card{\var(Q)})$.
\end{COR}
\begin{proof}
The result follows from Proposition~\ref{prop:vco}, since one can easily construct a vertex cover of~$G_Q$ of size~$\mathcal{O}(k+h\cdot m)$ that contains all vertices of~$G_Q$.
\end{proof}

%
%
%

\section{Omitted Details and Proofs for Section~\ref{sec:main}}

\subsection{Tight Lower Bound for Feedback Vertex Number}

\begin{EX}\label{ex:redu}
Recall QBF~$Q'$ from Example~\ref{ex:running3} with~$\matr(Q')=C{\wedge}D$. Consider the sparse feedback vertex set~$S\eqdef \{a,c\}$ of~$Q'$. Since~$\Card{S}=2$, we only need two indices and we assume that for every clause~$c_i\in C$, $a=\var(\lit{c_i}{1})$ is always the first variable and~$c=\var(\lit{c_i}{2})$ the second variable.
We further assume that~$\bval{a}{1}\eqdef\{\neg idx_1^1\}$, $\bval{c}{2}\eqdef \{\neg idx_2^1\}$ as well as~$\bval{a}{2}\eqdef\{idx_2^1\}$, $\bval{c}{1}\eqdef \{idx_2^1\}$. Note that, however, since $a$ ($c$) is always the first (second) variable, respectively, $\bval{a}{2}$ and~$\bval{c}{1}$ are not used.
$\mathcal{R}(Q',S)$ amounts to $Q''{=}\forall a,b.$ $ \exists c,d. \forall idx_1^1, idx_2^1, val_1, val_2, sat, sat_1, sat_2, sat_3, sat_4.$ $C' \vee D'$, where we discuss the corresponding DNF terms of~$C'$ and the singleton clauses of~$D'$ below.
\vspace{-.01em}%

	%
		%
\hspace{-1.75em}
		\begin{tabular}[h]{l@{\hspace{0.5em}}lll}%
%
%
		(\ref{red:guessatomj}) & $a \wedge \neg idx_1^1 \wedge \neg val_1$,\qquad $\neg c \wedge \neg idx_2^1 \wedge \neg val_2$ \\
		(\ref{red:guessnegatomj}) & $\neg a \wedge \neg idx_1^1 \wedge \neg val_1$,\qquad $\neg c \wedge \neg idx_2^1 \wedge \neg val_2$  \\
(\ref{red:usatv}) & $sat_1\wedge \neg b$,\quad $sat_2\wedge b$,\quad $sat_3\wedge d$,\quad $sat_4\wedge \neg d$\\
		(\ref{red:usatxp}) & $sat_1 \wedge idx_1^1$, $sat_2 \wedge idx_1^1$, $sat_3\wedge idx_1^1$, $sat_4\wedge idx_1^1$,\\  
 & $sat_1 \wedge idx_2^1$, $sat_2 \wedge idx_2^1$, $sat_3\wedge idx_2^1$, $sat_4\wedge idx_2^1$\\  
%
%
%
(\ref{red:usatxv}) &$sat_2\wedge val_1$, $sat_4\wedge val_1$,  
 $sat_1 \wedge val_2$, $sat_2\wedge val_2$ \\ 
(\ref{red:usatxnegv}) &  $sat_1 \wedge \neg val_1$,  \hspace{-.1em}$sat_3\wedge \neg val_1$,  \hspace{-.1em}$sat_3\wedge \neg val_2$, \hspace{-.1em}$sat_4\wedge$\\
(\ref{red:usatvd}) &  $\neg val_2$  \qquad $sat \wedge b$,\qquad $sat\wedge \neg d$  \\
(\ref{red:usatnot}) & $\neg sat_1$,\qquad $\neg sat_2$,\qquad $\neg sat_3$,\qquad $\neg sat_4$\\
(\ref{red:usatnotd})& $\neg sat$
%
%
%
		\end{tabular}

\vspace{-1em}
\end{EX}

	\begin{restatetheorem}[lem:compr]
\begin{LEM}[Decrease Feedback Vertex Number] 
Given QBF~$Q$ in 3,1-CDNF and a sparse feedback vertex set~$S$ of $Q$, $\mathcal{R}$ constructs QBF~$Q'$ with sparse feedback vertex set~$S'$ of~$Q'$ such that~$\Card{S'}$ is in~$\mathcal{O}(\log(\Card{S}))$.
\end{LEM}
\end{restatetheorem}
\begin{proof}\vspace{-.5em}
  As described above, $\mathcal{R}$ gives rise to a sparse feedback vertex set~$S'$ of~${Q'}$. Indeed, (i) $G_{Q'}-S'$ results in an acyclic graph, since each~$x\in S$ is isolated in~$G_{Q'}-S'$ and the only edges remaining in~$G_{Q'}-S'$ involve some~$sat_i$ and~$y$ of~$G_{Q'}-S'$.
  Towards a contradiction, assume that~$G_{Q'}-S$ contains a cycle $x_1,sat_{i_1},x_2,sat_{i_2},\dots ,sat_{i_{r-1}}, x_{r}$ with $x_{r} = x_1$.
  Variables $x_j$ and $x_{j+1}$ are adjacent to~$sat_{i_j}$ in~$G_{Q'}$ only if $x_j,x_{j+1} \in \var(c_i)$, so $x_j$ and $x_{j+1}$ are adjacent in~$G_{Q}$.
So if $r > 2$, we get a cycle in $G_Q-S$, contradicting the assumption that $S$ is a feedback vertex set.
If $r = 2$ the cycle is of the form $x_1,sat_{i_1},x_2,sat_{i_2},x_1$ and there are distinct clauses $c_{i_1}, c_{i_2}$ such that $x_1,x_2 \in \var(c_{i_1})$ and $x_1,x_2 \in \var(c_{i_2})$, contradicting the assumption that $S$ is sparse.
Further, (ii) the only terms where adjacent vertices of~$G_{Q'}-S'$ may occur together is in those of Formulas~(\ref{red:usatv}), since the other DNF terms use at most one variable that is not in $S'$.
This proves that $S'$ is a sparse feedback vertex set of $G_{Q'}$.
Finally,  by construction $\Card{S'}\leq 3\cdot\ceil{\log(\Card{S})}+4$, which is in~$\mathcal{O}(\log(\Card{S}))$.
\end{proof}

\begin{restatetheorem}[lab:runtime]
\begin{THM}[Runtime] 
For a QBF~$Q$ in 3,1-CDNF with $\matr(Q)=C\wedge D$ and set~$S\subseteq\var(Q)$ of variables of~$Q$, $\mathcal{R}$ runs in time~$\mathcal{O}(\ceil{\log(\Card{S}+1)} \cdot (\Card{S} + \Card{C}) + \Card{D})$.
\end{THM}
\end{restatetheorem}
\begin{proof}
There are~$\mathcal{O}(\Card{S})$ many instances of Formulas~(\ref{red:guessatomj}), (\ref{red:guessnegatomj}), each of size~$\mathcal{O}(\log(\Card{S}))$.
Further, there are~$\mathcal{O}(\Card{C})$ many instances of constant-size Formulas~(\ref{red:usatv}), (\ref{red:usatxv}), and (\ref{red:usatxnegv}). 
Finally, there are~$\mathcal{O}(\Card{C})$ many instances of Formulas~(\ref{red:usatnot}), whose size is bounded by~$\mathcal{O}(\log(\Card{S}))$,
$\mathcal{O}(\Card{C}\log(\Card{S}))$ many instances of Formulas~(\ref{red:usatxp}) of size 2,
as well as~$\mathcal{O}(\Card{D})$ many constant-size Formulas~(\ref{red:usatvd}).
%
%
%
%
%
%
\end{proof}

\begin{restatetheorem}[lab:corr]
\begin{THM}[Correctness]
  Given a QBF~$Q$ in 3,1-CDNF and a set~$S\subseteq\var(Q)$ of variables of~$Q$, 
  reduction $\mathcal{R}$ computes an instance~$Q'$ that is equivalent to $Q$. In fact, any assignment~$\alpha$ to variables of~$\matr(Q)$ satisfies~$\matr(Q)$ iff every extension~$\alpha'$ of~$\alpha$ to variables~$\lbvs{S} \cup \lbvvs{S} \cup \lsat$ satisfies~$\matr(Q')$.
\end{THM}
\end{restatetheorem}

\begin{proof}
%
%
%
$\Longrightarrow$: Assume that~$Q$ is valid. Then, in the following we show that for any satisfying assignment~$\alpha: \var(Q) \rightarrow \{0,1\}$ of~$C\wedge D$, we have that any extension~$\alpha'$ of~$\alpha$ to variables~$\lbvs{S} \cup \lbvvs{S} \cup \lsat$ is a satisfying assignment of~$\matr(Q')=C'\vee D'$. 
Assume towards a contradiction that there is such an extension~$\alpha'$ with~$C'[\alpha']\neq\{\emptyset\}$ and~$D'[\alpha']\neq\emptyset$.
Then, none of the terms of~$C'$ as given by Formulas~(\ref{red:guessatomj})--(\ref{red:usatvd}) are satisfied by~$\alpha'$ and at least one of the clauses of~$D'$ given by Formulas~(\ref{red:usatnot}) and~(\ref{red:usatnotd}) is not satisfied by~$\alpha'$ as well.
We distinguish two cases.

Case (a): Some of Formulas~(\ref{red:usatnot}) are not satisfied by~$\alpha'$, i.e., $\alpha'(sat_i)=1$ for some~$1\leq i \leq\Card{C}$.
Since~$\alpha'$ does not satisfy any of Formulas~(\ref{red:usatxp}), we have that the
$j$-th index ($1\leq j\leq 3$) precisely targets the $j$-th variable~$v_j\eqdef\var(\lit{c_i}{j})$ of clause~$c_i$, i.e., $\alpha'(\var(\bval{v_j}{j}))=\sgn(\bval{v_j}{j})$. 
Then, since $\alpha'$ satisfies none of Formulas~(\ref{red:guessatomj}) and~(\ref{red:guessnegatomj}), we have that
 assignment~$\alpha'$ sets the value~$val_j$ for the~$j$-th variable~$v_j$ of~$c_i$ that is in~$S$ (with~$v_j\in S$), precisely according to~$\alpha$, i.e., such that
$\alpha'(val_j)=\alpha(v_j)$.
Then, however, $\alpha$ does not satisfy clause~$c_i$ due to the assignment of any variable~$y$ that is in~$S$,
since none of Formulas~(\ref{red:usatxv}) and~(\ref{red:usatxnegv}) are satisfied by~$\alpha'$. 
Finally, since~$c_i$ is still satisfied by~$\alpha$, there is at least one such variable~$y\in \var(c_i)\setminus S$ with~$\alpha(y)=\alpha'(y)$
such that~$\alpha'$ satisfies precisely the instance of Formula~(\ref{red:usatv}), where~$l\in c_i$ is a literal over~$y$, i.e., $\var(l)=y$.
This contradicts the assumption that~$\alpha'$ neither satisfies~$C'$ nor~$D'$, as constructed by Formulas~(\ref{red:guessatomj})--(\ref{red:usatnotd}).

Case (b): Formula~(\ref{red:usatnotd}) is not satisfied by~$\alpha'$, i.e., $\alpha'(sat)=1$. Then, since none of Formulas~(\ref{red:usatvd}) is satisfied by~$\alpha'$, we have that~$\alpha'$ does not satisfy any~$\{l\}\in D$, i.e., $\alpha'(\var(l))\neq \sgn(l)$ for every~$\{l\}\in D$. 
Consequently, $D[\alpha']\neq\{\emptyset\}$ and therefore by construction of~$\alpha'$, we have that~$D[\alpha]\neq\{\emptyset\}$. This, however, contradicts the assumption that~$\alpha$ is a satisfying assignment of~$\matr(Q)$.

$\Longleftarrow$: We show this direction by contraposition, where we take any assignment~$\alpha: \var(Q)\rightarrow\{0,1\}$ that does not satisfy~$\matr(Q)$ and show that then there is an extension~$\alpha'$ of~$\alpha$ to variables~$\lbvs{S} \cup \lbvvs{S} \cup \lsat$ such that~$\alpha'$ does not satisfy $\matr(Q')$.
We proceed again by case distinction.

Case (a): $C[\alpha]\neq\emptyset$ due to at least one clause~$c_i\in C$, i.e., $\{c_i\}[\alpha]\neq\emptyset$.
Then, (i) we set~$\alpha'(sat)\eqdef 0$, $\alpha'(sat_i)\eqdef 1$ as well as $\alpha'(sat_{i'})\eqdef 0$ for any~$1\leq i'\leq \Card{C}$ such that~$i'\neq i$.
Further, for each~$1\leq j\leq 3$ we let~$v_j\eqdef\lit{c_i}{j}$ be the variable of the~$j$-th literal of clause~$c_i$.
Finally, (ii) we let the~$j$-th index point to~$v_j$, i.e., we set~$\alpha'(\var(b))\eqdef \sgn(b)$ for each~$b\in\bval{v_j}{j}$ and every~$1\leq j\leq 3$,
and (iii) we set the value of the~$j$-th index such that~$\alpha'(val_j)\eqdef\alpha(v_j)$.
Consequently, by construction of~$\alpha'$, no instance of Formulas~(\ref{red:guessatomj}) or~(\ref{red:guessnegatomj}) is satisfied by~$\alpha'$.
Since~$\alpha$ does not satisfy $\{c_i\}$, we follow that~$\alpha'$ does not satisfy any instance of Formulas~(\ref{red:usatv}).
Further, since by construction (ii) of~$\alpha'$, the~$j$-th index targets~$v_j$, neither one of Formulas~(\ref{red:usatxp}) can be satisfied by~$\alpha'$.
Similarly, by (iii) no instance of Formulas~(\ref{red:usatxv}) or~(\ref{red:usatxnegv}) is satisfied by~$\alpha'$.
Finally, by construction (i), neither Formulas~(\ref{red:usatvd}) are satisfied by~$\alpha'$, nor is~$D'$, since, e.g.,  not all Formulas~(\ref{red:usatnot}) are satisfied by~$\alpha'$. 

Case (b): $D[\alpha]\neq\{\emptyset\}$, i.e., $\alpha(l)\neq\sgn(l)$ for every~$\{l\}\in D$. In this case, (i) we set~$\alpha'(sat)\eqdef 1$ and (ii) $\alpha'(sat_i)\eqdef 0$ for any~$1\leq i\leq \Card{C}$. Further, for each~$1\leq j\leq 3$ and any arbitrary~$x\in S$ as well as~$b\in \bval{x}{j}$ (iii) we let~$\alpha'(\var(b))\eqdef \sgn(b)$, (iv) as well as~$\alpha'(val_j)\eqdef\alpha(x)$.
Then, by Construction (i) of~$\alpha'$ we have that~$D'[\alpha']\neq\emptyset$ due to Formula~(\ref{red:usatnotd}). Since~$D[\alpha]\neq\{\emptyset\}$, and
despite Construction (i), we have that~$\alpha'$ does not satisfy any of Formulas~(\ref{red:usatvd}). Further, by Construction (ii), neither one of Formulas~(\ref{red:usatv})--(\ref{red:usatxnegv}) is satisfied by~$\alpha'$ as well.
Finally, by Construction (iii) and (iv) neither one of Formulas (\ref{red:guessatomj}) or (\ref{red:guessnegatomj}) is satisfied by~$\alpha'$. Therefore, $C'[\alpha']\neq\{\emptyset\}$, which concludes this case.
\end{proof}


	\begin{restatetheorem}[lab:lb]
\begin{THM}[LB for Sparse Feedback Vertex Set]
Given an arbitrary QBF~$Q$ in CDNF of quantifier depth~$\ell$ and a minimum sparse feedback vertex set~$S$ of~$Q$ with~$k=\Card{S}$. 
Then, under ETH, $\QBFSAT_\ell$ on~$Q$ cannot be decided in time~$\tower(\ell, o(k))\cdot\poly(\Card{\var(Q)})$.
\end{THM}
\end{restatetheorem}
\begin{proof}\vspace{-.28em}
Without loss of generality, we assume~$Q$ in 3,1-CDNF.
The result with quantifier depth~$\ell=1$ corresponds to \SAT and therefore follows immediately by ETH,
since~$k\leq \Card{\var(Q)}$ and due to the fact that ETH implies that there is no algorithm for solving \SAT on 3-CNFs running in time~$2^{o(\Card{\var(Q)})}$.
For the case of~$\ell>1$, we apply induction.
Assume that the result holds for~$\ell-1$. Let~$Q$ be such a QBF of quantifier depth~$\ell-1$.

Case 1: Innermost quantifier~$Q_{\ell-1}$ of~$Q$ is $\exists$.
Then, we apply the reduction~$\mathcal{R}$ on~$(Q, S)$. Thereby, we obtain a resulting instance~$Q'$ and a feedback vertex set~$S'$ of~$G_{Q'}$ in time~$\mathcal{O}(\poly(\Card{\var(Q)})$,
cf.\ Theorem~\ref{lab:runtime}.
The reduction is correct by Theorem~\ref{lab:corr}, i.e., the set of 
satisfying assignments of~$\matr(Q)$ coincides with the set of satisfying assignments of~$\matr(Q')$ when restricted to variables~$\var(Q)$.
Further, we have that~$\Card{S'}\leq 3\cdot\ceil{\log(\Card{S})}$ by Lemma~\ref{lem:compr}.
Assume towards a contradiction that despite ETH, $\QBFSAT$ on~$Q'$ can be decided in time~$\tower(\ell, o(\log(k)))\cdot\poly(\Card{\var(Q')})$.
Then, the validity of~$Q$ can be decided in time~$\tower(\ell-1, o(k))\cdot\poly(\Card{\var(Q)})$, 
contradicting the hypothesis. 

Case 2: Innermost quantifier~$Q_{\ell-1}$ of~$Q$ is~$\forall$. There, we proceed similar to Case 1, but invert~$Q$ first, resulting in a QBF~$Q^\star$, whose quantifier blocks are flipped such that~$\matr(Q^\star)$ is again in 3,1-CDNF. Then, as in Case 1, after applying reduction~$\mathcal{R}$ on~$(Q^\star, S)$, we obtain~$Q'$. The lower bound follows similar to above, since the truth value of~$Q'$ can be simply inverted in constant~time. 
\end{proof}

\begin{restatetheorem}[lab:inc-lb]
\begin{COR}[LB for Incidence Feedback Vertex Set]
Given an arbitrary QBF~$Q$ with~$F=\matr(Q)$ in CNF (DNF) such that the innermost quantifier~$Q_\ell$ of~$Q$ is $Q_\ell=\exists$ ($Q_\ell=\forall$) 
with the feedback vertex number of~$I_{F}$ being~$k$.
Then, under ETH, $\QBFSAT_\ell$ on~$Q$ cannot be decided in time~$\tower(\ell, o(k))\cdot\poly(\Card{\var(Q)})$.
\end{COR}
\end{restatetheorem}
\begin{proof}
We show that the sparse feedback vertex number of any QBF~$Q'$ in 3,1-CDNF linearly bounds the feedback vertex number of~$I_F$, where~$F$ is obtained by converting~$\matr(Q')$ into CNF or DNF, depending on the innermost quantifier of~$Q'$.
Let~$Q'=Q_1 V_1. \cdots Q_\ell V_\ell. F'$ be any QBF in 3,1-CDNF and~$S'$ be a minimum sparse FVS of~$Q'$ with~$k'=\Card{S'}$.   
We construct QBF~$Q\eqdef Q_1 V_1. \cdots Q_\ell V_\ell. F$, with~$F$ being defined below. 

Case 1: $Q_\ell=\exists$ and therefore $F'=C \wedge D$.
We define~$F\eqdef C\cup \{f\}$ in CNF with~$f\eqdef \{l \mid \{l\} \in D\}$ being a clause. It is easy to see that~$S\eqdef S'\cup\{f\}$ is a feedback vertex set of~$I_F$. Assume towards a contradiction that there is a cycle in~$I_F-S$. By the definition of the incidence graph~$I_F$, the cycle alternates between vertices~$\var(F)\setminus S$ and~$F\setminus S$, i.e., either we have: (1) the cycle restricted to vertices~$\var(F)$ is already present in~$G_Q$, which contradicts that~$S$ is a (sparse) FVS of~$G_Q$; or (2) the cycle is of the form $x,c,y,c',x$ where $c$ and $c'$ are distinct clauses, contradicting the assumption that $S$ is sparse. As a result, assuming ETH and that we can decide the validity of~$Q$ in time~$\tower(\ell, o(\Card{S}))\cdot\poly(\Card{\var(Q)})$ contradicts Theorem~\ref{lab:lb}.

Case 2: $Q_\ell=\forall$, i.e., $F'=D \vee C$.
We define~$F\eqdef D\cup \{f\}$ in DNF with~$f\eqdef \{l \mid \{l\} \in C\}$ and proceed as in Case 1.
\end{proof}

\noindent Note that one can obtain a minimum FVS of a graph~$G$ 
in time~$2^{\mathcal{O}(k)}\cdot\poly(|{\var(G)}|)$, where~$k$ is the feedback vertex number of~$G$~\cite[Fig.\ 1]{LiNederlof20}. For computing a minimum sparse FVS of~$Q$, one might adapt standard FVS dynamic programming algorithms, using the stronger parameter treewidth $k'$ that can be $2$-approximated in time~$2^{\mathcal{O}(k')}\cdot\poly(|{\var(G_Q)}|)$~\cite{Korhonen22}. These algorithms can be adapted for sparse FVS and run in time~$2^{\mathcal{O}(k'\cdot\log(k'))}\cdot\poly(|{\var(G_Q)}|)$~\cite{CyganEtAl15}.

QSAT is well-known to be polynomial-time tractable when restricted to 2-CNF formulas~\cite{AspvallPT79}.
As an application of our reduction~$\mathcal{R}$, we now observe that allowing a single clause of arbitrary length already leads to intractability.

\begin{restatetheorem}[cor:hardnesslongclause]
\begin{COR}
The problem~$\QBFSAT$ over QBFs $Q=Q_1 V_1 \cdots Q_\ell V_\ell. C \wedge D$ of quantifier depth~$\ell\geq 2$ with~$Q_\ell=\exists$, $C$ being in 2-CNF, and~$D$ being in 1-DNF, is~$\SIGMA{{\ell-1}}{\Ptime}$-complete (if $Q_1=\exists$, $\ell$ odd) and~$\PI{{\ell-1}}{\Ptime}$-complete (if $Q_1=\forall$, $\ell$ even).
\end{COR}
\end{restatetheorem}
\begin{proof}[Proof (Sketch)]
We only sketch the proof for~$Q_1{=}\exists$ (the proof for the case~$Q_1{=}\forall$ is similar).
For membership, it is sufficient to observe that satisfiability of~$C \land D$ can be checked in polynomial time.
We simply try, for each literal~$l \in D$, whether the 2-CNF formula obtained by assigning $l$ true is satisfiable.
For hardness, let~$Q'=\forall V'_1 \cdots \exists V'_{\ell-1}. C'$ be a QBF with a $\PI{\ell-1}{}$-prefix and $C'$ in (3-)CNF.
Applying the reduction, we obtain an equivalent QBF~$\overline{Q} = \mathcal{R}(Q', \emptyset)$, effectively only using Formulas~(\ref{red:usatv}),(\ref{red:usatnot}),(\ref{red:usatnotd}).
The QBF~$\overline{Q}$ has a matrix~$C \lor D$, where $C$ is in 2-DNF and $D$ in $1$-CNF.
By negating $\overline{Q}$ (and flipping quantifier types), we get a QBF $Q$ with matrix $C \lor D$ where $C$ is in 2-CNF and $D$ in $1$-DNF.
Since $Q'$ was chosen arbitrarily, evaluating $\overline{Q}$ is $\PI{\ell-1}{P}$-hard, and thus evaluating $Q = \neg \overline{Q}$ is~$\SIGMA{{\ell-1}}{\Ptime}$-hard.
\end{proof}

\newpage
\subsection{Hardness Insights \& New Lower Bounds for Treedepth}
\label{appdx:tdpth}

The reduction~$\mathcal{R}_\tdx$ works with restricted treedepth decompositions, which we call \emph{$\alpha$-treedepth decompositions}.
For a given integer~$\alpha$, an $\alpha$-treedepth decomposition~$T$ consists of a path $S$, called \emph{main path}, where each node $v\in S$ has a tree of height at most~$\alpha$ attached to it.
Each leaf~$r$ of these trees is in turn connected to a path~$P_r$ of height at most~$\Card{S}$. Let~$\mathcal{P}$ denote the set of these paths. 

In addition to the variables introduced by~$\mathcal{R}$, we need the following sets of variables, thereby assuming an $\ell$-treedepth decomposition~$T$ of~$G_Q$ for the quantifier depth~$\ell$ of~$Q$.
For every~$P_r\in\mathcal{P}$ we require fresh index variables similar
to~$\lbvs{S}$ that are shared among different paths in~$\mathcal{P}$. We let
$\lbvsl{S}\eqdef \{idx_j^{\loc,1}, \ldots, idx_j^{\loc,\ceil{\log(\max_{P_r\in \mathcal{P}} \Card{P_r})}} \mid 1\leq j\leq 3\}$.
%
These~$2^{\ceil{\log(\max_{P_r\in{\mathcal{P}}} \Card{P_r})}}$ many combinations of variables per index~$j$ are sufficient to address any of the~$\Card{P_r}$ elements of~$P_r$.
Similar to above, we assign each element~$x\in P_r$ and each~$1\leq j\leq 3$ a set consisting of an arbitrary, but fixed \emph{combination of literals (unique within~$P_r$)} over these index variables~$idx_j^{\loc,1}, \ldots, idx_j^{\loc,\ceil{\log(\max_{P_r\in{\mathcal{P}}} \Card{P_r})}}$, which we denote by~$\bvali{x}{j}$.
Further, 
we also require three Boolean variables~$\lbvvsl{S}\eqdef \{val_1^\loc, val_2^\loc, val_3^\loc\}$, 
where~$val_j^\loc$ captures the \emph{truth value} for the element that is addressed via the variables in~$\lbvsl{S}$ for the~$j$-th index.
%
%
%
Finally, in order to set the \emph{context for these index variables} we define a set~$\lsel\eqdef\{sel_r^\loc \mid P_r\in{\mathcal{P}}\}$ of
\emph{selector variables}, where~$sel_r^\loc$ indicates whether the variables in~$\lbvsl{S}$ are used to address elements in~$P_r$. 
The reduction~$\mathcal{R}_\tdx$ takes~$Q$ and~$T$ (which gives rise to $S$ and~$\mathcal{P}$), 
and constructs
an instance~$Q'$ as well as a treedepth decomposition~$T'$ of~${Q'}$.
%
The constructed QBF equals $Q'\eqdef\exists V_1.\ \forall V_2.\ \cdots \exists V_\ell.\ 
\forall 
\lbvs{S}, \lbvvs{S},$
$\lsat, \lbvsl{S}, \lbvvsl{S}, \lsel
.\ C' \vee D', $
%
%
where $C'$
is in DNF, defined as a disjunction of terms:

\vspace{-.5em}
\begingroup
\allowdisplaybreaks
{\smallalign{\normalfont\small}
\begin{align}
	\label{red:guessatomj2}&x \wedge \bigwedge_{b \in \bval{x}{j}} b \wedge \neg val_j&\pushleft{\text{for each } x\in S, 1\leq j\leq 3}\qquad\raisetag{2.6em} \tag{\ref{red:guessatomj}}\\
	\label{red:guessnegatomj2}&\neg x \wedge \bigwedge_{b \in \bval{x}{j}} b \wedge val_j&\pushleft{\text{for each } x\in S, 1\leq j\leq 3}\qquad\raisetag{2.6em}\tag{\ref{red:guessnegatomj}}\\
%
%
%
%
%
	&sat_{i} \wedge l & \pushleft{\text{for each } c_i\in C, 1\leq j\leq 3\text{ with}}\notag\\[-.3em]
	&&\pushleft{\lit{c_i}{j}=l,
\var(l)\in \var(C)\setminus}\notag\\[-.3em]
\label{red:usatv2} && \pushleft{(\hspace{-.35em}\bigcup_{P_r\in{\mathcal{P}}}\hspace{-.6em}P_r \cup S)}\raisetag{1em}\\
	 &sat_{i} \wedge \neg b & \pushleft{\text{for each } c_i\in C, 1\leq j\leq 3, x\in S,}\notag\\[-.3em]
	\label{red:usatxp2}&& \pushleft{b\in \bval{x}{j}
\text{ with }\var(\lit{c_i}{j})=x}\tag{\ref{red:usatxp}}\\
%
	&sat_{i} \wedge val_j & \pushleft{\text{for each } c_i\in C, 1\leq j\leq 3, x\in S}\notag\\[-.3em] 
\label{red:usatxv2} &&\pushleft{\text{with }\lit{c_i}{j}=x}\tag{\ref{red:usatxv}}\raisetag{1.3em}\\
%
%
%
%
%
%
%
%
%
	&sat_{i} \wedge \neg val_j & \pushleft{\text{for each } c_i\in C, 1\leq j\leq 3, x\in S}\notag\\[-.3em]
\label{red:usatxnegv2} &&\pushleft{\text{with }\lit{c_i}{j}=\neg x} \tag{\ref{red:usatxnegv}}\raisetag{1.3em}\\[.5em]
%
%
%
%
%
%
%
	\label{red:usatvd2} &sat \wedge l & \pushleft{\text{for each } \{l\}\in D}\qquad\raisetag{1.3em}\tag{\ref{red:usatvd}} \\[1em]
	\label{red:guessatomloc}&sel_r^\loc \wedge x \wedge \hspace{-.75em}\bigwedge_{b \in \bvali{x}{j}} \hspace{-.75em}b \wedge \neg val_j^\loc&\pushleft{\text{for each } P_r\in\mathcal{P}, x\in P_r, 1{\,\leq\,} j{\,\leq\,} 3}\qquad\raisetag{1.6em}\\
	\label{red:guessnegatomloc}&sel_r^\loc \wedge \neg x \wedge\hspace{-.75em} \bigwedge_{b \in \bvali{x}{j}}\hspace{-.75em} b \wedge val_j^\loc&\pushleft{\text{for each } P_r\in\mathcal{P}, x\in P_r, 1{\,\leq\,}j{\,\leq\,}3}\qquad\raisetag{1.6em}\\
%
%
%
%
%
%
%
%
	&sat_{i} \wedge \neg b & \pushleft{\text{for each } c_i\in C, 1\leq j\leq 3, P_r}\notag\\[-.3em]
	&&\pushleft{\in\mathcal{P},x\in P_r, b\in \bvali{x}{j}
\text{ with}}\notag\\[-.3em]
	\label{red:usatxpl} &&\pushleft{\var(\lit{c_i}{j})=x}\raisetag{1.1em}\\
%
%
	 &sat_{i} \wedge val_j^\loc & \pushleft{\text{for each } c_i\in C, 1\leq j\leq 3, P_r}\notag\\[-.3em] 
	\label{red:usatxvl}&&\pushleft{\in\mathcal{P}, x\in P_r
\text{ with }\lit{c_i}{j}=x}\\
%
%
%
%
%
%
%
%
%
	 &sat_{i} \wedge \neg val_j^\loc & \pushleft{\text{for each } c_i\in C, 1\leq j\leq 3, P_r}\notag\\[-.3em] 
	\label{red:usatxnegvl} &&\pushleft{\in\mathcal{P}, x\in P_r
\text{ with }\lit{c_i}{j}{=}\neg x}\\
	&sat_{i} \wedge \neg sel_r^\loc & \pushleft{\text{for each } c_i\in C, P_r\in\mathcal{P},}\notag\\[-.3em]
	\label{red:usatsel} && \pushleft{\var(c_i)\cap P_r\neq\emptyset}\raisetag{1.1em}
%
%
%
%
%
%
\end{align}}\endgroup

\noindent Analogously to above, we define~$D'$ in 1-CNF, which is a conjunction of the following singletons.
\vspace{-.2em}
{\smallalign{\normalfont\small}
\begin{align}
	\label{red:usatnot2} &\pushleft{\neg sat_i} & \pushright{\makebox[2em]{}\text{for each }1\leq i \leq \Card{C}\hspace{3em}} \tag{\ref{red:usatnot}}\\
	\label{red:usatnotd2} &\pushleft{\neg sat} \tag{\ref{red:usatnotd}} 
\end{align}}
\vspace{-.85em}

Observe that thereby we basically reuse Formulas~(\ref{red:guessatomj2})--(\ref{red:usatvd2}), since Formulas~(\ref{red:usatv2}) are only a slight modification of~(\ref{red:usatv}).
Formulas~(\ref{red:guessatomloc}) and~(\ref{red:guessnegatomloc}) work similarly to~(\ref{red:guessatomj}) and~(\ref{red:guessnegatomj}), but synchronize indices and values for variables on a path in~$P_r\in\mathcal{P}$ and therefore additionally depend on the selector~$sel_r^\loc$ for~$P_r$ being true.
Then, Formulas~(\ref{red:usatxpl})--(\ref{red:usatxnegvl}) work analogously to~(\ref{red:usatxp})--(\ref{red:usatxnegv}), but for a specific path~$P_r\in\mathcal{P}$ (instead of~$S$).
Additionally, we only consider cases, where during determination of satisfiability ($sat_i$) of a clause~$c_i\in C$, also the 
corresponding selector~$sel_r^\loc$ for a path~$P_r\in\mathcal{P}$ holds, if~$P_r$ contains relevant variables. This is ensured by Formulas~(\ref{red:usatsel}).

Note that runtime as stated in Theorem~\ref{lab:runtime} works analogously for~$\mathcal{R}_\tdx$. 
Further, also the proof of correctness very similar to Theorem~\ref{lab:corr}, as observed below.
\begin{PROP}[Correctness]\label{lab:corrtd}
  Given a QBF~$Q$ of quantifier depth~$\ell$ with~$\matr(Q)=C\wedge D$ being in 3,1-CDNF and an $\ell$-treedepth decomposition~$T$ of~$G_Q$ comprising main path~$S$ and paths~$\mathcal{P}$. 
  Then, reduction $\mathcal{R}_\tdx(Q,T)$ computes an instance~$Q'$ that is equivalent to $Q$. In fact, any assignment~$\alpha$ to variables of~$\matr(Q)$ satisfies~$\matr(Q)$ if and only if every extension~$\alpha'$ of~$\alpha$ to variables~$\lbvs{S} \cup \lbvvs{S} \cup \lsat \cup \lbvsl{S} \cup \lbvvsl{S} \cup \lsel$ satisfies~$\matr(Q')$.
\end{PROP}
\begin{proof}[Proof (Sketch)]
Observe that the correctness proof is almost identical to the one of Theorem~\ref{lab:corr}.
The index variables in~$\lbvsl{S}$ and corresponding values~$\lbvvsl{}$ behave analogously (as~$\lbvs{S}$ and~$\lbvvs{}$ for the set~$S$).
The only difference lies in the additional selector variables~$\lsel$, where, crucially, Formulas~(\ref{red:usatsel})
ensure that satisfiability of a clause~$c_i$ is vacuously given whenever~$sat_i$ is set to~$1$,
but the corresponding selector variable~$sel_r^\loc$ for a path~$P_r\in\mathcal{P}$
containing variables of~$c_i$ is set to~$0$.
Note that setting more than one selector variables to~$1$ is also not an issue, since 
then, compared to one variable, 
potentially even more instances of Formulas~(\ref{red:guessatomloc}) and~(\ref{red:guessnegatomloc}) are~satisfied. 
\end{proof}
%

The missing ingredient for the lower bound is structure-awareness,
formalized as follows.

\begin{LEM}[Decreasing Treedepth
]\label{lab:sawpdp}
Given a QBF~$Q$ of quantifier depth~$\ell$ with~$\matr(Q)=C\wedge D$ being in 3,1-CDNF and an $\ell$-treedepth decomposition~$T$ of $G_Q$, comprising main path~$S$ and paths~$\mathcal{P}$, with $h=\Card{S}$. Then, reduction~$\mathcal{R}_\tdx$ on~$Q$ and~$T$ constructs a QBF~$Q'=\mathcal{R}_\tdx(Q,T)$ with an $(\ell{+}1)$-treedepth decomposition~$T'$ of~$G_{Q'}$ of height 
$\mathcal{O}(\log(h){+}\ell)$.
\end{LEM}
\begin{proof}
We show how~$\mathcal{R}_\tdx$ gives rise to an $(\ell{+}1)$-treedepth decomposition~$T'$ of height~$\mathcal{O}(\log(h){+}\ell)$, cf., Figure~\ref{fig:sketchtd}, depicting the special case of~$\mathcal{P}=\emptyset$. 
%
To this end, let~$S'\eqdef\lbvs{}\cup \lbvvs{} \cup\{sat\} \cup \lbvsl{} \cup \lbvvsl{}$ consist of auxiliary variables 
obtained by~$\mathcal{R}_\tdx$. 
We construct~$T'\eqdef(\var(Q'), F')$, where 
we assume any fixed total ordering~$s_1, \ldots, s_{\Card{S'}}$ among vertices~$s_i$ in~$S'$ with~$1\leq i\leq \Card{S'}$.
Let therefore~$F'$ be as follows: (i) For every~$1\leq i\leq \Card{S'}$, we define~$(s_i, s_{i+1})\in F'$, i.e., vertices of $S'$ are linked.
Additionally, (ii) we link variables in~$S$ to those in~$S'$, where for every~$s\in S$ we let~$(s_{\Card{S'}},s)\in F'$.
Then, (iii) we preserve edges within~$\var(Q) \setminus (\bigcup_{P_r\in\mathcal{P}} P_r \cup S)$, i.e., for every~$u,v\in\var(Q)\setminus (\bigcup_{P_r\in\mathcal{P}} P_r \cup S)$ with~$(u,v)\in F$, 
we let~$(u,v)\in F'$. 
Further, (iv) for every~$v\in\var(Q) \setminus (\bigcup_{P_r\in\mathcal{P}} P_r \cup S)$ with~$(s,v)\in F$ for some~$s\in S$, we construct~$(s_{\Card{S'}}, v)\in F'$.
We (v) link variables in~$\lsel{}$ to~$T'$, where for every~$P_r\in \mathcal{P}$ and~$x\in P_r$, we define $(r,sel_r^\loc)\in F'$ as well as~$(sel_r^\loc, x)\in F'$.
Finally, (vi) the variables in~$\lsat$ will be fresh leaves in~$F'$, where we define for every~$sat_j\in \lsat$ with~$c_j\in C$,
a unique descendant node of~$s_{\Card{S'}}$ in~$F'$: (1) if~$\var(c_j)\cap P_r\neq \emptyset$ for some~$P_r\in\mathcal{P}$, we require $(sel_r^\loc,sat_j)\in F'$, (2) if $\var(c_j)\subseteq S$, we define~$(s_{\Card{S'}}, sat_j)\in F'$, and (3) otherwise, let~$x\in\var(c_j)\setminus S$ be the variable such that $\var(c_j)\setminus S$ comprises only~$x$ and ancestors of~$x$ in $F$; then, we define~$(x, sat_j)\in F'$. 
Observe that by construction of~$T'$, the height of~$T'$ is bounded by~$\mathcal{O}(\log(h)+\ell)$,
since~$\Card{S'}\in\mathcal{O}(\log(h))$ and due to the height of~$T'-S'$ being in~$\mathcal{O}(\log(h)+\ell)$ since~$\Card{P_r}\in\mathcal{O}(h)$ for every~$P_r\in\mathcal{P}$. 
Further, note that the normalization of each term in Formulas~(\ref{red:guessatomj}) and~(\ref{red:guessnegatomj}) from DNF to 3-DNF as described in Section~\ref{sec:result} adds two paths of height~$\mathcal{O}(\log(h))$ to every~$x\in S$. 
Analogously, converting a term from Formulas~(\ref{red:guessatomloc}) and~(\ref{red:guessnegatomloc}) from DNF to 3-DNF adds two paths of height~$\mathcal{O}(\log(h))$ to every~$x\in P_r$ with~$P_r\in\mathcal{P}$. Since the variables in~$\bigcup_{P_r\in \mathcal{P}}P_r\cup S$ are not in an ancestor/descendant relationship, normalization preserves the height of~$T'$ ($\mathcal{O}(\log(h)+\ell)$).
Finally, note that $T'$ is indeed a well-defined ($\ell{+}1$)-treedepth decomposition of~$G_{Q'}$ and one can verify the formulas of~$\mathcal{R}_\tdx$ 
as follows. 
Formulas~(\ref{red:guessatomj}) and~(\ref{red:guessnegatomj}) are handled by (i) and (ii), i.e.,
variables of these terms are guaranteed to be in an ancestor/descendant relationship in~$T'$ by (i) and (ii).
By (iii), (iv), and (vi) we ensure that variables of Formulas~(\ref{red:usatv2}) are ancestors/descendants in~$T'$.
Observe that Formulas~(\ref{red:usatnot})--(\ref{red:usatnotd}) do not cause edges in~$G_{Q'}$,
Then, Formulas~(\ref{red:usatxp})--(\ref{red:usatxnegv}) are treated by (i) and (vi), and Formulas~(\ref{red:usatvd})
are handled by (i)--(v).
Formulas~(\ref{red:guessatomloc}) and~(\ref{red:guessnegatomloc}) are covered by (i) and (v). Further, Formulas~(\ref{red:usatxpl})--(\ref{red:usatxnegvl}) are treated by (i) and (vi) as well.
Finally, Formulas~(\ref{red:usatsel}) are covered by~$T'$ due to (v),
which completes the proof.
%
%
%
%
%
%
%
%
\end{proof}

\begin{restatetheorem}[lab:lb_pdp]
\begin{THM}[LB for $\ell$-Treedepth Decompositions] 
Given an arbitrary QBF~$Q$ in CDNF of quantifier depth~$\ell$ and an $\ell$-treedepth decomposition~$T$ of~$G_Q$ of height~$k=\pd{G_Q}$. 
Then, under ETH, $\QBFSAT_\ell$ on~$Q$ cannot be decided in time~$\tower(\ell, o(k-\ell))\cdot\poly(\Card{\var(Q)})$.
\end{THM}
\end{restatetheorem}
\begin{proof}
Without loss of generality, we assume~$Q$ in 3,1-CDNF.
The result with quantifier depth~$\ell=1$ corresponds to \SAT and therefore follows immediately by ETH,
since~$k\leq \Card{\var(Q)}$ and due to the fact that ETH implies that there is no algorithm for solving \SAT on 3-CNFs running in time~$2^{o(\Card{\var(Q)})}$.
For the case of~$\ell>1$, we apply induction.
Assume that the result holds for~$\ell-1$. 

Case 1: Innermost quantifier~$Q_{\ell}$ of~$Q$ is $\exists$.
Then, we apply the reduction~$\mathcal{R}_\tdx$ on~$(Q, T)$, where $S$ is the main path of~$T$. 
Thereby, we obtain a resulting instance~$Q'$ and a treedepth decomposition~$T'$ of~$G_{Q'}$ in time~$\mathcal{O}(\poly(\Card{\var(Q)})$,
cf.\ Theorem~\ref{lab:runtime}.
%
The reduction is correct by Proposition~\ref{lab:corrtd}, i.e., the set of 
satisfying assignments of~$\matr(Q)$ coincides with the set of satisfying assignments of~$\matr(Q')$ when restricted to variables~$\var(Q)$.
Further, we have that the height~$k'$ of~$T'$ is bounded by~$\mathcal{O}({\log(\Card{S})+\ell})$ by Lemma~\ref{lab:sawpdp}.
Assume towards a contradiction that despite ETH, $\QBFSAT$ on~$Q'$ can be decided in time~$\tower(\ell, o(k'-\ell))\cdot\poly(\Card{\var(Q')})$.
Then, the validity of~$Q$ can be decided in time~$\tower(\ell-1, o(k))\cdot\poly(\Card{\var(Q)})$, 
contradicting the induction hypothesis. 

Case 2: Innermost quantifier~$Q_\ell$ of~$Q$ is~$\forall$. We proceed similar to Case 1, but invert~$Q$ first, resulting in a QBF~$Q^\star$, whose quantifier blocks are flipped such that~$\matr(Q^\star)$ is in 3,1-CDNF. Then, as in Case 1, after applying reduction~$\mathcal{R}_\tdx$ on~$(Q^\star, S)$, we obtain~$Q'$. The lower bound follows similar to above, since the validity of~$Q'$ can be simply inverted in constant time. 
\end{proof}

\begin{OBS}\label{obs:exps}
  There exists a constant~$c$ such that for any integer~$\ell \geq 1$,
  we have that $\tower(\floor{\log(\ell)},0) \geq c \cdot \ell$.
\end{OBS}

\begin{restatetheorem}[lab:lb_tdp]
\begin{COR}[LB for Treedepth]
There exists a linear function~$f$ and an integer~$\ell_0>0$ such that for any
given QBF~$Q$ in 3,1-CDNF of quantifier depth~$\ell\geq \ell_0$, the following holds: 
Under ETH, 
$\QBFSAT_\ell$ on~$Q$ cannot be decided in time~$\tower(f(\ell), o(\td{G_Q}))\cdot\poly(\Card{\var(Q)})$.
\end{COR}
\end{restatetheorem}
\begin{proof}
Let~$k\eqdef \td{G_Q}$. 
Theorem~\ref{lab:lb_pdp} implies (A): Under ETH, $\QBFSAT_\ell$
on~$Q$ can not be solved in time~$\tower(\ell, o(k-\ell))\cdot\poly(\Card{\var(Q)})$.  
%
We proceed by case distinction.
Case~$\ell\leq g(k)$ for a linear function~$g$: 
By applying Consequence (A), we have that~$f(\ell)\eqdef \ell$ since~$o(k-\ell)=o(k)$.
%
%
Case~$\ell > g(k)$: 
Similar to the previous case, we apply Consequence (A). 
By Observation~\ref{obs:exps}, we have $\tower(\floor{\log(\ell)},0) \geq c \cdot \ell$ for some constant~$c$ and, 
consequently, (i) $\tower(\floor{\log(\ell)},0) \in \Omega(k)$ since $k< g^{-1}(\ell)$ by case assumption and $g^{-1}$ is a linear function. 
Further, we have (ii)~$2^{k-\ell} > 0$, even if~$\ell \gg k$. 
By combining (i) and (ii), we conclude~$f \in\Omega(\ell{-}\floor{\log(\ell)}{-}1)$, yielding the desired~result that $f$ is linear for sufficiently large~$\ell_0$.
In turn, this establishes the statement.
\end{proof}

\begin{restatetheorem}[lab:inc-tdp]
\begin{COR}[LB for Incidence Treedepth]
Given an arbitrary QBF~$Q$ with~$F=\matr(Q)$ in CNF (DNF) such that the innermost quantifier~$Q_\ell$ of~$Q$ is $Q_\ell=\exists$ ($Q_\ell=\forall$) 
and~$k=\td{I_Q}$. 
Then, under ETH, $\QBFSAT_\ell$ on~$Q$ cannot be decided in time~$\tower(\ell, o(k-\ell))\cdot\poly(\Card{\var(Q)})$.
\end{COR}
\end{restatetheorem}
\begin{proof}
The case for~$\ell=1$ immediately follows from ETH, so we assume~$\ell>1$.
We show that the treedepth of any QBF~$Q'=Q_1 V_1. \cdots Q_\ell V_\ell. F'$ in 3,1-CDNF linearly bounds the treedepth of~$I_F$, where~$F$ is obtained by converting~$\matr(Q')$ into CNF or DNF (exactly as defined in Corollary~\ref{lab:inc-lb}).
First, we compute a treedepth decomposition~$T'=(\var(F'), E')$
of~$G_{Q'}$ of treedepth~$k$ 
in time~$2^{\mathcal{O}(k^2)}\cdot\poly(\Card{\var(F')})$~\cite{ReidlEtAl14}.
%
Then, we construct a QBF~$Q\eqdef Q_1 V_1. \cdots Q_\ell V_\ell. F$, with~$F$ being defined as follows (see, Corollary~\ref{lab:inc-lb}).

Case 1: $Q_\ell=\exists$ and therefore $F'=C \wedge D$.
We define~$F\eqdef C\cup \{f\}$ in CNF with~$f\eqdef \{l \mid \{l\} \in D\}$ being a long clause. Next, we slightly adapt~$T'$, resulting in a treedepth decomposition~$T=(\var(F')\cup C \cup \{f\}, E' \cup E)$ of $I_F$, where~$E$ is defined as follows. First, $f$ is the new root of~$T$, i.e., for the root~$r$ of~$T'$, we define~$(f,r)\in E$. Further, for every clause~$c\in C$ and a variable~$v\in\var(c)$ such that every other variable in~$\var(c)$ is an ancestor of~$v$ in~$T'$, we add an edge~$(v,c)\in E$. Observe that the height of~$T$ is bounded by~$k+2$. As a result, assuming ETH and that we can decide the validity of~$Q$ in time~$\tower(\ell, o({k-\ell}))\cdot\poly(\Card{\var(Q)})$ contradicts 
Theorem~\ref{lab:lb_pdp}.

Case 2: $Q_\ell=\forall$, i.e., $F'=D \vee C$.
We define~$F\eqdef D\cup \{f\}$ in DNF with~$f\eqdef \{l \mid \{l\} \in C\}$ and proceed as in Case 1.
%
%
%
\end{proof}

\section{Omitted Proofs for Section~\ref{sec:vcfen-inc} (Algorithms for \QBFSATCNF Using Vertex Cover Number and Feedback Edge Number)}

\begin{restatetheorem}[thm:incvco]
\begin{THM}
  Given any CQBF~$Q$ of $\QBFSAT$ with~$k$ being the vertex cover
  number of~${I_Q}$. Then, the validity of~$Q$ can be decided in
  time~$1.709^{3^k}\cdot\poly(\Card{\var(Q)})$.
\end{THM}\end{restatetheorem}
\begin{proof}
  Let $C$ be a vertex cover of $I_Q$ containing $k_1$ variables and $k_2$ clauses. Then, $Q$
  can contain at most $3^{k_1}+k_2$ distinct clauses. This is because the variables of every clause outside of $C$
  must be a subset of $C$ and every other clause is contained in $C$. The lemma now follows
  because every CQBF with $m$ clauses can be solved in time $\bigoh(1.709^m)$~\cite{Williams02}.
\end{proof}

We now give a formal proof of Theorem~\ref{thm:kernel} and Theorem~\ref{thm:kernel-inc} starting with the former.
Our
kernelization  starts by computing a smallest feedback edge
set $D$ for $G_\qbfformula$ using
Proposition~\ref{pro:comp-fes}. We denote by $V(D)$ the set of
vertices incident with the edges of $D$.  
Let $H$ be
the graph $G_\qbfformula-D$.
Then $H$ is a spanning forest of $G_\qbfformula$. 

We start with some simple reduction rules that are applicable to any
$\QBFCNF$. In particular, we observe that we can assume that
$\qbfformula$ does not contain a \emph{unit clause}, i.e., a clause containing
exactly one literal, or a
\emph{pure literal}, i.e., a variable that either only appears
positively or only occurs negatively in $\qbfformula$. 
Towards showing the former assume that $\qbfformula$ contains a
unit clause $C$ on some variable $v$. If $v$ is universal, then $\qbfformula$
is unsatisfiable and we can return a constant size unsatisfiable
formula as our kernel. If, on the other hand, $v$ is existential, then $\qbfformula$
is equivalent to the formula $\qbfformula[v=b]$, where $b$ is $1$ if $v$
occurs positively in $C$ and $b=0$ if $v$ occurs negatively in $C$;
note that in both cases the size of a feedback edge set does not
increase (because $G_{\qbfformula[v=b]}$ is a subgraph of $G_\qbfformula$) and the
number of variables in the formula decreases.

Note that every vertex $v$ of $H$ that is not in $V(D)$ is only
contained in clauses of size at most two. Moreover, since $\qbfformula$ does
no longer contain unit clauses, it follows that every such vertex $v$
is only contained in clauses of size exactly two.

 We now provide a
series of reduction rules that allow us to obtain a reduced instance,
where $u$ and $v$ are contained together in exactly one clause of size two. The following lemma allows us to handle the cases that $u$
and $v$ are contained together in 3 or 4 clauses.

\begin{LEM}
\label{lem:fes-3-4-clauses}    
  Let $e=\{u,v\}$ be a clean edge of $H$. Then:
  \begin{itemize}
  \item[(1)] If $u$ and $v$ are contained in 4 clauses of $\qbfformula$, then
    $\qbfformula$ is unsatisfiable.
  \item[(2)] If $u$ and $v$ are contained in 3 clauses of $\qbfformula$, then
    either $\qbfformula$ is unsatisfiable or there is an assignment
    $\tau: \{u,v\} \rightarrow \{0,1\}$ such that $\qbfformula$ is equivalent
    to $\qbfformula[\tau]$.
  \end{itemize}
\end{LEM}
\begin{proof}
 Property (1) is trivial because if $u$ and $v$ are contained in four
 clauses of $\qbfformula$, then there is no assignment of $u$ and $v$
 satisfying all four of these clauses and therefore $\qbfformula$ is
 unsatisfiable.

 Towards showing Property (2), suppose that $u$ and $v$ are contained
 in exactly three clauses of $\qbfformula$. Then, there is a unique
 assignment $\tau : \{u,v\}\rightarrow \{0,1\}$ that satisfies all
 three clauses. Therefore, if one of the variables $u$ and $v$ is
 universal, $\qbfformula$ is unsatisfiable. Moreover, otherwise both
 variables are
 existential and therefore $\qbfformula$ is equivalent with $\qbfformula[\tau]$.
\end{proof}
Note that Lemma~\ref{lem:fes-3-4-clauses}, in particular an exhaustive
application of the lemma to every clean edge of $H$, shows that in
polynomial time we can either compute a constant size kernel or an
equivalent instance $\qbfformula'$ such that $G_{\qbfformula'}$ is a subgraph of
$G_\qbfformula$ and the endpoints of every clean edge of $H$ are contained
together in at most 2 clauses. Therefore, in the following we can
assume that the endpoints of every clean edge of $H$ appear together
in at most 2 clauses of $\qbfformula$. The following lemma now strengthens
this to the point where we can assume that the endpoints of every
clean edge of $H$ appear together in exactly one clause of $\qbfformula$.

\newcommand{\renam}[2]{#1|(#2)}
Before we present the lemma, we need the following notation that
allows us to replace (substitute) a literal with another one.
Let $l$ and $l'$ be two literals on distinct variables. We define
$\renam{\qbfformula}{l\rightarrow l'}$ to be the formula obtained from $\qbfformula$ after
replacing every occurrence of $l$ ($\neg l$) in the matrix of $\qbfformula$
with $l'$ ($\neg l'$) and removing the variable of $l$ from the
prefix of $\qbfformula$.
\begin{LEM}
  \label{lem:fes-2-clauses}
  Let $e=\{u,v\}$ be a clean edge of $H$ such that $u$ and $v$ occur together in exactly 2
  clauses of $\qbfformula$. Then, either:
  \begin{itemize}
  \item $\qbfformula$ is unsatisfiable, or
  \item there is an assignment $\tau$ of either $u$ or $v$ such that
    $\qbfformula$ is equivalent to $\qbfformula[\tau]$, or
  \item $\qbfformula$ is equivalent to $\renam{\qbfformula}{l\rightarrow l'}$,
    where $l$ and $l'$ are literals on $u$ and $v$ and $l$ and $l'$
    are on distinct variables. 
  \end{itemize}
\end{LEM}
\begin{proof}
  Because $u$ and $v$ are contained together in exactly 2 clauses (and
  these clauses do not contain any other variables), say
  $C_1$ and $C_2$, of
  $\qbfformula$, there are exactly two assignments $\tau_1 :
  \{u,v\}\rightarrow\{0,1\}$ and $\tau_2 : \{u,v\}\rightarrow\{0,1\}$
  that satisfy these two clauses. We distinguish the following
  cases. If there is a variable $w \in \{u,v\}$ such that
  $\tau_1(w)=\tau_2(w)$, then the only way to satisfy $C_1$ and $C_2$
  is to set the variable $w$ to $\tau_1(w)$. Therefore, if $w$ is
  universal, then $\qbfformula$ is not satisfiable. Moreover, if $w$ is
  existential, then $\qbfformula$ is equivalent to $\qbfformula[\tau]$, where
  $\tau$ is the assignment setting $w$ to $\tau_1(w)$.

  Otherwise, it holds that $\tau_1(u)\neq \tau_2(u)$ and
  $\tau_1(v)\neq \tau_2(v)$. Therefore, either $\tau_i(u)=\tau_i(v)$
  or $\tau_i(u)=\lnot \tau_i(v)$ for every $i$ with $1 \leq i \leq
  2$. Because both cases are similar, we will only show the statement
  for the former case (i.e., $\tau_i(u)=\tau_i(v)$).
  Without loss of generality, let $v$ be the variable occurring
  behind $u$ in the prefix of $\qbfformula$. If $v$ is universal, then $\qbfformula$ is
  unsatisfiable, because a universal winning strategy can ensure that
  $v$ is assigned different from $u$.
  If, on the other hand, $v$ is existential,
  then, we claim that $\qbfformula$ is equivalent to $\renam{\qbfformula}{v\rightarrow u}$.
  Towards showing the forward direction, let $\Tau$ be winning
  existential strategy for $\qbfformula$ and let $\delta$ be an arbitrary
  play of the universal player. Then, $\alpha(\Tau,\delta)$ satisfies
  the two clauses containing $u$ and $v$ and therefore
  $\alpha(\Tau,\delta)[u]=\alpha(\Tau,\delta)[v]$, which shows that
  $\alpha(\Tau,\delta)$ also satisfies $\renam{\qbfformula}{v\rightarrow
    u}$. If, on the other hand, $\Tau$ is an existential winning
  strategy for $\renam{\qbfformula}{v\rightarrow u}$, then the strategy
  $\Tau'$ obtained from $\Tau$ by setting $v$ to the value of $u$,
  which is possible because $v$ occurs after $u$ in the prefix of
  $\qbfformula$, is easily seen to be an existential winning strategy for
  $\qbfformula$.
\end{proof}
Therefore, due to Lemmas~\ref{lem:fes-2-clauses}
and~\ref{lem:fes-3-4-clauses}, we can from now on assume that the
endpoints of every clean edge of $H$ occur together in exactly 1
clause of $\qbfformula$ (and this clause contains no other variables).
The next lemma now allows us to remove leaves that are contained in a
clean edge of $H$.
\begin{LEM}
\label{lem:fes-leaf}
  Let $l$ be a leaf of $H$ and $l\notin V(D)$. Then, there is an
  assignment $\tau : \{l\}\rightarrow \{0,1\}$ such that $\qbfformula$ is equivalent to $\qbfformula[\tau]$.
\end{LEM}
\begin{proof}
  Because $l$ is a leaf of $H$, $l\notin V(D)$, and we can assume that $\qbfformula$ does not
  contain any unit clause, $l$ occurs only in clauses with its
  parent $p$ in $H$. Note also that $e=\{p,l\}$ 
is a clean edge of $H$,
  because $l$ is a leaf and $l \notin V(D)$. Therefore, $l$ is contained in exactly one clause $C$
  (together with $p$) of $\qbfformula$. But then, if $l$ is existential (universal), then $\qbfformula$
  is equivalent to the formula $\qbfformula[l=b]$ ($\qbfformula[l=\neg b]$), where $b$ is $1$ if $l$
  occurs positively in $C$ and $b=0$ if $l$ occurs negatively in $\qbfformula$.
\end{proof}
Note that after exhaustively applying Lemma~\ref{lem:fes-leaf}, we
obtain an instance such that every leaf of $H$ is one of the at most
$2|D|$ endpoints of the edges in $D$. Therefore, $H$ has at most
$2|D|$ leaves and therefore at most $2|D|-2$ vertices of degree at
least 3. It therefore, only remains to reduce the number of vertices
in $H$ with degree exactly 2, which we will achieve using the
following observation and lemma.

We say that a path $P$ of $H$ is \emph{clean} if all inner vertices of
$P$ have degree two in $H$ and are not in $V(D)$.
\begin{observation}
\label{obs:fes-path}
  Let $P$ be a clean path of $H$ with at least two inner
  vertices. Then, every edge of $P$ is clean. Moreover, if $v$ is an inner vertex
  of $P$, then, $v$ appears in exactly two clauses (both of size 2) in
  $\qbfformula$, i.e., one clause $C_u$ with its left neighbor $u$ in $P$
  and one clause $C_w$ with its right neighbor $w$ in $P$. Finally,
  $v,\neg v \in C_u\cup C_w$, i.e., $v$ appears complementary in
  $C_u$ and $C_w$.
\end{observation}
\begin{proof}
  Because no inner vertex of $P$ is in $V(D)$ and $P$ has at least two
  inner vertices, it follows that $G_\qbfformula$ cannot contain a triangle
  that contains an edge of $P$. Therefore, all edges of $P$ are clean.
  
  The fact that $v$ appears only in the clauses $C_u$ and $C_w$
  follows from the fact that $v$ has degree two in $H$ and our
  assumption that every clean edge $e$ of $H$ 
  corresponds to exactly one clause. Finally, $v,\neg v \in C_u\cup C_w$
  follows because otherwise $v$ would be a pure literal, but we
  assume that $\qbfformula$ does not contain a pure literal.
\end{proof}

\begin{LEM}
\label{lem:fes-degt}
  Let $P$ be a clean path of $H$ with at least two inner vertices.
  Let $v$ be the innermost variable in the prefix of $\qbfformula$ among all inner vertices
  of $P$ and let $C_u$ and $C_w$ be the two clauses (see
  Observation~\ref{obs:fes-path}) containing $v$
  (and its neighbors $u$ respectively $w$). Then, either:
  \begin{itemize}
  \item $v$ is existential and $\qbfformula$ is equivalent to the formula $\qbfformula_\exists$
    obtained from $\qbfformula$ after removing $C_u$ and $C_w$ and adding
    instead the clause $(C_u\cup C_w)\setminus\{v,\neg v\}$.
  \item $v$ is universal and $\qbfformula$ is equivalent to the formula $\qbfformula_\forall$
    obtained from $\qbfformula$ after removing $C_u$ and $C_w$ and adding
    instead the clauses $C_u\setminus\{v,\neg v\}$ and $C_w\setminus\{v,\neg v\}$.
  \end{itemize}
\end{LEM}
\begin{proof}
  Because of Observation~\ref{obs:fes-path}, we obtain that
  $C_u=l_u\lor l_v$ and $C_w=l_w\lor \neg l_v$ for some literals
  $l_u$, $l_w$, and $l_v$ on the variables $u$, $w$, and $v$,
  respectively and moreover $C_u$ and $C_w$ are the only clauses of $\qbfformula$
  containing $v$.
  
  We start by showing the lemma for the case that $v$ is
  existential. Towards showing the forward direction of the
  equivalence between $\qbfformula$ and $\qbfformula_\exists$ suppose that $\qbfformula$ is
  satisfiable and has an existential winning strategy $\Tau$. Then,
  for every play $\delta$ of the universal player, it holds that
  $\alpha(\Tau,\delta)$ must satisfy $C_u$ and $C_w$, which in turn
  implies that $\alpha(\Tau,\delta)$ satisfies $(C_u\cup
  C_v)\setminus\{v,\neg v\}$. Therefore, $\Tau$ is also an
  existential winning strategy for $\qbfformula_\exists$, showing that
  $\qbfformula_\exists$ is satisfiable. Now, suppose that $\qbfformula_\exists$ is
  satisfiable and let $\Tau$ be an existential winning strategy for
  $\qbfformula_\exists$ witnessing this. Because $P$ has at least two inner
  vertices, it holds that either $u$ or $w$ is an inner vertex of
  $P$; in the following we assume without loss of generality that $u$
  is an inner vertex of $P$. Moreover, because $v$ is the innermost variable in the prefix of $\qbfformula$ (among all inner vertices of $P$), it
  follows that $v$ appears after $u$ in the prefix of $\qbfformula$. Let
  $\Tau'$ be the strategy for $\qbfformula$ obtained from $\Tau$ such that
  $\Tau'$ sets $v$ to $l_v$ if the assignment for $u$ does not satisfy
  $C_u$ (note again that because $v$ is after $u$ in the quantifier
  prefix, the assignment for $u$ is already fixed once $\Tau'$ has to
  decide on an assignment for $v$) and otherwise $v$ is set to
  $\neg l_v$. We claim that $\Tau'$ is an existential winning
  strategy for $\qbfformula$. Towards showing this first note that because of
  the clause $(C_u\cup C_w)\setminus\{v,\neg v\}$ in $\qbfformula_\exists$,
  it holds that $\alpha(\Tau,\delta)$ satisfies either $l_u$ or $l_w$
  for every universal play $\delta$. Now, if $\alpha(\Tau,\delta)$
  satisfies $l_u$ ($l_w$), then for $\alpha(\Tau',\delta)$ the clause $C_u$ is
  satisfied by $l_u$ ($l_v$) and the clause $C_w$ is satisfied by
  $\neg l_v$ ($l_w$). This shows that $\Tau'$ is an existential
  winning strategy for $\qbfformula$ and therefore $\qbfformula$ is satisfiable as
  required.
  
  We now continue with the case that $v$ is universal, i.e., we will
  show that $\qbfformula$ is unsatisfiable if and only if so is
  $\qbfformula_\forall$. So suppose that $\qbfformula$ is unsatisfiable and let
  $\Lambda$ be a universal winning strategy for $\qbfformula$. We claim that
  $\Lambda$ is also a universal winning strategy for
  $\qbfformula_\forall$. Towards showing this let $\delta$ be any play of the
  existential player. Then, $\alpha(\Lambda,\delta)$ does not satisfy
  some clause $C$ of $\qbfformula$. If $C$ is also in $\qbfformula_\exists$, then
  also $\qbfformula_\forall$ is not satisfied. Otherwise, $C=C_u$ (or $C=C_w$) and therefore
  $\alpha(\Lambda,\delta)$ does not satisfy $l_u$ ($l_w$), which shows that
  $\alpha(\Lambda,\delta)$ does not satisfy $\qbfformula_\forall$, as
  required.
  Suppose now that $\qbfformula_\forall$ is not satisfiable and let $\Lambda$
  be a universal winning strategy for $\qbfformula_\forall$.
  Because $P$ has at least two inner
  vertices, it holds that either $u$ or $w$ is an inner vertex of
  $P$; in the following we assume without loss of generality that $u$
  is an inner vertex of $P$. Moreover, because $v$ is the innermost variable in the prefix of $\qbfformula$, it
  follows that $v$ appears after $u$ in the prefix of $\qbfformula$. Let
  $\Lambda'$ be the strategy for $\qbfformula$ obtained from $\Lambda$ such that
  $\Lambda'$ sets $v$ to $\neg l_v$ if the assignment for $u$ does not satisfy
  $C_u$ (note again that because $v$ is after $u$ in the quantifier
  prefix, the assignment for $u$ is already fixed once $\Lambda'$ has to
  decide on an assignment for $v$) and otherwise $v$ is set to
  $l_v$. We claim that $\Lambda'$ is a universal winning
  strategy for $\qbfformula$. Towards showing this let $\delta$ be any play
  of the existential player on $\qbfformula_\forall$ (and also on
  $\qbfformula$). Because $\Lambda'$ wins on $\qbfformula_\forall$, it holds that
  $\alpha(\Lambda',\delta)$ does not satisfy some clause $C$ of
  $\qbfformula_\forall$. Clearly, if $C$ is also in $\qbfformula$, then there is
  nothing to show. Otherwise, $C=l_u$ (or $C=l_w$) and therefore
  $\alpha(\Lambda',\delta)$ does not satisfy $C_u$ (or $C_v$), as
  required. Therefore, $\Lambda'$ is a universal winning strategy for
  $\qbfformula$ and therefore $\qbfformula$ is unsatisfiable, as required.
\end{proof}

Note that after an exhaustive application of Lemma~\ref{lem:fes-degt},
we obtain an instance such that every clean
path in $H$ has at most 2 inner vertices. We are now ready to prove
Theorem~\ref{thm:kernel}.

\begin{proof}[Proof of Theorem~\ref{thm:kernel}]
  Let $\qbfformula$ be the instance obtained from the original formula after
  applying all of the above defined reduction rules
  exhaustively. Moreover, let $G$ be the
  primal graph of $\qbfformula$, $D$ be a smallest FES of $G$, $k=|D|$, and let $H=G-D$.
  Then, every clean
  path in $H$ has at most 2 inner vertices. Moreover, the number of
  (inclusion-wise) maximal clean paths in $H$ is at most
  $2k+2k-2-1=4k-3$. This is because, as argued above, $H$ has at most
  $2k$ leaves and at most $2k-2$ vertices of degree larger than $2$ and
  any tree with $2k+2k-2$ vertices has at most $2k+2k-3$
  edges. Therefore, the total number of vertices (edges) of $H$ is at most
  $2k+2k-2+2(4k-3)=12k-8$ ($12k-9$). This already bounds the number of
  variables of $\qbfformula$ to be at most $12k-8$ and it remains to bound the
  number of clauses of $\qbfformula$. Because every edge of $H$ whose endpoints
  are not in $D$ is clean, we obtain that all but at most $2k$ edges
  of $H$ are clean and therefore contribute at most one clause to $\qbfformula$.
  Moreover, in the worst case (w.r.t. the number of clauses that can be
  obtained), the remaining $2k$ edges of $H$ together with the edges
  in $|D|$ can form a clique of size at most $\lfloor (\sqrt{24k+1}+1)/2\rfloor$
  (this is because the number $n$ of vertices of a clique with $2k+k=3k$
  edges satisfies $\binom{n}{2}=3k$). Therefore, the kernel has at
  most $12k-9-2k=10k-9$ clauses coming from the clean edges of $H$ and
  at most $3^{\lfloor (\sqrt{24k+1}+1)/2\rfloor}$ clauses coming from the
  remaining edges of $G$, which gives a total of at most
  $10k-9+3^{\lfloor (\sqrt{24k+1}+1)/2\rfloor}$ clauses for $\qbfformula$.
  
  Note that if we restrict ourselves to $c$-$\QBFCNF$, then the
  resulting kernel can only have cliques of size at most $c$. Therefore,
  the $3k$ edges that are not clean in $G$ can contribute to at most
  $3k/\binom{c}{2}$ cliques of size at most $c$ each, which reduces the
  number of clauses resulting from these edges to
  $3^c(3k/\binom{c}{2})$, which is linear in $k$.
\end{proof}

\begin{proof}[Proof of Theorem~\ref{thm:kernel-inc}]
  Let $\qbfformula$ be the given CQBF formula, let $D\subseteq
  E(I_Q)$ be a FES for $I_Q$, i.e.,  the graph $I_Q-D$ in the
  following denoted by $F$ is a forest.
  As in the case of the primal graph, we
  first observe that we can assume that $\qbfformula$ does not contain
  any unit clauses or pure literals. 
  Consider a leaf $l$ that is not in
  $V(D)$. Then, $l$ is not a clause because $\qbfformula$ does not
  contain unit clauses and $l$ is also not a variable because
  $\qbfformula$ does not contain pure literals. Therefore, every leaf of
  $F$ is in $V(D)$, which implies that $F$ has at most $2k$ leaves.
  Since this implies that $F$ has at most $2k-2$ vertices of degree
  larger than $2$, it only remains to bound the number of vertices of
  $F$ of degree exactly $2$. As for the primal graph, we say that a
  path $P$ of $F$ is \emph{clean} if all inner vertices of $P$ have
  degree two in $F$ and are not in $V(D)$. Let $P$ be a clean path and
  let $v$ be a variable that is also an inner vertex of $P$ such that
  both neighbors of $v$ are also inner vertices of $P$. Then, $v$
  satisfies all properties implied by Observation~\ref{obs:fes-path},
  i.e., $v$ is contained in exactly two clauses $C_u$ and $C_w$ both
  of arity $2$ (i.e., its neighbors on $P$) and $v$ appears
  complementary in $C_u$ and $C_w$. Therefore, for any clean path
  containing such a variable $v$, we can apply
  Lemma~\ref{lem:fes-degt} to eliminate at least one variable on
  $P$. Since such a variable always exists if $P$ has at least $5$
  inner vertices, we can assume that every clean path has at most $5$
  inner vertices. Using the same argument as in the case of the primal
  graph, we obtain that $F$ contains at most $4k-3$(inclusion-wise)
  maximal clean paths. Consequently, $F$ has at most
  $5(4k-3)+2k+2k-2=24k-17$ vertices and therefore the reduced formula
  $\qbfformula$ has at most that many variables and clauses.
\end{proof}

\section{Omitted Proofs for Section~\ref{sec:tdp} (Exploring the Limits of Tractabilty for \QBFSATCNF on Primal Graphs)}

\begin{restatetheorem}[pro:computedelsetgeneral]
\begin{PROP}
  Let $\mP$ be any efficiently computable property and let $\qbfformula$ be a
  \QBFCNF{}. Then, computing a smallest $c$-deletion set $D$
  of $\qbfformula$ that satisfies $\mP(\qbfformula,D,c)$ is fixed-parameter tractable
  parameterized by $|D|+c$.
\end{PROP}
\end{restatetheorem}
\begin{proof}
  It is shown in~\cite[Theorem 12]{DBLP:journals/algorithmica/DrangeDH16} that
  deciding whether a graph $G$ has a $c$-deletion set of size at most
  $k$ is fixed-parameter tractable parameterized by $k+c$. The proof
  uses a bounded-depth search tree algorithm that can actually be used
  to enumerate all possible $c$-deletion sets of size at most $k$ in
  fpt-time parameterized by $k+c$. Therefore, using this algorithm, we
  can enumerate all $c$-deletion sets of size at most $k$ of a
  \QBFCNF{} $\qbfformula$ in the
  required time. Then, for each such $c$-deletion set $D$ we can use the
  algorithm for fact that the property $\mP$ is efficiently
  computable to decide whether $\mP(\qbfformula,D,c)$ is true or false.
  We can then return the smallest set such that $\mP(\qbfformula,D,c)$ is
  true or return false if no such set exists. Finally, by starting with $k=1$
  and increasing $k$ by long as the algorithm returns false, we can
  find a smallest $c$-deletion set for $\qbfformula$.
\end{proof}

We need the following notions.
Let $\qbfformula$ be a
$\QBFCNF$, $c$ an integer and $D \subseteq \var(\qbfformula)$ be a
$c$-deletion set of $\qbfformula$. We define
the following equivalence relation $\sim$ over the set of components
of $G_F-D$. That is, $C\sim C'$ for two components $C$ and $C'$ of
$G_F-D$ if there is a bijection $\eta : \var(C) \rightarrow
\var(C')$ such that $\qbfind{\qbfformula}{C'\cup D}$ is equal to the formula obtained
from $\qbfind{\qbfformula}{C\cup D}$ after renaming all variables in $C$ according to
$\eta$ (into variables of $C'$). Therefore, we let, for a subgraph or a set of vertices of $G_\qbfformula$,
$\qbfind{\qbfformula}{A}$ be the \QBFCNF{} obtained
from $\qbfformula$ after removing all variables outside of $A$
together with all clauses that have at least one variable outside of $A$.
We say two components $C$ and $C'$ with $C\sim C'$ have the same
\emph{(component) type} and denote by $\ctypes(\qbfformula,D)$ the set of \emph{all
component types} of $\qbfformula$ (with respect to $D$). Moreover, for a type $t \in \ctypes(\qbfformula,D)$,
we denote by $\CbyT(\qbfformula,D,t)$, the set of \emph{all components} of $\qbfformula$
having type $t$.

\begin{PROP}
\label{pro:ds-ct-nr}
  The relation $\sim$ has at most $(2|D|+1)^c2^{3^{c+|D|}}$ equivalence classes.
\end{PROP}
\begin{proof}
  Let $C$ be a component of $G_\qbfformula-D$. Then, the type of $C$ is
  completely characterized by where (relative to the variables in $D$)
  and how (existential or universal) the variables of $C$ occur in the
  prefix of $\qbfformula$ as well as the set of clauses containing variables
  of $C$. Since every variable $v$ of $C$ can be placed in at most
  $|D|+1$ distinct positions in the prefix of $\qbfformula$ and can only be
  quantified either existentially or universally, it follows that there
  are at most $(2|D|+1)^c$ distinct ways that the variables of $C$ can
  appear in the prefix of $\qbfformula$. Moreover, since variables of $C$ can
  only appear together in a clause with variables in $C\cup D$, we
  obtain that there are at most $3^{c+|D|}$ distinct clauses
  containing variables of $C$ and therefore at $2^{3^{c+|D|}}$
  distinct sets of such clauses. Therefore, there are at most
  $(2|D|+1)^c2^{3^{c+|D|}}$ distinct types of components, as required.
\end{proof}

\begin{PROP}
  \label{pro:delset-comptypes}
  Let $\qbfformula$ be a $\QBFCNF$
  and let $D$ be a $c$-deletion set for $G_\qbfformula$. Then, we can
  compute $\ctypes(\qbfformula,D)$ and $\CbyT(\qbfformula,D,t)$ for every $t \in
  \ctypes(\qbfformula,D)$ in polynomial time.
\end{PROP}
\begin{proof}
  Let $C_1,\dotsc,C_m$ be the components of $G_\qbfformula-D$. We start by
  setting $T=\emptyset$ and then for every $i$ with $1 \leq i \leq m$,
  we do the following: If $T=\emptyset$, then we add the set $\{C_i\}$
  to $T$. Otherwise, we go through every $t \in T$ and test whether
  $C_i$ is of type $t$ as follows. Let $C_t \in t$ be arbitrary. If $|\var(C_i)|\neq
  |\var(C_t)|$, then $C_i$ is not of type $t$ and we proceed to the next $t \in T$. Otherwise, $\eta : \var(C_i)
  \rightarrow \var(C_t)$ be the unique mapping that maps the $i$-the
  variable of $C_i$ in the prefix of $\qbfind{\qbfformula}{C_i}$ to the $i$-th variable of
  $C_t$ in the prefix of $\qbfind{\qbfformula}{C_t}$. We then check, whether
  $\qbfind{\qbfformula}{C_t\cup D}$ is equal to the formula obtained from
  $\qbfind{\qbfformula}{C_i\cup D}$ after renaming every variable in $C_i$ according
  to $\eta$. Note that this can be achieved in polynomial time by
  going over all variables and clauses of $\qbfind{\qbfformula}{C_t\cup D}$. If this
  is the case, then $C_i$ is of type $t$ and we add $C_i$ to $t$, otherwise we proceed to the next
  $t \in T$. Finally, if there is no $t \in T$ such that $C_i$ is of
  type $t$, then we add the set $\{C_i\}$ to $T$. After having
  considered all components in this manner, the resulting set $T$ is
  equal to $\ctypes(\qbfformula,D)$ and for every $t \in T$ the set
  $\CbyT(\qbfformula,D,t)$ is equal to $t$.
\end{proof}

We now show that we can eliminate all universal variables in a
$c$-deletion set $D$ of a $\QBFCNF$ without losing the
structure of the formula, i.e., after eliminating all universal
variables we obtain a formula $\qbfformula'$ and a $2^cc$-deletion set of
size at most $2^cc$.
\begin{restatetheorem}[pro:delset-exists]
\begin{PROP}
  Let $\qbfformula$ be a $\QBFCNF$ and let $D$ be a $c$-deletion set for $\qbfformula$. Then, in time
  $\mathcal{O}(2^u\CCard{\qbfformula})$, where $u=|D\cap \varu(\qbfformula)|$, we can construct an equivalent $\QBFCNF$
  $\qbfformula'$ and a set $D'\subseteq \vare(\qbfformula')$ with $|D'|\leq 2^u|D|$
  such that $D'$ is a $2^uc$-deletion set for $\qbfformula'$.
\end{PROP}
\end{restatetheorem}
\begin{proof}
  Let $\qbfformula=Q_{1}v_1Q_2 v_2\cdots Q_n v_n F$.
  If $\varu(D)=\emptyset$, then we simple return $\qbfformula'=\qbfformula$ and
  $D'=D$. Otherwise, let $v_i \in \varu(D)$ be the universal variable in $D$
  that is the innermost variable in the prefix of $\qbfformula$. We will use quantifier
  expansion to eliminate $v_i$. That is, let $\qbfformula(v_i)$ be the formula
  obtained from $\qbfformula$ after eliminating $v_i$, i.e.:

  \[\begin{array}{cc}
      \qbfformula(v_i) = & \pushleft{Q_{1}v_1\cdots Q_{i-1} v_{i-1}((Q_{i+1}v_{i+1}\cdots Q_{n}
      v_{n}) }\\&\pushleft{F[v_i=0]) \land (Q_{i+1}v_{i+1}\cdots Q_{n}v_{n}
      F[v_i=1]))}
    \end{array}
  \]
  Note that $\qbfformula(v_i)$ is equivalent with $\qbfformula$ and can be computed
  in time $\bigO(\CCard{\qbfformula})$. However,
  $\qbfformula(v_i)$ is not in prenex normal form. To bring $\qbfformula_{v_i}$
  into prenex normal form, we introduce a copies $v'$ for every
  variable $v \in \{v_{i+1},\dotsc,v_n\}$ and we then rewrite
  $\qbfformula(v_i)$ into the equivalent formula $\qbfformula_c(v_i)$ given as:
  \[\begin{array}{cc}
      \qbfformula_c(v_i) = & \pushleft{Q_{1}v_1\cdots Q_{i-1} v_{i-1}Q_{i+1}v_{i+1}\cdots Q_{n}v_n}\\&\pushleft{Q_{i+1}v_{i+1}'\cdots Q_{n}
      v_{n}' F[v_i=0]) \land F'[v_i=1]}
    \end{array}
  \]
  where $F'$ is the CNF formula obtained from $F$ after renaming every
  occurrence of a variable $v \in \{v_{i+1},\dotsc,v_n\}$ to $v'$. Then:
  \begin{itemize}
  \item $\qbfformula_c(v_i)$ has at most twice the size of $\qbfformula$ and can be
    computed in time $\bigO(\CCard{\qbfformula})$,
  \item Let $A$ be the set of all (existential) variables in $D$
    that occur after $v_i$ in the prefix of $\qbfformula$ and let $A'=\SB
    v' \SM v \in A\SE$. Then, $D''=(D\setminus \{v_i\})\cup A'$ is a $2c$-deletion set of
    $\qbfformula_{v_i}'$. Moreover, $D''$ contains one less universal
    variable than $D$. Therefore, $\qbfformula_c(v_i)$ has a $2c$-deletion
    set that is at most twice the size of $D$ and contains one less
    universal variable.
  \end{itemize}
  It follows that if we repeat the above process for every universal
  variable in $D$ in the reverse order those are occurring in the
  prefix of $\qbfformula$, we obtain an equivalent formula $\qbfformula'$ having a
  $2^uc$-deletion set $D'$ of size at most $2^u|D|$ with $D' \subseteq
  \vare(\qbfformula')$ in time $\bigO(\CCard{\qbfformula'})=\bigO(2^u\CCard{\qbfformula})$.
\end{proof}

}

\subsection{An Algorithm for Components of Type $\exists^{\leq 1}\forall^u$}

\begin{LEM} 
  \label{lem:delset-ef-red}
  Let $\qbfformula$ be a $\QBFCNF$ and let $D \subseteq \var(\qbfformula)$ be a
  $c$-deletion set for $\qbfformula$. Let $t \in \ctypes(\qbfformula)$ be of the
  form $\exists^{\leq 1}\forall$ with $\CbyT(\qbfformula,D,t)=\{A_1,\dotsc,A_{r}\}$,
  where $(A_1,\dotsc,A_r)$ is the ordering of the components of type
  $t$ according to the occurrence of the unique existential variable in
  the component.
  Then, $\qbfformula$ is equivalent to the formula $\qbfformula'$ obtained from
  $\qbfformula$ after removing all clauses containing
  variables of the components $\{A_{2^{c}},\dotsc,A_r\}$ (and
  all variables of the components $\{A_{2^{c}},\dotsc,A_r\}$).
\end{LEM}
\begin{proof}
  We will show the lemma for the case that $t$ is of the form
  $\exists\forall^u$, which is the most general case, where $u$ is the
  number of universal variables in the components of type $t$.
  Let $x_i$ be the unique existential variable of $A_i$ and let $Y_i$
  be the set of all universal variables of $A_i$. Note that
  $x_1,\dotsc,x_r$ is the order the existential variables occur in the
  prefix of $\qbfformula$. Moreover, let $X=\{x_1,\dotsc,x_n\}$, $X_P=\{x_1,\dotsc,x_{2^{u+1}-1}\}$, and let
  $Y=\bigcup_{i=1}^rY_i$.
  
  We show the lemma by showing that $\qbfformula$ has an existential winning
  strategy if and only if $\qbfformula'$ does. The forwards direction of
  the claim is trivial, because the clauses of $\qbfformula'$ are a subset of
  the clauses of of
  $\qbfformula$. Towards showing the reverse direction, let $\Tau'=(\tau_x':
  \{0,1\}^{\var(\qbfformula')_{<x}^{\forall}}\rightarrow \{0,1\})_{x\in \vare(\qbfformula')}$ be
  an existential winning strategy for $\qbfformula'$. We will show how to
  construct an existential winning strategy $\Tau=(\tau_x:
  \{0,1\}^{\var(\qbfformula)_{<x}^{\forall}}\rightarrow \{0,1\})_{x\in \vare(\qbfformula)}$ for
  $\qbfformula$.

  Let $x \in \vare(\qbfformula)$ and let $\delta : V^\forall_{<x}
  \rightarrow \{0,1\}$ be any play of the universal variables that
  occur before $x$ in the prefix of $\qbfformula$. Let $\delta'$ be the
  restriction of $\delta$ to the variables of $\qbfformula'$, i.e.,
  $\delta'=\delta_{\var(\qbfformula')}$. We now define a stronger play
  $\delta^S$ for the variables in $\var(\delta')$ iteratively as
  follows. Let $(y_1,\dotsc,y_n)$ be the variables in $\var(\delta')$
  as they occur in the quantifier prefix of $\qbfformula'$. We set
  $\delta^S=\delta^S_n$ and $\delta^S_0=\delta'$.
  Moreover, for every $i$ with $0<i\leq n$,
  $\delta^S_i$ is obtained from $\delta^S_{i-1}$ as follows.
  \begin{itemize}
  \item If $y_i$ is not a variable in a component of type $t$, i.e.,
    $y_i \notin Y$, then we set $\delta^S_i=\delta^S_{i-1}$.
  \item Otherwise, $y_i$ is the $a$-th universally quantified variable
    of a component $A_l$ for some $1\leq a \leq u$ and $1 \leq l \leq
    2^{u+1}-1$. We start by setting $\delta_i^S=\delta_{i-1}^S$.
    Let $\beta : X_P \cap \var(\delta_{i-1}^S)$ be the assignment of existential variables in
    $X_P \cap \var(\delta_{i-1}^S)$ obtained obtained when $\Tau'$ is
    played against $\delta^S_{i-1}$, i.e.,
    $\beta(v)=(\alpha(\Tau',\delta_{i-1}^S))_{X_P \cap \var(\delta_{i-1}^S)}$. Note that 
    $x_l$ occurs before $y_i$ in the prefix of $\qbfformula'$ and
    therefore $\beta(x_j)$ is well-defined.
    Let $S$ be the set of all indices $j$ such that
    $\delta_{i-1}^S(Y_j^a)=\delta_{i-1}^S(Y_l^a)$ and
    $\beta(x_l)=(\alpha(\Tau',\delta_{i-1}^S))(x_j)$, i.e., all
    indices of the components $(A_1,\dotsc,A_{2^{u+1}-1})$, where the
    assignment of the unique existential variable as well as the assignment
    of the first $a$ universal variables coincides with $A_l$, when
    $\Tau'$ is played against $\delta_{i-1}^S$. Then, we set
    $\delta_i^S(y_i)=\delta_{i-1}^S(y_i)$ if $|S|\leq 2^{u-a}$ and
    $\delta_i^S(y_i)=1-\delta_{i-1}^S(y_i)$, otherwise.
  \end{itemize}

  We are now ready to define $\tau_x(\delta)$ as follows. 
  \begin{itemize}
  \item if $x \in \vare(\qbfformula')$, we set
    $\tau_x(\delta)=\tau_x'(\delta^S)$
  \item Otherwise, $x=x_i$ for some $i$ with $2^{u+1}\leq i \leq r$.
    Therefore, $x$ occurs after all variables
    $x_1,\dotsc,x_{2^{u+1}-1}$ and every such variable $x_j$ has already been
    assigned the value $\tau_{x_j}'(\delta^S)$. Let $b\in \{0,1\}$ be
    the value that occurs most often among the values
    $\tau_{x_1}'(\delta^S),\dotsc,\tau_{x_{2^{u+1}-1}}'(\delta^S)$. Note that
    $b$ occurs at least $(2^{u+1}-1)/2\geq 2^u$ times. Then, we set
    $\tau_x(\delta)=b$.
  \end{itemize}
  This completes the definition of $\Tau$ and it remains to show that
  $\Tau$ is indeed a winning strategy for the existential player on
  $\qbfformula$. Towards showing this, we start by showing the following
  properties for the assignment $\delta^S$ obtained from $\delta$.

  
  Suppose not, then there is a universal play $\delta :
  V^\forall\rightarrow \{0,1\}$ such that $\alpha(\Tau,\delta)$ does
  not satisfy some clause $C$ of $\qbfformula$.

  If $C$ is not in a component $A_i$, then $C$ is also part of $\qbfformula$
  and moreover all variables in $C$ are assigned according to
  $\alpha(\Tau',\delta^S)$. Therefore, $\alpha(\Tau',\delta^S)$ would
  not satisfy $\qbfformula'$ a contradiction to our assumption that $\Tau'$
  is a winning existential strategy for $\qbfformula'$. Therefore, we can
  assume that $C$ is in a component $A_i$.

  If $i\geq 2^{u+1}$, then
  $\tau_{x_i}(\delta)$ plays the value $b$ occurring at least $2^u$
  times among
  $\tau_{x_1}'(\delta^S),\dotsc,\tau_{x_{2^{u+1}-1}}'(\delta^S)$. Therefore,
  by the definition of $\delta^S$, there is a $1\leq j \leq 2^{u+1}-1$ such that $\tau_{x_j}(\delta)=b$
  and $\delta^S(Y_j)=\delta(Y_i)$, which implies that the copy $C'$ of
  $C$ in $A_j$ is also not satisfied. But then, $C'$ is also not
  satisfied by $\alpha(\Tau',\delta^S)$ contradicting our
  assumption that $\Tau'$ is winning strategy for $\qbfformula'$.

  If, on the other hand, $i<2^{u+1}$, then by the definition of
  $\delta^S$, there is a $1\leq j\leq 2^{u+1}-1$ such that
  $\tau_{x_j}(\delta)=\tau_{x_i}(\delta)$ and
  $\delta^S(Y_j)=\delta(Y_i)$, which implies that the copy $C'$ of
  $C$ in $A_j$ is also not satisfied. But then, $C'$ is also not
  satisfied by $\alpha(\Tau',\delta^S)$ contradicting our
  assumption that $\Tau'$ is winning strategy for $\qbfformula'$.
\end{proof}

\begin{restatetheorem}[thm:delsetformfpt]
\begin{THM}
  Let $\qbfformula$ be a $\QBFCNF$ and let $D \subseteq \var(\qbfformula)$ be a
  $c$-deletion set for $\qbfformula$ into components of the
  form $\exists^{\leq 1}\forall^u$. Then, deciding whether $\qbfformula$ is
  satisfiable is fixed-parameter tractable parameterized by $|D|+c$.
\end{THM}
\end{restatetheorem}
\begin{proof}
  For every type $t \in \ctypes(\qbfformula,D)$, we use
  Lemma~\ref{lem:delset-ef-red} to reduce the number of components of
  type $t$ to at most $2^c$. Let $\qbfformula'$ be the formula obtained from
  $\qbfformula$ after an exhaustive application of Lemma~\ref{lem:delset-ef-red}. Because of
  Proposition~\ref{pro:ds-ct-nr}, there are at most
  $a=(2|D|+1)^c2^{3^{c+|D|}}$ types in $\ctypes(\qbfformula,D)$. Therefore,
  $\qbfformula'$ has at most $a2^cc+|D|=(2|D|+1)^c2^{3^{c+|D|}}2^cc+|D|$
  variables and can therefore be solved by brute-force in fpt-time
  parameterized by $|D|+c$.
\end{proof}

\subsection{Single-Variable Deletion Sets}\label{sec:1cdeletion}
We now consider the case of deletion sets consisting of a single variable~$e$.
We can assume that the variable~$e$ is existentially quantified---if it is universally quantified, we can apply Shannon expansion to obtain a QBF that decomposes into variable-disjoint subformulas that can be evaluated separately.
Furthermore, we can assume without loss of generality that $e$ is the innermost variable in the quantifier prefix, since an innermost variable that occurs only within a component can be removed by universal reduction if it is universally quantified or Shannon expansion (also called \emph{variable elimination} in this case) if it is existentially quantified.
Note that applying variable elimination to the deletion variable $e$ does not help, since it disjoins every pair of clauses from distinct components and completely obfuscates the formula's structure.

For the remainder of this section, let $\qbfformula = \mathbf{P}.F$ be a QBF with prefix $\mathbf{P} = Q_1v_1\dots Q_\ell v_\ell \exists e$ and matrix $F_1 \land \dots \land F_m$ such that $\var(F_i) \cap \var(F_j) \subseteq \{e\}$ for $1 \leq i < j \leq m$. We will show that the truth value of $\qbfformula$ can be efficiently computed from properties of the individual components. For $1 \leq i \leq m$, we let $\qbfformula^{(i)} = \mathbf{P}.F_i$ denote the \emph{component QBF} that has the same quantifier prefix as $\qbfformula$ and matrix~$F_i$.

For $0 \leq i \leq n$, let $V_i = \{v_1, \dots, v_i\}$ denote the set containing the first $i$ variables of the prefix (defining $V_0 = \emptyset$), and $V^q_i = V_i \cap \var^q(\qbfformula)$ its restriction to quantifier $q \in \{\exists, \forall\}$.
Below, we consider \emph{partial} universal (existential) strategies, by which we mean strategies that are only defined on universal (existential) variables in $V_i$, for some $0 \leq i \leq \ell$.
We say that a partial universal (respectively, existential) strategy $\Tau$ \emph{forbids} (\emph{permits}) an assignment $\sigma$ in a Boolean formula $\booleanformulaA$ if any assignment $\beta$ that is consistent with $\Tau$ and $\sigma$ falsifies (satisfies) $\booleanformulaA$.
Here, an assignment $\beta$ is \emph{consistent} with a partial strategy $\Tau$ and an assignment $\sigma$ if $\beta$ extends $\sigma$ and $\tau_{v_i}(\beta|_{V^q_{i}}) = \beta(v)$ for each variable $v_i$ defined by $\Tau$ such that $V^q_i \subseteq \dom(\beta)$, where $q = \exists$ if $\Tau$ is a partial universal strategy, and $q = \forall$ if $\Tau$ is a partial existential strategy.
A strategy forbids (permits) a \emph{set} $\Sigma$ of assignments in $\booleanformulaA$ if it forbids (permits) each $\sigma \in \Sigma$ in $\booleanformulaA$.
For brevity, we will identify the literal $e$ with the assignment $\sigma: e \mapsto 1$, and the literal $\neg e$ with the assignment $\sigma': e \mapsto 0$.

Clearly, a universal strategy for $\qbfformula$ is a winning strategy if, and only if, it forbids both $e$ and $\neg e$ in $F$.
A simple sufficient condition for unsatisfiability of $\qbfformula$ is the existence of a universal strategy for an individual component QBF~$\qbfformula^{(i)}$ that already forbids both $e$ and $\neg e$.
The following lemma states a---slightly weaker---necessary condition.
\begin{LEM}\label{lem:unsatcomponent}
  If $\qbfformula$ is unsatisfiable, there must be an index $1 \leq i \leq m$ and a universal strategy that forbids~$\sigma$ in $F_i$, for each assignment $\sigma \in \{e, \neg e\}$.
\end{LEM}
\begin{proof}
  If $\qbfformula$ is unsatisfiable then both~$\qbfformula[e]$ and $\qbfformula[\neg e]$ must be  unsatisfiable. Upon assigning~$e$, the matrix is disconnected, so for each assignment~$\sigma \in \{e, \neg e\}$, there must be a component QBF $\qbfformula^{(i)}$ such that~$\qbfformula^{(i)}[\sigma]$ is unsatisfiable. The corresponding universal winning strategy forbids the assignment $\sigma$ in ~$F_i$.
\end{proof}
If there are \emph{distinct} component QBFs that forbid $e$ and $\neg e$, respectively, then the corresponding strategies can be unified into a \emph{single} strategy that achieves both.
This is because the component QBFs do not share variables apart from the deletion variable~$e$. More generally, any property of a strategy that is determined by the clauses in a component (or a set of components) is preserved when changing only strategy functions for variables that do not occur in this component (or set of components).
\begin{DEF}
  A class $\mathcal{C}$ of (partial) strategies defined on variables $V_i$ is \emph{indifferent towards variable $w$} if, whenever $\Tau \in \mathcal{C}$ and $\Tau'$ is a partial strategy that differs from $\Tau$ only on $w \in V_i$, then $\Tau' \in \mathcal{C}$.
  We say that $\mathcal{C}$ is indifferent towards a set $W$ of variables if it is indifferent towards each variable $w \in W$.
\end{DEF}

\begin{LEM}\label{lem:combinestrategies}
  Let $\mathcal{C}_1$ be a class of (partial) strategies defined on variables $V_i$ that is indifferent towards $W_1$ and let $\mathcal{C}_2$ be a class of partial strategies defined on $V_i$ that is indifferent towards $W_2$, such that $V_i \subseteq W_1 \cup W_2$. Then the intersection of $\mathcal{C}_1$ and $\mathcal{C}_2$ is non-empty.
\end{LEM}
\begin{proof}
Let~$\Tau_1 \in \mathcal{C}_1$ and $\Tau_2 \in \mathcal{C}_2$. Consider the strategy $\Tau = \{ f_v \in \Tau_1 \:|\: v \in V_i \setminus W_1\} \cup \{ f_w \in \Tau_2 \:|\: w \in V_i \cap W_1\}$.
Since it differs from $\Tau_1$ only on variables in $W_1$, and $\Tau_1$ is indifferent to $W_1$, we have $\Tau \in \mathcal{C}_1$. Similarly, $\Tau$ differs from $\Tau_2$ only on variables in $V_i \setminus W_1 \subseteq W_2$, so $\Tau \in \mathcal{C}_2$ holds as well.
\end{proof}
This result can be applied to strategies that forbid or permit certain assignments of variable $e$.
\begin{LEM}\label{lem:indifferent}
  Let $\booleanformulaA$ be a Boolean formula with $\var(\booleanformulaA) \subseteq \var(\qbfformula)$ and let $W = \var(\qbfformula) \setminus \var(\booleanformulaA)$.
  The following classes of partial strategies are indifferent to $W$:
  \begin{enumerate}
  \item The class of partial universal strategies that forbid  $\Sigma \subseteq \{e, \neg e\}$ in $\booleanformulaA$.
  \item The class of partial existential strategies that permit  $\Sigma \subseteq \{e, \neg e\}$ in~$\booleanformulaA$.
  \end{enumerate}
\end{LEM}
\begin{proof}
  Whether a partial strategy~$\Tau$ forbids or permits an assignment in $\booleanformulaA$ only depends on whether plays~$\beta$ that are consistent with~$\Tau$ satisfy or falsify $\booleanformulaA$, and this is unaffected by changing functions for variables outside $\booleanformulaA$.
\end{proof}
\begin{LEM}\label{lem:winninguniversal}
  Let $\booleanformulaA$ and $\booleanformulaB$ be Boolean formulas such that $\var(\booleanformulaA) \cap \var(\booleanformulaB) \subseteq \{e\}$ and $\var(\booleanformulaA) \cup \var(\booleanformulaB) \subseteq \qbfformula$.
 Let $\Tau$ be a universal strategy that forbids $e$ in $\booleanformulaA$, and $\Tau'$ a universal strategy that forbids $\neg e$ in $\booleanformulaB$. Then $\mathbf{P}.\booleanformulaA \land \booleanformulaB$ has a universal winning strategy.
\end{LEM}
\begin{proof}
  By Lemma~\ref{lem:indifferent}, the partial strategy $\Tau$ is indifferent to $\var(\qbfformula) \setminus \var(\booleanformulaA)$ and $\Tau'$ is indifferent to $\var(\qbfformula) \setminus \var(\booleanformulaB)$.
  Because $(\var(\qbfformula) \setminus \var(\booleanformulaA)) \cup (\var(\qbfformula) \setminus \var(\booleanformulaB)) = \var(\qbfformula)$, there must be a universal strategy $\Tau''$ that forbids $e$ in $\booleanformulaA$ and $\neg e$ in $\booleanformulaB$ by Lemma~\ref{lem:combinestrategies}.
  Thus $\Tau''$ forbids both $e$ and $\neg e$ in $\booleanformulaA \land \booleanformulaB$, and it is a winning universal strategy.
\end{proof}
We can conclude that, whenever there are distinct component QBFs and universal strategies that forbid $e$ and $\neg e$, there is a universal winning strategy.

The next lemma shows that any component $F_i$ such that there exists a (partial) existential strategy that permits both $e$ and $\neg e$ in $F_i$ can be removed from the matrix without changing the truth value.
\begin{LEM}\label{lem:permitsboth}
  Let $\booleanformulaA$ and $\booleanformulaB$ be Boolean formulas such that $\var(\booleanformulaA) \cup \var(\booleanformulaB) \subseteq \var(\qbfformula)$ and $\var(\booleanformulaA) \cap \var(\booleanformulaB) \subseteq \{e\}$. Let $1 \leq i \leq \ell$ and let $\Tau$ be a partial existential strategy defined on $V^\exists_i$ for such that $\Tau$ permits $e$ and $\neg e$ in $\booleanformulaB$.
  Then $\mathbf{P}.\booleanformulaA$ and $\mathbf{P}.\booleanformulaA \land \booleanformulaB$ have the same truth value.
\end{LEM}
\begin{proof}
  If $\mathbf{P}.\booleanformulaA$ is false the result is immediate. Otherwise, there is an existential winning strategy $\Tau'$ for $\mathbf{P}.\booleanformulaA$.
  Let $f_e \in \Tau'$ be the strategy function for variable~$e$, and let $\Tau''$ be the partial existential strategy obtained by removing $f_e$ from $\Tau'$.
  By construction, this partial strategy can be turned into an existential winning strategy for $\mathbf{P}.\booleanformulaA$ by adding the function $f_e$, and the corresponding class of partial existential strategies is indifferent to $V_i^\exists \setminus \var(\booleanformulaA)$.
  By Lemma~\ref{lem:indifferent}, the class of partial existential strategies that permit $e$ and $\neg e$ in $\booleanformulaB$ is indifferent to $V_i^\exists \setminus \var(\booleanformulaB)$.
  Moreover, $(V_i^{\exists} \setminus \var(\booleanformulaA)) \cup (V_i^{\exists} \setminus \var(\booleanformulaB)) = V_i^{\exists}$, so we can apply Lemma~\ref{lem:combinestrategies} and conclude that there is a partial existential strategy defined on $V_i^\exists$ that permits $e$ and $\neg e$ in $\booleanformulaB$ and that can be turned into an existential winning strategy for $\mathbf{P}.\booleanformulaA$ by extending it with $f_e$.
  By adding $f_e$ to this partial strategy, we obtain an existential winning strategy for $\mathbf{P}.\booleanformulaA \land \booleanformulaB$.
\end{proof}
Let us briefly take stock.
We can now assume there are no components for which there is a partial universal strategy that forbids both assignments of $e$ (in this case $\qbfformula$ is false), or a partial existential strategy that permits both assignments (these do not affect the truth value).
Furthermore, we may assume there do not exist distinct components for which universal can forbid $e$ and $\neg e$, respectively ($\qbfformula$ is false in this case).
By Lemma~\ref{lem:unsatcomponent}, for $\qbfformula$ to be false, there must be a component $F_i$ such that there is a universal strategy that forbids $e$ in $F_i$, and a universal strategy that forbids~$\neg e$ in $F_i$.
However, our assumptions tell us that there is no strategy that forbids both assignments in $F_i$.
Any remaining components $F_j$ for $j \neq i$ must have a dual property: there exists a partial existential strategy that permits $e$ in $F_j$, and a partial existential strategy that permits $\neg e$ in $F_j$, but no strategy that permits both assignments.

So for each remaining component, at some point in any play, universal (respectively, existential) must commit to forbidding (permitting) either $e$ or $\neg e$.
Until then, there is a strategy that keeps both options alive.
\begin{DEF}
  Let $\booleanformulaA$ be a Boolean formula with $\var(\booleanformulaA) \subseteq \var(\qbfformula)$, let $0 \leq i \leq \ell$, and let $\Tau$ be a partial universal (respectively, existential) strategy defined on $V_i^\forall$ ($V_i^\exists$) satisfying the following property: for any assignment $\beta: V_ i \rightarrow \{0, 1\}$ that is consistent with $\Tau$, universal (existential) has a partial strategy defined on $V_{\ell}^\forall \setminus V_i$ ($V_{\ell}^{\exists} \setminus V_i$) that forbids (permits) $e$ in $\booleanformulaA[\beta]$, and a strategy that forbids (permits) $\neg e$ in $\booleanformulaA[\beta]$.
We say that \emph{$\Tau$ leaves a choice at index~$i$ (in $\booleanformulaA$)}.
\end{DEF}
There is a dual notion of a partial strategy ensuring that, at a certain point, universal can forbid \emph{some} assignment (respectively, existential can permit some assignment).
\begin{DEF}
  Let $\booleanformulaA$ be a Boolean formula with $\var(\booleanformulaA) \subseteq \var(\qbfformula)$, let $0 \leq i \leq \ell$, and let $\Tau$ be a partial universal (respectively, existential) strategy defined on $V_i^\forall$ ($V_i^\exists$) satisfying the following property: for any play $\beta: V_i \rightarrow \{0, 1\}$ that is consistent with $\Tau$, there is an assignment $\sigma_\beta \in \{e, \neg e\}$ such that universal (existential) has a partial strategy defined on $V_{\ell}^{\forall} \setminus V_i$ ($V_{\ell}^{\exists} \setminus V_i$) that forbids (permits) $\sigma_\beta$ in $\booleanformulaA[\beta]$. In this case, wee say that \emph{$\Tau$ empowers at index~$i$ (in $\booleanformulaA$)}.
\end{DEF}
\begin{LEM}\label{lem:leaveschoice}
  Let $\booleanformulaA$ be a Boolean formula with $\var(\booleanformulaA) \subseteq \var(\qbfformula)$, and let $0 \leq i \leq \ell$.
  The class of partial universal (respectively, existential) strategies that leave a choice at $i$, and the class of partial universal (respectively, existential) strategies that empower at $i$, are indifferent to $V_i \setminus \var(\booleanformulaA)$.
\end{LEM}
\begin{proof}
  Whether a partial strategy~$\Tau$ leaves a choice or empowers at~$i$ depends only on $\booleanformulaA[\beta]$ for each play $\beta$ consistent with~$\Tau$, and modifying $\beta$ on variables in $V_i$ that do not occur in $\booleanformulaA$ does not change this formula.
\end{proof}
\begin{LEM}\label{lem:combineempowerandchoice}
  Let $\booleanformulaA$ and $\booleanformulaB$ be a Boolean formulas with $\var(\booleanformulaA) \cap \var(\booleanformulaB) \subseteq \{e\}$ and $\var(\booleanformulaA) \cup \var(\booleanformulaB) \subseteq \var(\qbfformula)$.
  If there is a partial universal (respectively, existential) strategy $\Tau$ that empowers at $i$ in $\booleanformulaA$, and a partial universal (existential) strategy $\Tau'$ that leaves a choice at $i$ in $\booleanformulaB$, then there is a universal (existential) winning strategy for $\mathbf{P}.\booleanformulaA \land \booleanformulaB$.
\end{LEM}
\begin{proof}
  By Lemma~\ref{lem:leaveschoice} and Lemma~\ref{lem:combinestrategies}, there is a partial universal (existential) strategy~$\Tau$ for $V^\forall_i$ ($V^\exists_i$) such that, for any play $\beta: V_i \rightarrow \{0, 1\}$ consistent with $\Tau$, universal (existential) has a partial strategy that forbids (permits) $\sigma_\beta \in \{e, \neg e\}$ in $\booleanformulaA[\beta]$, and a partial strategy that forbids (permits) $\neg \sigma_\beta$ (respectively, $\sigma_\beta$) in $\booleanformulaB[\beta]$.
  By Lemma~\ref{lem:indifferent} and Lemma~\ref{lem:combinestrategies}, this implies that there is a partial universal (existential) strategy $\Tau_\beta$ that forbids (permits) $\sigma_\beta$ in $\booleanformulaA$ and forbids (permits) $\neg \sigma_\beta$ ($\sigma_\beta$) in $\booleanformulaB$.
  That is, any play consistent with $\beta$ and $\Tau_\beta$ falsifies $\booleanformulaA \land \booleanformulaB$ (for the case of existential strategies, any play consistent with $\beta$, $\Tau_\beta$, and $\sigma_\beta$ satisfies $\booleanformulaA \land \booleanformulaB$).
By combining $\Tau$ with $\Tau_{\beta}$ (and, for existential strategies, the function that assigns $e$ according to $\sigma_\beta$), one obtains a universal (existential) winning strategy for $\mathbf{P}.\booleanformulaA \land \booleanformulaB$.
\end{proof}
The notions of ``empowering'' and ``leaving a choice'' are dual in the following sense: a player has a strategy that is empowering at $i$ if, and only if, the other player does not have a strategy that leaves a choice at $i$.
\begin{LEM}\label{lem:choiceduality}
  Let $\booleanformulaA$ be a Boolean formula such that $\var(\booleanformulaA) \subseteq \var(\qbfformula)$ and let $0 \leq i \leq \ell$.
There is a partial universal (respectively, existential) strategy that leaves a choice at index~$i$ if, and only if, there is no partial existential (universal) strategy that empowers at~$i$.
\end{LEM}
\begin{proof}
  Consider the assignment game between universal and existential where universal wins a play~$\beta: V_i \rightarrow \{0, 1\}$ if there is an assignment $\sigma_\beta \in \{e, \neg e\}$ and a partial universal strategy that forbids $\sigma_\beta$ in $\booleanformulaA[\beta]$.
  This game is determined, so either there is a partial universal strategy that empowers at~$i$ (if universal has a winning strategy), or an existential strategy~$\Tau$ such that, for any $\beta$ consistent with $\Tau$, there is no $\sigma \in \{e, \neg e\}$ such that universal has a strategy that forbids $\sigma$ in $\booleanformulaA[\beta]$ (if existential wins).
  By Lemma~\ref{lem:ex-uni-str}, for any assignment~$\beta: V_i \rightarrow \{0, 1\}$, there is either a universal strategy that forbids $\sigma \in \{e, \neg e\}$ in $\booleanformulaA[\beta]$, or an existential strategy that forces (and thus permits) $\sigma$ in $\booleanformulaA[\beta]$.
  That is, there is a partial existential strategy that permits $e$ in $\booleanformulaA[\beta]$, and a partial existential strategy that permits $\neg e$ in $\booleanformulaA[\beta]$, which is to say that $\Tau$ leaves a choice at $i$.
 The proof of the duality between partial existential strategies that leave a choice and partial universal strategies that empower is similar.
 \end{proof}
  Each remaining component can be characterized by the highest index~$i$ such that there is a universal strategy that leaves a choice, and the lowest index~$i$ such that there is a universal strategy that empowers.
 \begin{DEF}
   Let $\booleanformulaA$ be a Boolean formula with $\var(\booleanformulaA) \subseteq \var(\qbfformula)$.
   If there is a partial universal strategy that leaves a choice at some index~$i$ with $0 \leq i \leq \ell$ in $\booleanformulaA$, we let $LC(\booleanformulaA)$ be the largest such index~$i$.
   Similarly, if there is a partial universal strategy that empowers at index~$i$ with $0 \leq i \leq \ell$, we let $EE(\booleanformulaA)$ denote the smallest such index $i$.
 \end{DEF}
 These indices allow us to state a necessary and sufficient condition for the existence of a universal winning strategy in the remaining components.
\begin{LEM}\label{lem:indexcomparison}
  Let $\booleanformulaA$ and $\booleanformulaB$ be Boolean formulas with $\var(\booleanformulaA) \cap \var(\booleanformulaB) \subseteq \{e\}$ and $\var(\booleanformulaA) \cup \var(\booleanformulaB) \subseteq \var(\qbfformula)$, such that
  \begin{enumerate}
  \item there is a partial universal strategy that leaves a choice at some index $0 \leq i \leq \ell$ in $\booleanformulaA$, and
  \item there is a partial universal strategy that empowers at index some index $0 \leq j \leq \ell$ in $\booleanformulaB$.
  \end{enumerate}
  Then there is a winning universal strategy for $\mathbf{P}.\booleanformulaA \land \booleanformulaB$ if, and only if, $EE(\booleanformulaB) \leq LC(\booleanformulaA)$.
\end{LEM}
\begin{proof}
  Suppose $EE(\booleanformulaB) \leq LC(\booleanformulaA)$.
  By definition, there is a partial universal strategy that leaves a choice at index $LC(\booleanformulaA)$ in $\booleanformulaA$.
  It follows that there are partial universal strategies that leave a choice at every index $0 \leq i < LC(\booleanformulaA)$, in particular at index~$EE(\booleanformulaB)$.
By Lemma~\ref{lem:combineempowerandchoice}, there is a universal winning strategy for $\mathbf{P}.\booleanformulaA \land \booleanformulaB$.
Now suppose $LC(\booleanformulaA) < EE(\booleanformulaB)$.
We claim that in this case, we must in fact have $LC(\booleanformulaA) + 1 < EE(\booleanformulaB)$. To see this, note that the variable~$v$ at index $EE(\booleanformulaB)$ must occur in $\booleanformulaB$.
Otherwise, the partial universal strategy that empowers at index $EE(\booleanformulaB)$ in $\booleanformulaB$ would already empower at index $EE(\booleanformulaB) - 1$, since the assignment of~$v$ would be irrelevant.
But then $v \notin \var(\booleanformulaA)$, and if $v$ were the variable at index $LC(\booleanformulaA)$, we could extend the partial universal strategy that leaves a choice at index~$LC(\booleanformulaA)$ with an arbitrary function for variable $v$ so that it still leaves a choice at index~$LC(\booleanformulaA) + 1$, a contradiction. Thus the claim is proved.
By Lemma~\ref{lem:choiceduality}, there is a partial existential strategy that empowers at index $LC(\booleanformulaA) + 1$ in $\booleanformulaA$, and a partial existential strategy that leaves a choice at index $EE(\booleanformulaB) - 1$ in $\booleanformulaB$, which means there has to be a partial existential strategy that leaves a choice at every index $i \leq EE(\booleanformulaB) - 1$.
By the claim proved above, this includes the index $LC(\booleanformulaA) + 1$, and it follows from Lemma~\ref{lem:combineempowerandchoice} that there is an existential winning strategy for $\mathbf{P}.\booleanformulaA \land \booleanformulaB$.
\end{proof}
There is only a single component $F_i$ with a universal strategy that leaves a choice, for which the index $LC(F_i)$ can be computed using brute force.
However, one cannot simply use brute force to compute the index $EE(\bigwedge_{j \neq i} F_j)$, since the number of components in this subformula is not bounded by the constant $c$.
The next lemma shows that the index can be computed from the indices of individual components.
\begin{LEM}\label{lem:indexminimum}
  Let $\booleanformulaA$ and $\booleanformulaB$ be Boolean formulas with $\var(\booleanformulaA) \cap \var(\booleanformulaB) \subseteq \{e\}$ and $\var(\booleanformulaA) \cup \var(\booleanformulaB) \subseteq \var(\qbfformula)$ such that $EE(\booleanformulaA)$ and $EE(\booleanformulaB)$ are defined. Then \[EE(\booleanformulaA \land \booleanformulaB) = \min(EE(\booleanformulaA), EE(\booleanformulaB)).\]
\end{LEM}
\begin{proof}
  A partial universal strategy $\Tau$ that empowers at index $i$ in $\booleanformulaA$ or $\booleanformulaB$ also empowers at index $i$ in $\booleanformulaA \land \booleanformulaB$, so $EE(\booleanformulaA \land \booleanformulaB) \leq \min(EE(\booleanformulaA), EE(\booleanformulaB))$.
  For the other inequality, assume that $\min(EE(\booleanformulaA), EE(\booleanformulaB)) > 0$, and let $0 \leq i < \min(EE(\booleanformulaA), EE(\booleanformulaB))$.
  By Lemma~\ref{lem:choiceduality}, there is a partial existential strategy $\Tau$ that leaves a choice at index $i$ in $\booleanformulaA$, and a partial existential strategy $\Tau'$ that leaves a choice at index $i$ in $\booleanformulaB$.
  Accordingly, by Lemma~\ref{lem:leaveschoice} and Lemma~\ref{lem:combinestrategies}, there is a partial existential strategy $\Tau''$ that leaves a choice at index $i$ in $\booleanformulaA \land \booleanformulaB$.
  Applying Lemma~\ref{lem:choiceduality} again, we conclude that there cannot be a partial universal strategy that empowers at index~$i$ in $\booleanformulaA \land \booleanformulaB$, so $\min(EE(\booleanformulaA), EE(\booleanformulaB)) \leq EE(\booleanformulaA \land \booleanformulaB)$.
\end{proof}
Since the properties of strategies with respect to a component $F_i$ we are interested in are all indifferent to variables that do not occur in $F_i$, we can determine whether there exists a strategy with any of these properties by a brute-force enumeration of strategies restricted to components. \begin{LEM}\label{lem:indifferentpropertytractable}
  Let $\booleanformulaA$ be a Boolean formula with $\var(\booleanformulaA) \subseteq \var(\qbfformula)$.
  The following problems can be decided in time $f(\Card{\var(\booleanformulaA)})$ for some function $f$:
  \begin{enumerate}
  \item Is there a partial universal (respectively, existential) strategy that forbids (permits) $\Sigma \subseteq \{e, \neg e\}$ in $\booleanformulaA$.
  \item Is there a partial universal (existential) strategy that empowers at index $i$ with $0 \leq i \leq \ell$ in $\booleanformulaA$?
  \item Is there a partial universal (existential) strategy that leaves a choice at index $i$ with $0 \leq i \leq \ell$ in $\booleanformulaA$?
  \end{enumerate}
\end{LEM}
\begin{restatetheorem}[thm:1cdeletionset]
\begin{THM}
Let $\qbfformula$ be a $\QBFCNF$ with a $c$-deletion set of size $1$. Deciding whether $\qbfformula$ is true is fixed-parameter tractable parameterized by $c$.
\end{THM}
\end{restatetheorem}
\begin{proof}
  Let $e$ denote the variable in the deletion set. We assume without loss of generality that it is existentially quantified and innermost in the prefix of $\qbfformula$.
  Algorithm $A$ first determines whether there is a component for which there is a universal winning strategy and outputs $0$ if there is.
  Otherwise, it removes any component for which there is a partial existential strategy that permits $e$ and $\neg e$, which preserves the truth value by Lemma~\ref{lem:permitsboth}. If no more than one component remains, the algorithm outputs $1$.
  Otherwise, the algorithm determines whether there are components and universal strategies that forbid $e$ and $\neg e$ in these (not necessarily distinct) components.
  If that is not the case, the algorithm outputs~$1$, which is correct by Lemma~\ref{lem:unsatcomponent}.
  If there are such components, and they are distinct, $A$ outputs~$0$, which is correct by Lemma~\ref{lem:winninguniversal}.
  Let $F_i$ denote the component such that there are universal strategies that forbid $e$ and $\neg e$, respectively.
  The algorithm computes $LC(F_i)$ and $\min_{j\neq i}\{EE(F_j)\} = EE(\bigwedge_{j \neq i} F_j)$ for the remaining components $F_j$.
  By Lemma~\ref{lem:indifferentpropertytractable} and Lemma~\ref{lem:indexminimum}, this can be done in time $f(c) p(\CCard{\qbfformula})$ for suitable $f$ and $p$.
  Finally, $A$ outputs $0$ if $EE(F_i) \leq  \min_{j\neq i}\{EE(F_j)\}$, and $1$ otherwise, which is correct by Lemma~\ref{lem:indexcomparison}.
\end{proof}

\subsection{An Algorithm for Formulas with Many Components of Each Type}


Here, we provide a formal proof for
Theorem~\ref{thm:qbfcnfunifptex}. We start with introducing the
following notions.
Let $\qbfformula=Q_{1} v_1 Q_2v_2 \cdots Q_n v_n.F$ be a
$\QBFCNF$ and let $D \subseteq \vare(\qbfformula)$ be a
$c$-deletion set for $\qbfformula$.
Consider an assignment $\beta : D \rightarrow \{0,1\}$ for the
variables in $D$. We say that an existential strategy $\Tau$ \emph{forces}
$\beta$ if for every universal play $\delta: \varu(\qbfformula)\rightarrow
\{0,1\}$, it holds that $\alpha(\Tau,\delta)$ extends $\beta$.
Moreover, we say that a universal strategy $\Lambda$ \emph{forbids}
$\beta$ if $\alpha(\Lambda,\delta)$ does not satisfy $F$ for every existential play $\delta$ that extends $\beta$.

The main ideas behind the proof of
Theorem~\ref{thm:qbfcnfunifptex} are as follows. First, we show
that for every component $C$ of $G_\qbfformula-D$ and every assignment $\beta :D
\rightarrow \{0,1\}$ of the (existential) variables in the deletion
set $D$, the following holds for the formula $\qbfind{\qbfformula}{D\cup C}$: either the existential player has a winning
strategy  that always plays the assignment $\beta$ or the universal player has a strategy $\Lambda$ that forbids
$\beta$, i.e., every existential play against $\Lambda$ that plays
$\beta$ loses against $\Lambda$. To decide whether $\qbfformula$ is
satisfiable, we first compute the set $A(t)$ for every component
type $t$. Here $A(t)$ is the set of all assignments
$\beta :D \rightarrow \{0,1\}$ such that the existential player has a
winning strategy on the formula $\qbfind{\qbfformula}{D\cup C_t}$ that
forces $\beta$, where $C_t$ is any
component of type $t$. Note that the sets $A(t)$ can be computed 
efficiently since both the number of component types as well as the
number of variables of the formulas $\qbfind{\qbfformula}{D\cup C_t}$ is bounded by a
function of the parameters. We then show that $\qbfformula$ is satisfiable if
and only if $\bigcap_{t \in \ctypes(\qbfformula,D)}A(t)\neq
\emptyset$. This is because if there is a $\beta \in \bigcap_{t \in
  \ctypes(\qbfformula,D)}A(t)$, then the existential player can win by playing its
winning strategy that forces $\beta$ in every
component. Moreover, if there is no such $\beta$, then the universal
player has sufficiently many copies of each component type to forbid
every possible assignment $\beta$ of the deletion set variables for
the existential player.

So, the following lemma is crucial for our algorithm.
\begin{LEM}
\label{lem:ex-uni-str}
  There is an existential winning straggly that forces $\beta$ if and
  only if there is no universal strategy that forbids $\beta$.
\end{LEM}
\begin{proof}
  To ease notation, we set $V=\var(\qbfformula)$.
  Towards showing the forward direction, let $\Tau=(\tau_v: \{0,1\}^{V_{<v}^{\forall}}\rightarrow
  \{0,1\})_{v \in V^\exists}$ be an existential winning strategy
  that forces $\beta$ and suppose for a contradiction
  that there is a universal strategy $\Lambda=(\lambda_v:
  \{0,1\}^{V_{<v}^{\exists}}\rightarrow \{0,1\})_{v\in V^\forall}$
  that forbids $\beta$. Let $\delta : V \rightarrow
  \{0,1\}$ be the assignment resulting from $\Tau$ playing
  against $\Lambda$, i.e., $\delta=(\alpha(\Tau,\Lambda))$.
  Then, $\alpha(\Lambda,\delta^\exists)$ does not satisfy $F$, because $\delta^\exists$
  extends $\beta$ and $\Gamma$ forbids $\beta$. However,
  $\alpha(\Lambda,\delta^\exists)=\alpha(\Tau,\delta^{\forall})$ also satisfies $F$ because $\Tau$ is a
  winning strategy for the existential player, a contradiction.

  Towards showing the reverse direction, we will construct an
  existential winning strategy $\Tau=(\tau_v: \{0,1\}^{V_{<v}^{\forall}}\rightarrow
  \{0,1\})_{v \in V^\exists}$ that forces $\beta$ as follows.
  We build $\Tau$ iteratively from left to right though the quantifier
  prefix. Namely, let $v \in V^\exists$ and suppose that we have
  already defined the partial existential strategy $\Tau_v=(\tau_u: \{0,1\}^{V_{<u}^{\forall}}\rightarrow
  \{0,1\})_{u \in V_{<v}^\exists}$, i.e., $\Tau_v$ is the partial
  existential strategy defined for all existential variables occurring
  before $v$ in the prefix. Let
  $U$ be the set of all universal strategies for $\qbfformula$. Moreover, for
  an assignment $\delta : V_{<v}^{\forall}\rightarrow \{0,1\}$, let $U_{\delta}$ be the set of
  all universal strategies that play $\delta$ against the partially
  defined strategy $\Tau_v$, i.e., $\alpha(\Tau_v,\Lambda)=\delta$ for
  every $\Lambda \in U_{\delta}$. We will maintain the following invariant:

  \begin{itemize}
  \item[(*)] For every $\delta : V_{<v}^{\forall}\rightarrow \{0,1\}$
    and every $\Lambda \in U_{\delta}$, there is a play
    $\delta_{\exists} : V^{\exists}_{\geq v}\rightarrow \{0,1\}$ that extends
    $\beta$ such that $\alpha(\Lambda,\delta'\cup \delta_{\exists})$,
    where $\delta'=\alpha(\Tau_v,\Lambda)^{\exists}$ satisfies $F$.
  \end{itemize}



  Because we assume that there is no universal strategy that forbids
  $\beta$, we obtain that $U$ satisfies (*). Moreover, if the
  invariant holds at the end, then $\Tau=\Tau_n$ is a winning strategy
  for the existential player that forces $\beta$.

  Now, let $v \in V^\exists$ and suppose that we have
  already defined $\Tau_v$, i.e., $\tau_{u}$ for all $u \in V_{<v}^\exists$.
  Let $\delta : V_{<v}^{\forall}\rightarrow \{0,1\}$.
  Note that if $v \in D$, then setting $\tau_v(\delta)=\beta(v)$
  provides an extension of $\Tau_v$ that satisfies (*) by the
  induction hypothesis. So suppose that $v \notin D$.
  Note that if there is an
  assignment $\delta_v : v \rightarrow \{0,1\}$ such that for every
  $\Lambda \in U_{\delta}$, there is a play
  $\delta_{\exists} : V^{\exists}_{> v}\rightarrow \{0,1\}$ that extends
  $\beta$ such that $\alpha(\Lambda,\delta'\cup \delta_v\cup
  \delta_{\exists})$, where $\delta'=(\alpha(\Tau_v,\Lambda))^{\exists}$,
  does satisfy $F$, then we can define the response of $\tau_v$
  for $\delta$ as $\delta_v$. We show next that such an assignment
  $\delta_v : v \rightarrow \{0,1\}$ must always exists, which
  completes the proof.

  Suppose for a contradiction that this is not the case, then for
  every assignment $\delta_v : v \rightarrow \{0,1\}$ there is some
  $\Lambda^{\delta_v} \in U_{\delta}$, such that for every play $\delta_{\exists} : V^{\exists}_{> v}\rightarrow \{0,1\}$ that extends
  $\beta$, the assignment $\alpha(\Lambda^{\delta_v},\delta'\cup  \delta_v\cup
  \delta_{\exists})$, where $\delta'=\alpha(\Tau_v,\Lambda^{\delta_v})^{\exists}$,
  does not satisfy $F$.

  Consider now the following composition
  $\Lambda$ of the strategies $\Lambda^{\delta_v}$, which is defined
  as follows:
  \begin{itemize}
  \item For every $u \in V_{<v}^\forall$, we set $\lambda_u=\lambda_u^{\tau}$, where $\tau$ is the assignment
    that sets $v$ to $0$.
  \item For every $u\in V_{>v}^\forall$, we set $\lambda_u(\delta'\cup
    \delta_v\cup \delta'')=\lambda_u^{\delta_v}(\delta'\cup\delta_v\cup\delta'')$
    for every $\delta' : V_{<v}^\exists\rightarrow \{0,1\}$, every $\delta_v :
    v\rightarrow \{0,1\}$, and every $\delta'' : V_{>v}^\exists\cap
    V_{<u}^\exists\rightarrow \{0,1\}$.
  \end{itemize}
  Then, $\Lambda \in
  U_{\delta}$, however, $\Lambda$ does not satisfy (*), a contradiction
  to our induction hypothesis.
\end{proof}

\begin{restatetheorem}[thm:qbfcnfunifptex]
\begin{THM}
  Let $\qbfformula$ be a $\QBFCNF$ and let $D \subseteq \vare(\qbfformula)$ be a
  universally complete $c$-deletion set for $\qbfformula$. Then deciding whether $\qbfformula$ is satisfiable is fixed-parameter
  tractable parameterized by $|D|+c$.
\end{THM}
\end{restatetheorem}
\begin{proof}
  For every $t \in \ctypes(\qbfformula,D)$, let $A(t)$ be the set of all assignments $\beta : D \rightarrow
  \{0,1\}$ such that the existential player has a winning strategy for
  $\qbfind{\qbfformula}{D\cup C_t}$ that forces $\beta$ for some $C_t \in \CbyT(\qbfformula,D,t)$. We claim that $\qbfformula$ is satisfiable if and only
  if $\bigcap_{t \in \ctypes(\qbfformula)}A(t)\neq \emptyset$. Note that this
  completes the proof of the theorem, because deciding whether
  $\bigcap_{t \in \ctypes(\qbfformula)}A(t)\neq \emptyset$ can be achieved in
  fpt-time parameterized by $|D|+c$ as follows. First, we use
  Proposition~\ref{pro:delset-comptypes} to compute $\ctypes(\qbfformula,D)$ and
  $\CbyT(\qbfformula,D,t)$ for every $t \in \ctypes(\qbfformula,D)$ in
  polynomial time. Then, for every of the at most $\CCard{\qbfformula}$ types $t
  \in \ctypes(\qbfformula,D)$ and every of the at most $2^{|D|}$ assignments $\beta : D \rightarrow
  \{0,1\}$, we decide whether the existential player has a winning strategy for
  $\qbfind{\qbfformula}{D\cup C_t}$ that forces $\beta$ for some $C_t \in \CbyT(\qbfformula,D,t)$ in time at most
  $\bigO(2^{|D|+c})$, which can be achieved by brute-force because
  $\qbfind{\qbfformula}{D\cup C_t}$ has at most $|D|+c$ variables.

  Towards showing the forward direction, we show the
  contraposition. That is, assuming that $\bigcap_{t \in
    \ctypes(\qbfformula)}A(t)= \emptyset$, we show that $\qbfformula$ is not satisfiable.
  Because of Lemma~\ref{lem:ex-uni-str}, it holds that for every $t
  \in \ctypes(\qbfformula)$ and every $\beta : D \rightarrow \{0,1\}$, if $\beta
  \notin A(t)$, then there is a universal strategy $\Gamma^{t,\beta}$
  for $\qbfind{\qbfformula}{D \cup C_t}$ that forbids $\beta$, where $C_t \in \CbyT(\qbfformula,D,t)$. Because $\qbfformula$ is
  universally complete, for every such $t$ and $\beta$ we can choose a
  private component $C^{t,\beta}$ of type $t$. Now let $\Gamma$ be the
  universal strategy that plays $\Gamma^{t,\beta}$ in the component
  $C^{t,\beta}$. Then, because $\bigcap_{t \in
    \ctypes(\qbfformula)}A(t)= \emptyset$, we obtain that for every $\beta : D
  \rightarrow \{0,1\}$, there is at least one component $C^{t,\beta}$
  such that $\Gamma$ plays $\Gamma^{t,\beta}$ in $\qbfind{\qbfformula}{D\cup
  C^{t,\beta}}$. In other words $\Gamma$ forbids every possible $\beta
  : D \rightarrow \{0,1\}$ and therefore $\Gamma$ is a winning
  strategy for the universal player on $\qbfformula$, showing that $\qbfformula$
  is not satisfiable.

  Towards showing the reverse direction, let $\beta \in \bigcap_{t \in
    \ctypes(\qbfformula)}A(t)$. Then, for every component $C \in \CbyT(\qbfformula,D,t)$, there
  is a winning strategy $\Tau^{C,\beta}$ for the existential player in
  $\qbfind{\qbfformula}{D\cup C}$ that forces $\beta$. It is straightforward to verify
  that the composition $\Tau^{\beta}$ of all $\Tau^{C,\beta}$'s is an
  existential winning strategy for $\qbfformula$.
\end{proof}

\end{document}